\documentclass[unsortedaddress,superscriptaddress,pra,notitlepage]{revtex4}

\usepackage{amsfonts}
\usepackage[latin2]{inputenc}
\usepackage{t1enc}
\usepackage{amssymb}
\usepackage{amsmath}
\usepackage[mathscr]{eucal}
\usepackage{amsthm}
\usepackage{array}
\usepackage{color}
\usepackage[hidelinks]{hyperref}
\usepackage{pdfpages}

\usepackage{graphicx}
\usepackage{xcolor}
\usepackage{tikz}

\usetikzlibrary{calc,decorations.pathreplacing}
\usetikzlibrary{arrows,shapes}

\theoremstyle{definition} 
\newtheorem{defin}{Definition}[section]
\newtheorem{definition}{Definition}[section]
\newtheorem{thm}[defin]{Theorem}

\newtheorem{rem}[defin]{Remark}
\newtheorem{ex}[defin]{Example}
\newtheorem{cor}[defin]{Corollary}

\newtheorem{lemma}[defin]{Lemma}
\newtheorem{prop}[defin]{Proposition}

\renewcommand{\theHWexercise}{\theHWexercise$^*$}

\def\vfi{\varphi}
\def\theta{\vartheta}
\def\hil{{\mathcal H}}
\def\kil{{\mathcal K}}

\def\B{{\mathcal B}}

\def\I{{\mathcal I}}
\def\J{{\mathcal J}}

\def\L{{\mathcal L}}

\def\S{{\mathcal S}}

\def\half{\frac{1}{2}}
\def\iff{\Longleftrightarrow}
\def\imp{\Longrightarrow}

\def\bN{\mathbb{N}}

\def\bP{\mathbb{P}}

\def\bR{\mathbb{R}}

\def\bH{\mathbb{H}}
\def\bz{\left(}
\def\jz{\right)}

\def\inv{^{-1}}

\def\egy{\mathbf 1}

\def\map{\Phi}

\def\sa{\mathrm{sa}}

\def\rho{\varrho}
\def\povm{\mathrm{POVM}}

\def\nn{\nonumber}

\def\nw{^{*}}

\def\meas{\mathrm{meas}}

\def\scli{\underline{\mathrm{sc}}}
\def\scls{\overline{\mathrm{sc}}}
\def\scl{\mathrm{sc}}

\def\oll{\overline}

\def\p{_{\gneq 0}}

\def\sa{\mathrm{sa}}

\def\valt{\cdot}
\def\wo{\mathrm{(wo)}}
\def\so{\mathrm{(so)}}
\def\fa{_{\mathrm{fa}}}
\def\ofa{_{\overline{\mathrm{fa}}}}
\def\f{_{\mathrm{f}}}
\def\tos{\gamma}
\def\alphaz{\mathbb{A}}
\def\Dmax{D_{\max}}

\newcommand{\ki}[1]{\textit{\textit{#1}}}

\newcommand{\s}{\mbox{ }}
\newcommand{\ds}{\mbox{ }\mbox{ }}

\newcommand{\norm}[1]{\left\| #1\right\|}

\newcommand{\inner}[2]{\left\langle #1 , #2\right\rangle}

\newcommand{\abs}[1]{\left| #1 \right|}

\newcommand{\vecc}[1]{\underline{#1}}

\newcommand{\bra}[1]{\left\langle #1\right|}
\newcommand{\ket}[1]{\left|#1\right\rangle}
\newcommand{\diad}[2]{\left|#1\right\rangle\!\left\langle #2\right|}

\newcommand{\pr}[1]{\diad{#1}{#1}}

\newcommand{\derleft}[1]{\partial^{-} #1}
\newcommand{\derright}[1]{\partial^{+} #1}

\newcommand{\vertleq}{\rotatebox{90}{$\,\ge$}}

\newcommand{\verteq}{\rotatebox{90}{$=$}}

\newcommand{\cupdot}{\mathbin{\mathaccent\cdot\cup}}

\makeatletter
\renewcommand{\p@enumii}{}
\makeatother

\DeclareMathOperator{\id}{id}
\DeclareMathOperator{\Tr}{Tr}

\DeclareMathOperator{\supp}{supp}

\DeclareMathOperator{\ran}{ran}
\DeclareMathOperator{\dom}{dom}

\DeclareMathOperator{\spec}{spec}

\DeclareMathOperator{\ootimes}{\otimes\ldots\otimes}

\DeclareMathOperator{\logn}{\widehat\log}

\DeclareMathOperator{\oran}{\overline{\ran}}
\begin{document}

\title{The strong converse exponent of discriminating infinite-dimensional quantum states}

\author{Mil\'an Mosonyi}
\email{milan.mosonyi@gmail.com}

\affiliation{
MTA-BME Lend\"ulet Quantum Information Theory Research Group
}

\affiliation{
Mathematical Institute, Budapest University of Technology and Economics, \\
M\H uegyetem rkp.~3., H-1111 Budapest, Hungary
}

\begin{abstract}
\centerline{\textbf{Abstract}}
\vspace{.3cm}

The sandwiched R\'enyi divergences
of two finite-dimensional density operators 
quantify their asymptotic distinguishability in the strong converse domain.
This establishes the sandwiched R\'enyi divergences as the operationally relevant ones among the infinitely many 
quantum extensions of the classical R\'enyi divergences for R\'enyi parameter $\alpha>1$. 
The known proof of this goes by showing that the sandwiched R\'enyi divergence coincides with 
the regularized measured R\'enyi divergence, which in turn is proved by asymptotic pinching, a 
fundamentally finite-dimensional technique. 
Thus, while the notion of the sandwiched R\'enyi divergences was extended 
recently to 
density operators on an infinite-dimensional Hilbert space
(in fact, even for states of an arbitrary von Neumann algebra),
these quantities were so far lacking an operational interpretation 
similar to the finite-dimensional case, and it has also been 
open whether they coincide with the regularized measured R\'enyi divergences.  
In this paper we fill this gap by answering both questions in the positive for density 
operators on an infinite-dimensional Hilbert space, using a simple finite-dimensional approximation technique.

We also initiate the study of the sandwiched R\'enyi divergences, and the related problem of 
the strong converse exponent, for 
pairs of positive semi-definite operators that are not necessarily trace-class (this 
corresponds to considering weights in a general von Neumann algebra setting). 
This is motivated by the need to define conditional R\'enyi entropies in the infinite-dimensional setting,
while it might also be interesting from the purely mathematical point of view of extending 
the concept of R\'enyi (and other) divergences to settings beyond the standard one of positive trace-class operators (positive normal functionals in the von Neumann algebra setting). 
In this spirit, we also discuss the definition and some properties 
of the more general family of R\'enyi $(\alpha,z)$-divergences of positive semi-definite operators on an infinite-dimensional separable Hilbert space.
\end{abstract}

\maketitle

\section{Introduction}

In a simple binary i.i.d.~quantum state discrimination 
problem, an experimenter is presented with several identically prepared quantum systems, 
all in the same state that
is either described by a density operator $\rho$ on the system's Hilbert space $\hil$, 
(\ki{null-hypothesis} $H_0$), 
or by another density operator $\sigma$
(\ki{alternative hypothesis} $H_1$).
The experimenter's task is to guess which hypothesis is correct, based on the result of a $2$-outcome measurement, represented by a pair of operators
$(T_n(0)=:T_n,T_n(1)=I-T_n)$, where $T_n\in\B(\hil_n)_{[0,I]}$ is 
a test on $\hil_n:=\hil^{\otimes n}$,
and $n$ is the number of identically prepared systems.
If the outcome of the measurement is $k$, described by the measurement operator $T_n(k)$, the experimenter decides that hypothesis $k$ is true. 
The \ki{type I success probability}, i.e., the probability that 
the experimenter correctly identifies the state to be $\rho_n$, and the \ki{type II error probability},
i.e., the probability  that 
the experimenter erroneously identifies the state to be $\rho_n$, are given by 
\begin{align}\label{errors intro}
\tos_n(T_n|\rho_n):=\Tr\rho_n T_n,\ds\ds\ds
\beta_n(T_n|\sigma_n):=\Tr \sigma_n T_n,
\end{align}
respectively, where $\rho_n=\rho^{\otimes n}$, $\sigma_n=\sigma^{\otimes n}$. 

In the asymptotic analysis of the problem, it is customary to look for the optimal asymptotics 
of the type I success probabilities under the constraint that the type II error probabilities 
decrease at least as fast as $\beta_n\sim e^{-nr}$ with some fixed $r$. 
It is known that if $r$ is smaller then the relative entropy of $\rho$ and $\sigma$ then 
the type I success probabilities converge to $1$ exponentially fast, and the optimal exponent
(the so-called direct exponent) is equal to the 
Hoeffding divergence $H_r$ of $\rho$ and $\sigma$ \cite{ANSzV,Hayashicq,Nagaoka,JOPS}.
The Hoeffding divergences are defined from the Petz-type R\'enyi divergences
$D_{\alpha}$ with $\alpha\in(0,1)$, and the above result 
establishes the operational significance of these divergences
\cite{Nagaoka,MH}.

On the other hand, 
it was shown in \cite{MO} (see also \cite{Hayashibook2,N,ON}) that if the 
Hilbert space is finite-dimensional,
(equivalently, the density operators are of finite rank), 
and $r$ is larger than the relative entropy, then 
the type I success probabilities converge to $0$ 
exponentially fast, and the optimal exponent
(the so-called strong converse exponent) is equal to the 
Hoeffding anti-divergence $H_r\nw$ of $\rho$ and $\sigma$.
($H_r\nw$, as well as the various divergences mentioned below, will be precisely defined in the main text.)
The Hoeffding anti-divergences are defined from the sandwiched R\'enyi divergences
$D_{\alpha}\nw$ with $\alpha>1$ \cite{Renyi_new,WWY},
and this result establishes the operational significance of these divergences
 
A key step in the proof of the strong converse exponent in \cite{MO} is showing that the
regularized measured R\'enyi divergence $\oll D_{\alpha}^{\meas}$
coincides with the sandwiched R\'enyi divergence $D_{\alpha}\nw$ for any $\alpha>1$,
which was proved using the pinching inequality 
\cite{H:pinching}, a fundamentally finite-dimensional technique. 
Thus, while the notion of the sandwiched R\'enyi divergences was extended 
recently to 
density operators on an infinite-dimensional Hilbert space
(in fact, even for states of an arbitrary von Neumann algebra)
in \cite{BST} and \cite{Jencova_NCLp},
these quantities were so far lacking an operational interpretation 
similar to the finite-dimensional case described above, and it has also been 
open whether they coincide with the regularized measured R\'enyi divergences.  
In this paper we fill this gap by answering both questions in the positive for density 
operators on an infinite-dimensional Hilbert space.  

We also initiate the study of the sandwiched R\'enyi divergences, and the related problem of 
the strong converse exponents, for 
pairs of positive semi-definite operators that are not necessarily trace-class (this corresponds to considering weights in a general von Neumann algebra setting). 
This is motivated by the need to define conditional R\'enyi entropies in the infinite-dimensional setting,
while it might also be interesting from the purely mathematical point of view of extending the concept of 
R\'enyi (and other) divergences to settings beyond the standard one of positive trace-class operators (or positive normal functionals, in the von Neumann algebra setting). 
In this spirit, we also discuss the definition and some properties 
of the more general family of R\'enyi $(\alpha,z)$-divergences \cite{AD,JOPP} in this setting. 
To the best of our knowledge, this is new even for trace-class operators when the underlying Hilbert space is infinite-dimensional .

The structure of the paper is as follows.
In Section \ref{sec:prelim} we collect some necessary preliminaries.
In Section \ref{sec:sand Renyi} we define the R\'enyi $(\alpha,z)$-divergences for an arbitrary pair of positive semi-definite operators on a possibly infinite-dimensional Hilbert space,
and establish some of their properties. 
The most important part of this section for the later applications is the recoverability of the 
sandwiched R\'enyi divergence from finite-dimensional restrictions, given in Proposition 
\ref{cor:fa equality}. Based on this, in Section \ref{sec:measured Renyi} we show that the sandwiched R\'enyi divergence is equal to the 
regularized measured R\'enyi divergence for pairs of states, extending the finite-dimensional 
result of \cite{MO} to infinite dimension.
In Section \ref{sec:sc Hanti} we 
consider a generalization of the state discrimination problem where the hypotheses are given by (not necessarily trace-class) positive semi-definite operators, and establish lower and upper bounds on the strong converse exponents in this setting. In particular,
we show that the strong converse exponent is equal to the Hoeffding anti-divergence 
for quantum states, thereby giving an operational interpretation of the sandwiched R\'enyi divergences analogous to the finite-dimensional case. Moreover, we prove the above equality also in the case where the reference operator $\sigma$ is only assumed to be compact, 
and to dominate the first operator $\rho$ as $\rho\le \lambda\sigma$ for some $\lambda>0$. 
In Section \ref{sec:cutoff}, we give a direct operational interpretation to the sandwiched R\'enyi divergences as generalized cutoff rates, extending the analogous 
interpretations given previously for classical \cite{Csiszar} and finite-dimensional quantum states \cite{MO}.
In Section \ref{sec:mon} we use the strong converse result from Section \ref{sec:sc Hanti} to show the monotonicity of the sandwiched R\'enyi divergences under the action of the dual of a normal unital completely positive map. While this follows from \cite{BST,Jencova_NCLp} for density operators, our proof is completely different, and also applies to other settings, e.g., for a compact $\sigma$ that dominates $\rho$.

\section{Preliminaries}
\label{sec:prelim}

Throughout the paper, $\hil$ and $\kil$ will denote separable Hilbert spaces
(of finite or infinite dimension), and 
$\B(\hil,\kil)$ will denote the set of everywhere defined bounded linear operators from $\hil$ to $\kil$, with $B(\hil,\hil)=:\B(\hil)$. 
We will use the notations 
$\B(\hil)_{\sa}$ for the set of self-adjoint,
and $\B(\hil)\p$, for the set of non-zero positive semi-definite (PSD), 
operators in $\B(\hil)$ respectively, and 
\begin{align*}
\B(\hil)_{[0,I]}:=\{T\in\B(\hil):\,0\le T\le I\}
\end{align*}
for the set of \ki{tests} in $\B(\hil)$.
A test $T$ is \ki{projective} if $T^2=T$.
We will denote the set of all projections on $\hil$ by $\bP(\hil)$, and the set of finite rank projections by $\bP_f(\hil)$. The set of finite-rank operators on $\hil$ will be denoted by 
$\B_f(\hil)$. The set of density operators, or states, on 
$\hil$ will be denoted by $\S(\hil)$.
For two PSD operators $\rho,\sigma\in\B(\hil)\p$, we will use the notations 
\begin{align*}
\B(\hil,\kil)_{\rho,\sigma}^+:=\{K\in\B(\hil,\kil):\,K\rho K^*\ne 0,\s K\sigma K^*\ne 0\},
\end{align*}
and $\B_f(\hil)_{\rho,\sigma}^+:=\B_f(\hil)\cap \B(\hil)_{\rho,\sigma}^+$,
$\bP(\hil)_{\rho,\sigma}^+:=\bP(\hil)\cap \B(\hil)_{\rho,\sigma}^+$,
$\bP_f(\hil)_{\rho,\sigma}^+:=\bP_f(\hil)\cap \B(\hil)_{\rho,\sigma}^+$.

For a (possibly unbounded) self-adjoint operator $A$ on a Hilbert space $\hil$, let 
$P^A(\valt)$ denote its spectral PVM, and for any complex-valued measurable function 
$f$ defined at least on $\spec(A)$, let $f(A)=\int_{\bR}f\,dP^A$ be the operator defined via the usual functional calculus. We will use the relations
\begin{align}
&(f(A))^*=\oll f(A),\label{fcalc adjoint}\\
&\oll{f(A)g(A)}=(fg)(A),\ds\ds\dom(f(A)g(A))=\dom(g(A))\cap\dom((fg)(A)),
\label{fcalc product}
\end{align}
where $\oll f$ stands for the pointwise complex conjugate of $f$, and for a closable operator $X$, $\oll X$ denotes its closure.

We say that a (not necessarily everywhere defined or bounded) linear operator $A$ 
on a Hilbert space is positive semi-definite (PSD), if 
it is self-adjoint, and 
$\spec(A)\subseteq[0,+\infty)$. 
If $A$ is PSD then we may define its real powers as 
\begin{align*}
A^p:=\id_{(0,+\infty)}^p(A)=\int_{(0,+\infty)}\id_{(0,+\infty)}^p\,dP^A,\ds\ds\ds p\in\bR.
\end{align*}
In particular, $A^0$ is the projection onto $(\ker A)^{\perp}=\oran A=:\supp A$, 
\begin{align*}
(A^p)\inv=(A\inv)^p=A^{-p},\ds\ds\ds p\in\bR,
\end{align*}
and 
\begin{align*}
A\ds\text{ is bounded}\ds\ds\imp\ds\ds A^{-p} A^p=I,\ds\ds\ds A^pA^{-p}=I_{\ran A^p},\ds\ds\ds p>0.
\end{align*}

For any $X\in\B(\hil,\kil)$ with polar decomposition $X=V|X|$, we have 
$|X^*|=V|X|V^*$, whence $|X^*|^p=V|X|^pV^*$ for any $p\in\bR$. In particular,
\begin{align*}
\Tr (X^*X)^p=\Tr (XX^*)^p,\ds\ds\ds p>0,
\end{align*}
which we will use in many proofs below without further notice. 
We will use the notation $\norm{X}_p:=(\Tr|X|^p)^{1/p}$ for 
$X\in\B(\hil,\kil)$ and $p>0$. When $p\ge 1$, $\norm{\valt}_p$ is a norm on the 
Schatten $p$-class
\begin{align*}
\L^p(\hil):=\{X\in\B(\hil):\,\Tr|X|^p<+\infty\}.
\end{align*}
We will denote the usual operator norm on $\B(\hil)$ by $\norm{\valt}_{\infty}$.

\begin{lemma}\label{lemma:Holder}
(H\"older inequality)
Let $p_0,p_1,p>0$ be such that $\frac{1}{p_0}+\frac{1}{p_1}=\frac{1}{p}$. For any 
$A,B\in\B(\hil)$, 
\begin{align}\label{Holder}
\norm{AB}_p\le\norm{A}_{p_0}\norm{B}_{p_1}.
\end{align}
Moreover, if $\norm{A}_{p_0}\norm{B}_{p_1}<+\infty$ then equality holds in 
\eqref{Holder} if and only if $A=\lambda B$ or $B=\lambda A$ for some $\lambda\ge 0$.
\end{lemma}
\begin{proof}
The inequality is well-known; see, e.g., \cite[Proposition 2.7]{Hiai_expop}.
The characterization of equality has been known for a long time in the case $p=1$;
see, e.g., \cite{Dixmier_Holder} and
\cite{Larotonda_Holder_equality}. For the case of a general positive $p$, see 
\cite{HM_Holder}.
\end{proof}

We will use the notations $\wo\lim$ and $\so\lim$ for limits in the weak and the strong operator topologies, respectively.
The following two statements are from \cite{Grumm}.

\begin{lemma}\label{lemma:Grumm}
Let $A\in \L^p(\hil)$ for some $p\ge 1$, and
$B_n\in\B(\hil,\kil)$, $C_n\in\B(\kil,\hil)$, $n\in\bN$, be two sequences
bounded in operator norm and
converging strongly to some
$B_{\infty}\in\B(\hil,\kil)$ and $C_{\infty}\in\B(\kil,\hil)$, respectively.
Then 
\begin{align}
\lim_{n\to+\infty}\norm{B_nAC_n-B_{\infty}AC_{\infty}}_p=0,\ds\ds\ds
\lim_{n\to+\infty}\norm{B_nAC_n}_p=\norm{B_{\infty}AC_{\infty}}_p.\label{Grumm1}
\end{align}
\end{lemma}
\begin{proof}
The first limit in \eqref{Grumm1} is immediate from \cite[Theorem 1]{Grumm},
and the second limit follows from it trivially. 
\end{proof}

The following is Theorem 2 in \cite{Grumm}:

\begin{lemma}\label{lemma:Grumm2}
Let $p\in[1,+\infty)$ and $A,A_n\in \L^p(\hil)$, $n\in\bN$, be such that 
$\so\lim_n A_n=A$, $\so\lim_n A_n^*=A^*$, and $\lim_n\norm{A_n}_p=\norm{A}_p$. Then
$\lim_n\norm{A_n-A}_p=0$.
\end{lemma}

The following is a special case of \cite[Proposition 2.11]{Hiai_expop}:

\begin{lemma}\label{lemma:pnorm lsc}
Assume that a sequence $A_n\in\B(\hil)$, $n\in\bN$, converges to some 
$A\in\B(\hil)$ in the weak operator topology. For any 
$p\in[1,+\infty]$, 
\begin{align*}
\norm{A}_{p}\le\liminf_{n\to+\infty}\norm{A_n}_{p}.
\end{align*}
\end{lemma}
\medskip

We will need the following straightforward generalization of the minimax theorem from \cite[Corollary A.2]{MH}. Its proof is essentially the same, which we include for readers' convenience.

\begin{lemma}\label{lemma:minimax2}
Let $X$ be a compact topological space, $Y$ be an upward directed partially ordered set, and let $f:\,X\times Y\to \bR\cup\{-\infty,+\infty\}$ be a function. Assume that
\smallskip

\s(i) $f(.\,,\,y)$ is upper semicontinuous for every $y\in Y$ and
\smallskip

(ii) $f(x,.)$ is monotonic decreasing for every $x\in X$.

\noindent Then 
\begin{align}\label{minimax statement}
\sup_{x\in X}\inf_{y\in Y}f(x,y)=
\inf_{y\in Y}\sup_{x\in X}f(x,y),
\end{align}
and the suprema in \eqref{minimax statement} can be replaced by maxima.
\end{lemma}
\begin{proof}
The inequality $\sup_{x\in X}\inf_{y\in Y}f(x,y)\le\inf_{y\in Y}\sup_{x\in X}f(x,y)$
is trivial, and for the converse inequality it is sufficient to prove
that for any finite subset $Y'\subseteq Y$, 
\begin{align*}
\sup_{x\in X}\inf_{y\in Y'}f(x,y)\ge \inf_{y\in Y}\sup_{x\in X}f(x,y),
\end{align*}
according to \cite[Lemma A.1]{MH} (applied to $-f$ in place of $f$).
Due to $Y$ being upward directed, for any finite subset $Y'\subseteq Y$, there exists 
a $y^*\in Y$ such that $y\le y^*$ for every $y\in Y'$. Since 
$f(x,.)$ is assumed to be monotone decreasing, we get
\begin{align*}
\sup_{x\in X}\inf_{y\in Y'}f(x,y)\ge
\sup_{x\in X}f(x,y^*)
\ge
\inf_{y\in Y}\sup_{x\in X}f(x,y),
\end{align*}
as required. The assertion about the maxima is straightforward from the assumed semi-continuity and the compactness of $X$.
\end{proof}

\section{The R\'enyi $(\alpha,z)$-divergences in infinite dimension}
\label{sec:sand Renyi}

The sandwiched R\'enyi $\alpha$-divergences for pairs of finite-dimensional density operators 
were introduced in \cite{Renyi_new,WWY}.
The R\'enyi $(\alpha,z)$-divergences \cite{AD,JOPP} give a $2$-parameter extension of this family, which includes both the sandwiched R\'enyi divergences (corresponding to $z=\alpha$)
and the Petz-type, or standard R\'enyi divergences \cite{P86} (corresponding to $z=1$) as special cases. 
 
The concept of the sandwiched R\'enyi divergences was extended recently to 
pairs of positive normal linear functionals on a general von Neumann algebra
in \cite{BST,Jencova_NCLp,Jencova_NCLpII}, while the Petz-type R\'enyi divergences 
have been studied in this more general setting for a long time 
\cite{Kosaki1982,Petz_QE_vN,Hiai_fdiv_Springer}. 
These extensions require 
advanced knowledge of von Neumann algebras, 
and the details of the proofs might be difficult to verify for those who are not experts in the subject. Below we give a more pedestrian exposition of the definition and basic 
properties of the R\'enyi divergences in the simpler case where the 
states are represented by density operators on a possibly infinite-dimensional Hilbert space, while in the same time we also generalize the above works in this setting to the case where density operators may be replaced by arbitrary positive semi-definite operators. Since these are mostly not assumed to be trace-class, they cannot be normalized to states in the properly infinite-dimensional case. Moreover, we also consider the more general notion of 
R\'enyi $(\alpha,z)$-divergences in this setting. 

The recoverability of the sandwiched R\'enyi divergences from finite-size restrictions, 
given in Proposition \ref{cor:fa equality}, seems to be new even for 
density operators, although in that case it follows easily from 
the known properties of monotonicity and lower semi-continuity of the sandwiched R\'enyi divergences.

\subsection{Definition and basic properties}

The sandwiched R\'enyi divergence of $\rho$ and $\sigma$ is finite according to 
the definition in \cite{Jencova_NCLp} if and only if 
$\rho$ is in Kosaki's interpolation space $\L^{\alpha}(\hil,\sigma)$. 
The following lemma gives various alternative characterizations of this condition, and 
also an extension that we will use in the definition of the 
R\'enyi $(\alpha,z)$-divergences in this setting.
The lemma is essentially a special case of Douglas' range inclusion theorem 
\cite{Douglas} for PSD operators with $A:=\rho^{\frac{\alpha}{2z}}$ and 
$B:=\sigma^{\frac{\alpha-1}{2z}}$ (points \ref{rho alpha4}--\ref{rho alpha6}) as well as
an extension with further equivalent characterizations 
(points \ref{rho alpha1}--\ref{rho alpha3}), and it is inspired by a similar statement 
for the $\alpha=z=+\infty$ case given in \cite{LI2021175}.

Let us introduce the notation
\begin{align*}
\alphaz:=(1,+\infty)\times(0,+\infty)\cupdot\{(+\infty,+\infty)\}.
\end{align*}
For $(\alpha,z):=(+\infty,+\infty)$, we will use the convention
$\frac{\alpha}{z}:=1$, and define similar expressions by a formal calculus, e.g., 
$\frac{\alpha}{2z}:=\half\frac{\alpha}{z}=\half$,
$\frac{\alpha-1}{2z}:=\frac{\alpha}{2z}-\frac{1}{2z}=\half$, etc.

\begin{lemma}\label{lemma:rho alpha}
Let $\rho,\sigma\in\B(\hil)\p$, and let 
$(\alpha,z)\in\alphaz$.
The following are equivalent:
\begin{enumerate}
\item\label{rho alpha1}
There exists an $R\in\B(\hil)$ such that 
\begin{align}\label{rho in weighted alpha}
\rho^{\frac{\alpha}{z}}=\sigma^{\frac{\alpha-1}{2z}}R\sigma^{\frac{\alpha-1}{2z}}.
\end{align}
\item\label{rho alpha2}
$\ran\rho^{\frac{\alpha}{z}}\subseteq\ran\sigma^{\frac{\alpha-1}{2z}}$, and 
$\sigma^{\frac{1-\alpha}{2z}}\rho^{\frac{\alpha}{z}}\sigma^{\frac{1-\alpha}{2z}}$\s 
is densely defined and bounded.
\item\label{rho alpha3}
$\rho^0\le\sigma^0$, and 
for any/some sequences 
$0<c_n<d_n$ with $c_n\to 0$, $d_n\to+\infty$,
the sequence of bounded operators 
\begin{align}\label{truncated ops}
\sigma_n^{\frac{1-\alpha}{2z}}\rho^{\frac{\alpha}{z}}
\sigma_n^{\frac{1-\alpha}{2z}},\ds \ds n\in\bN,
\end{align}
converges in the weak/strong operator topology,
where $\sigma_n:=P_n\sigma P_n$, $P_n:=\egy_{(c_n,d_n)}(\sigma)$.
\item\label{rho alpha4}
$\ran\rho^{\frac{\alpha}{2z}}\subseteq\ran\sigma^{\frac{\alpha-1}{2z}}$.
\item\label{rho alpha5}
$\sigma^{\frac{1-\alpha}{2z}}\rho^{\frac{\alpha}{2z}}\in\B(\hil)$. 
\item\label{rho alpha6}
There exists a $\lambda\ge 0$ such that 
$\rho^{\frac{\alpha}{z}}\le\lambda\sigma^{\frac{\alpha-1}{z}}$.
\end{enumerate}
Moreover, if the above hold then $\oll{\rho^{\frac{\alpha}{2z}}\sigma^{\frac{1-\alpha}{2z}}}=
\bz\sigma^{\frac{1-\alpha}{2z}}\rho^{\frac{\alpha}{2z}}\jz^*$, and 
among all operators $R$ as in \eqref{rho in weighted alpha} 
there exists a unique PSD operator
with the property $R^0\le\sigma^0$, 
denoted by $\rho_{\sigma,\alpha,z}$,
which can be expressed as
\begin{align}
\rho_{\sigma,\alpha,z}
&=
\oll{\sigma^{\frac{1-\alpha}{2z}}\rho^{\frac{\alpha}{z}}\sigma^{\frac{1-\alpha}{2z}}}
=
\bz\sigma^{\frac{1-\alpha}{2z}}\rho^{\frac{\alpha}{2z}} \jz
\oll{\rho^{\frac{\alpha}{2z}}\sigma^{\frac{1-\alpha}{2z}}}
=
\bz\sigma^{\frac{1-\alpha}{2z}}\rho^{\frac{\alpha}{2z}} \jz\bz\sigma^{\frac{1-\alpha}{2z}}\rho^{\frac{\alpha}{2z}} \jz^*\label{rho alpha eq1}\\
&=
\so\lim_{n\to+\infty}\sigma_n^{\frac{1-\alpha}{2z}}\rho^{\frac{\alpha}{z}}\sigma_n^{\frac{1-\alpha}{2z}}
=
\wo\lim_{n\to+\infty}\sigma_n^{\frac{1-\alpha}{2z}}\rho^{\frac{\alpha}{z}}\sigma_n^{\frac{1-\alpha}{2z}},\label{rho alpha eq2}
\end{align}
where $(\sigma_n)_{n\in\bN}$ is any sequence as in \ref{rho alpha3}.
This unique $\rho_{\sigma,\alpha,z}$ is in the von Neumann algebra generated by 
$\rho$ and $\sigma$, and its operator norm is equal to the smallest $\lambda$ for which 
\ref{rho alpha6} holds.
\end{lemma}

\begin{proof}
Note that if $R\in\B(\hil)$ satisfies 
\eqref{rho in weighted alpha} then so does 
$\sigma^0R\sigma^0$ as well.
Moreover, any of the conditions above imply $\rho^0\le\sigma^0$. Hence,
we may assume without loss of generality that $\supp\sigma=\hil$, so that 
$\oran(\sigma^{\frac{\alpha-1}{2z}})=
\bz\ker\bz\sigma^{\frac{\alpha-1}{2z}}\jz\jz^{\perp}=(\ker\sigma)^{\perp}=\hil$.

Assume that \ref{rho alpha1} holds. Then $\ran\rho^\frac{\alpha}{z}\subseteq\ran\sigma^{\frac{\alpha-1}{2z}}$ holds trivially, and
\begin{align*}
\sigma^{\frac{1-\alpha}{2z}}\rho^{\frac{\alpha}{z}}\sigma^{\frac{1-\alpha}{2z}}
&=
\underbrace{\sigma^{\frac{1-\alpha}{2z}}\sigma^{\frac{\alpha-1}{2z}}}_{=I}
R
\underbrace{\sigma^{\frac{\alpha-1}{2z}}\sigma^{\frac{1-\alpha}{2z}}}_{=I_{\ran\sigma^{\frac{\alpha-1}{2z}}}}
=
RI_{\ran\sigma^{\frac{\alpha-1}{2z}}}
=
R\Big\vert_{\ran\sigma^{\frac{\alpha-1}{2z}}},
\end{align*}
whence its closure is equal to $R$. This proves \ref{rho alpha2} and 
the existence of the unique $\rho_{\sigma,\alpha,z}$ with the postulated properties, as well as
the first equality in \eqref{rho alpha eq1}.
Moreover, for any 
$0<c_n<d_n$ with $c_n\to 0$, $d_n\to+\infty$,
we have 
\begin{align}\label{rho alpha proof1}
\sigma_n^{\frac{1-\alpha}{2z}}\rho^\frac{\alpha}{z}\sigma_n^{\frac{1-\alpha}{2z}}
&=
\underbrace{\sigma_n^{\frac{1-\alpha}{2z}}\sigma^{\frac{\alpha-1}{2z}}}_{=\egy_{(c_n,d_n)}(\sigma)}
\rho_{\sigma,\alpha,z}
\underbrace{\sigma^{\frac{\alpha-1}{2z}}\sigma_n^{\frac{1-\alpha}{2z}}}_{=\egy_{(c_n,d_n)}(\sigma)}\xrightarrow[n\to+\infty]{\so}\rho_{\sigma,\alpha,z},
\end{align}
where we used \eqref{fcalc product}. This proves \ref{rho alpha3} and 
the first equality in \eqref{rho alpha eq2}. Since 
$\sigma_n^{^{\frac{1-\alpha}{2z}}}\rho^\frac{\alpha}{z}\sigma_n^{\frac{1-\alpha}{2z}}$ is in the von Neumann algebra generated by 
$\rho$ and $\sigma$, so is $\rho_{\sigma,\alpha,z}$, according to 
\eqref{rho alpha proof1}.
Obviously, \ref{rho alpha1} also implies
\begin{align*}
\rho^\frac{\alpha}{z}\le\sigma^{\frac{\alpha-1}{2z}}(\norm{\rho_{\sigma,\alpha,z}}_{\infty}I)
\sigma^{\frac{\alpha-1}{2z}}
=
\norm{\rho_{\sigma,\alpha,z}}_{\infty}\sigma^{\frac{\alpha-1}{z}},
\end{align*}
whence \ref{rho alpha6} follows with $\lambda:=\norm{\rho_{\sigma,\alpha,z}}_{\infty}$.
As a consequence, $\lambda_{\min}\le\norm{\rho_{\sigma,\alpha,z}}_{\infty}$, where 
$\lambda_{\min}$ denotes the smallest $\lambda$ for which \ref{rho alpha6} holds.
Conversely, let $\lambda$ be as in \ref{rho alpha6}. Multiplying both sides by
$\sigma_n^{\frac{1-\alpha}{2z}}$ yields
$ \sigma_n^{\frac{1-\alpha}{2z}}\rho^\frac{\alpha}{z}\sigma_n^{\frac{1-\alpha}{2z}}
\le \lambda\egy_{(c_n,d_n)}(\sigma)$, which in combination with \eqref{rho alpha proof1}
gives $\norm{\rho_{\sigma,\alpha,z}}_{\infty}\le\lambda$. Thus,
$\lambda_{\min}=\norm{\rho_{\sigma,\alpha,z}}_{\infty}$, as stated.

Assume next that \ref{rho alpha2} holds. Then 
\begin{align}\label{rho alpha proof2}
\sigma^{\frac{\alpha-1}{2z}}\,
\oll{\sigma^{\frac{1-\alpha}{2z}}\rho^{\frac{\alpha}{z}}\sigma^{\frac{1-\alpha}{2z}}}
\,\sigma^{\frac{\alpha-1}{2z}}
\supseteq
\underbrace{\sigma^{\frac{\alpha-1}{2z}}
\sigma^{\frac{1-\alpha}{2z}}}_{=I_{\ran\sigma^{\frac{\alpha-1}{2\alpha}}}}
\rho^{\frac{\alpha}{z}}
\underbrace{\sigma^{\frac{1-\alpha}{2z}}
\sigma^{\frac{\alpha-1}{2z}}}_{=I}=\rho^{\frac{\alpha}{z}},
\end{align}
where the last equality follows from the assumption
$\ran\rho^{\frac{\alpha}{z}}\subseteq\ran\sigma^{\frac{\alpha-1}{2z}}$.
Since $\rho^{\frac{\alpha}{z}}$ is everywhere defined, it is actually equal to the 
first operator in \eqref{rho alpha proof2}, and thus \ref{rho alpha1} holds.
Moreover, if \eqref{rho alpha proof2} holds then for any $\phi\in\hil$, 
\begin{align*}
0\le\inner{\phi}{\rho^{\frac{\alpha}{z}}\phi}
=
\inner{\phi}{\sigma^{\frac{\alpha-1}{2z}}\,
\oll{\sigma^{\frac{1-\alpha}{2z}}\rho^{\frac{\alpha}{z}}\sigma^{\frac{1-\alpha}{2z}}}
\,\sigma^{\frac{\alpha-1}{2z}}\phi}
=
\inner{\sigma^{\frac{\alpha-1}{2z}}\phi}{
\oll{\sigma^{\frac{1-\alpha}{2z}}\rho^{\frac{\alpha}{z}}\sigma^{\frac{1-\alpha}{2z}}}
(\sigma^{\frac{\alpha-1}{2z}}\phi)}.
\end{align*}
Since $\ran\sigma^{\frac{\alpha-1}{2z}}$ is dense and 
$\oll{\sigma^{\frac{1-\alpha}{2z}}\rho^{\frac{\alpha}{z}}\sigma^{\frac{1-\alpha}{2z}}}$ is bounded, it follows that $\oll{\sigma^{\frac{1-\alpha}{2z}}\rho\sigma^{\frac{1-\alpha}{2z}}}$
is PSD.

Assume now \ref{rho alpha3}, i.e., that for some sequences
$0<c_n<d_n$ with $c_n\to 0$, $d_n\to+\infty$,
the sequence of operators
$\bz\sigma_n^{\frac{1-\alpha}{2z}}\rho^{\frac{\alpha}{z}}\sigma_n^{\frac{1-\alpha}{2z}}\jz_{n\in\bN}$
converges in the weak operator topology to some operator $R_{\sigma,\alpha,z}$. Then 
\begin{align*}
\sigma^{\frac{\alpha-1}{2z}}R_{\sigma,\alpha,z}\sigma^{\frac{\alpha-1}{2z}}
=
\wo\lim_{n\to+\infty}
\underbrace{\sigma^{\frac{\alpha-1}{2z}}\sigma_n^{\frac{1-\alpha}{2z}}}_{=\egy_{(c_n,d_n)}(\sigma)}
\rho^{\frac{\alpha}{z}}
\underbrace{\sigma_n^{\frac{\alpha-1}{2z}}\sigma^{\frac{\alpha-1}{2z}}}_{=\egy_{(c_n,d_n)}(\sigma)}=\rho^{\frac{\alpha}{z}},
\end{align*}
and hence \ref{rho alpha1} holds, as well as the second equality in 
\eqref{rho alpha eq2}.

The equivalence of \ref{rho alpha4}, \ref{rho alpha5}, and \ref{rho alpha6} follows from 
Douglas' range inclusion theorem \cite{Douglas}.
Note that \ref{rho alpha4}$\iff$\ref{rho alpha5} is simple, as 
$\sigma^{\frac{1-\alpha}{2z}}\rho^{\frac{\alpha}{2z}}$ being everywhere defined is equivalent to 
$\ran\rho^{\frac{\alpha}{2z}}\subseteq\dom \sigma^{\frac{1-\alpha}{2z}}=\ran\sigma^{\frac{\alpha-1}{2z}}$, and boundedness 
of $\sigma^{\frac{1-\alpha}{2z}}\rho^{\frac{\alpha}{2z}}$ is automatic from the boundedness of 
$\rho^{\frac{\alpha}{2z}}$ and the closedness of $\sigma^{\frac{1-\alpha}{2z}}$, due to the closed graph theorem. 
Moreover, we have 
\begin{align*}
\sigma^{\frac{\alpha-1}{2z}}\bz\sigma^{\frac{1-\alpha}{2z}}\rho^{\frac{\alpha}{2z}}\jz
=
I_{\ran \sigma^{\frac{\alpha-1}{2z}}}\rho^{\frac{\alpha}{2z}}=\rho^{\frac{\alpha}{2z}},
\end{align*}
whence
\begin{align*}
\rho^{\frac{\alpha}{2z}}\sigma^{\frac{1-\alpha}{2z}}
=
\bz\rho^{\frac{\alpha}{2z}}\jz^*\sigma^{\frac{1-\alpha}{2z}}
=
\bz\sigma^{\frac{1-\alpha}{2z}}\rho^{\frac{\alpha}{2z}}\jz^*
\underbrace{\sigma^{\frac{\alpha-1}{2z}}\sigma^{\frac{1-\alpha}{2z}}
}_{=I_{\ran \sigma^{\frac{\alpha-1}{2z}}}},
\end{align*}
which is densely defined and bounded. Thus,
$\oll{\rho^{\frac{\alpha}{2z}}\sigma^{\frac{1-\alpha}{2z}}}=
\bz\sigma^{\frac{1-\alpha}{2z}}\rho^{\frac{\alpha}{2z}}\jz^*$.
Finally, 
\begin{align}\label{rho alpha proof3}
\sigma^{\frac{\alpha-1}{2z}}\left[\bz\sigma^{\frac{1-\alpha}{2z}}\rho^{\frac{\alpha}{2z}}\jz
\oll{\rho^{\frac{\alpha}{2z}}\sigma^{\frac{1-\alpha}{2z}}}\,\right]
\sigma^{\frac{\alpha-1}{2z}}
\supseteq
\underbrace{\sigma^{\frac{\alpha-1}{2z}}\sigma^{\frac{1-\alpha}{2z}}
}_{=I_{\ran \sigma^{\frac{\alpha-1}{2z}}}}
\rho^{\frac{\alpha}{2z}}
\rho^{\frac{\alpha}{2z}}
\underbrace{\sigma^{\frac{1-\alpha}{2z}}
\sigma^{\frac{\alpha-1}{2z}}}_{=I}
=
\rho^{\frac{\alpha}{z}},
\end{align}
where the last equality follows from the assumption $\ran\rho^{\frac{\alpha}{2z}}\subseteq
\ran \sigma^{\frac{\alpha-1}{2z}}$. 
Thus, \ref{rho alpha1} follows with $\rho_{\sigma,\alpha,z}=
\bz\sigma^{\frac{1-\alpha}{2z}}\rho^{\frac{\alpha}{2z}}\jz
\oll{\rho^{\frac{\alpha}{2z}}\sigma^{\frac{1-\alpha}{2z}}}$, and we also have the 
second and the third equalities in \eqref{rho alpha eq1}.
Note that the last expression in \eqref{rho alpha eq1} gives another proof for the positive semi-definiteness of $\rho_{\sigma,\alpha,z}$.
\end{proof}

\begin{defin}\label{def:L alpha}
For $\sigma\in\B(\hil)\p$ and $(\alpha,z)\in\alphaz$, let
\begin{align*}
\B^{\alpha,z}(\hil,\sigma)
&:=
\left\{\rho\in\B(\hil)\p:\,
\exists\,R\in\B(\hil)\ds\text{s.t.} \ds
\rho^{\frac{\alpha}{z}} = \sigma^{\frac{\alpha-1}{2z}}R\sigma^{\frac{\alpha-1}{2z}}
\right\}.
\end{align*}
When $\alpha=z$, we will use the shorthand notation $\B^{\alpha,z}(\hil)=:\B^{\alpha}(\hil)$.
\end{defin}

\begin{rem}
Note that $\rho\in\B^{\alpha,z}(\hil,\sigma)$ 
if and only if it satisfies (i) in Lemma \ref{lemma:rho alpha}, which is equivalently characterized by all the other points in Lemma \ref{lemma:rho alpha}.
In particular, there exists a unique PSD $\rho_{\sigma,\alpha,z}$ 
satisfying $\rho_{\sigma,\alpha,z}^0\le\sigma^0$ and 
$\rho^{\frac{\alpha}{z}} = \sigma^{\frac{\alpha-1}{2z}}\rho_{\sigma,\alpha,z}\sigma^{\frac{\alpha-1}{2z}}$, and thus the map
$\rho\mapsto \rho_{\sigma,\alpha,z}$ is well-defined from $\B^{\alpha,z}(\hil,\sigma)$
onto $\{\tau\in\B(\hil)\p:\,\tau^0\le\sigma^0\}$, and it is also injective, hence it is a 
bijection.
When $\alpha=z$, we will use the notation 
$\rho_{\sigma,\alpha,z}=:\rho_{\sigma,\alpha}$.
\end{rem}

\begin{lemma}\label{lemma:B alpha z increasing}
For any $\sigma\in\B(\hil)\p$, $(0,+\infty)\ni z\mapsto \B^{\alpha,z}(\hil,\sigma)$ is increasing, i.e., 
\begin{align*}
0<z<z'\ds\imp\ds \B^{\alpha,z}(\hil,\sigma)\subseteq \B^{\alpha,z'}(\hil,\sigma).
\end{align*}
\end{lemma}
\begin{proof}
Let $\rho\in\B^{\alpha,z}(\hil,\sigma)$. Then, by \ref{rho alpha6} of 
Lemma \ref{lemma:rho alpha}, $\rho^{\frac{\alpha}{z}}\le\lambda\sigma^{\frac{\alpha-1}{z}}$
for some $\lambda\in(0,+\infty)$. Since $z\le z'$, $\id_{[0,+\infty)}^{\frac{z}{z'}}$ is operator monotone, whence
\begin{align*}
\rho^{\frac{\alpha}{z'}}
=
\bz\rho^{\frac{\alpha}{z}}\jz^{\frac{z}{z'}}
\le\lambda^{\frac{z}{z'}}\bz\sigma^{\frac{\alpha-1}{z}}\jz^{\frac{z}{z'}}
=
\sigma^{\frac{\alpha-1}{z'}}. 
\end{align*}
Again by \ref{rho alpha6} of Lemma \ref{lemma:rho alpha},
$\rho\in\B^{\alpha,z'}(\hil,\sigma)$.
\end{proof}

\begin{rem}
By \eqref{rho alpha eq1}--\eqref{rho alpha eq2}, for 
$P_n$ and $\sigma_n$ as in \eqref{truncated ops}, 
\begin{align*}
P_n\rho_{\sigma,\alpha,z}P_n
=\sigma_n^{\frac{1-\alpha}{2\alpha}}\rho\sigma_n^{\frac{1-\alpha}{2\alpha}},
\end{align*}
and if $\alpha=z$, then we further have
\begin{align*}
P_n\rho_{\sigma,\alpha}P_n
=\sigma_n^{\frac{1-\alpha}{2\alpha}}\rho
\sigma_n^{\frac{1-\alpha}{2\alpha}}
=
(P_n\sigma P_n)^{\frac{1-\alpha}{2\alpha}}(P_n\rho P_n)
(P_n\sigma P_n)^{\frac{1-\alpha}{2\alpha}}.
\end{align*}
Thus, with $\rho_n:=P_n\rho P_n$,
\begin{align*}
\rho_{\sigma,\alpha}=
\wo\lim_{n\to+\infty}
\sigma_n^{\frac{1-\alpha}{2\alpha}}\rho_n
\sigma_n^{\frac{1-\alpha}{2\alpha}}.
\end{align*}
\end{rem}

\begin{rem}\label{rem:smaller alpha}
Note that if $\rho\in\B^{\alpha}(\hil,\sigma)$, i.e., 
$\rho=\sigma^{\frac{\alpha-1}{2\alpha}}\rho_{\sigma,\alpha}\sigma^{\frac{\alpha-1}{2\alpha}}$ with 
$\rho_{\sigma,\alpha}\in\B(\hil)\p$, 
$\rho_{\sigma,\alpha}^0\le\sigma^0$,
then 
for any $\alpha'<\alpha$,
\begin{align*}
\rho=\sigma^{\frac{\alpha'-1}{2\alpha'}}
\sigma^{\frac{1}{2\alpha'}-\frac{1}{2\alpha}}
\rho_{\sigma,\alpha}
\sigma^{\frac{1}{2\alpha'}-\frac{1}{2\alpha}}
\sigma^{\frac{\alpha'-1}{2\alpha'}},
\end{align*}
whence $\rho\in\B^{\alpha'}(\hil,\sigma)$, and 
\begin{align}\label{smaller alpha}
\rho_{\sigma,\alpha'}=\sigma^{\frac{1}{2\alpha'}-\frac{1}{2\alpha}}
\rho_{\sigma,\alpha}
\sigma^{\frac{1}{2\alpha'}-\frac{1}{2\alpha}}.
\end{align}
In particular, if $\rho\in \B^{\infty}(\hil,\sigma)$, i.e., 
$\rho=\sigma^{1/2}\rho_{\sigma,\infty}\sigma^{1/2}$ with some 
$\rho_{\sigma,\infty}\in\B(\hil)\p$,
$\rho_{\sigma,\infty}^0\le\sigma^0$, then 
$\rho\in\B^{\alpha}(\hil,\sigma)$ for every $\alpha>1$, and 
\begin{align*}
\rho_{\sigma,\alpha}=\sigma^{\frac{1}{2\alpha}}\rho_{\sigma,\infty}\sigma^{\frac{1}{2\alpha}}.
\end{align*}
As an immediate consequence,
\begin{align*}
\cap_{\alpha>1}\B^{\alpha}(\hil,\sigma)
\supseteq
\B^{\infty}(\hil,\sigma)=\left\{\rho\in\B(\hil):\,\Dmax(\rho\|\sigma)<+\infty\right\},
\end{align*}
where 
\begin{align}\label{Dmax def}
\Dmax(\rho,\sigma):=\inf\{\kappa\in\bR:\,\rho\le e^{\kappa}\sigma\}
\end{align}
is the max-relative entropy of $\rho$ and $\sigma$ \cite{Datta,RennerPhD}.
\end{rem}

\begin{defin}
For $\sigma\in\B(\hil)\p$ and $(\alpha,z)\in(1,+\infty)\times(0,+\infty)$, let
\begin{align*}
\L^{\alpha,z}(\hil,\sigma)
&:=
\left\{
\rho\in\B^{\alpha,z}(\hil):\,
\Tr\rho_{\sigma,\alpha,z}^z<+\infty
\right\}.
\end{align*}
Again, when $\alpha=z$, we will use the notation $\L^{\alpha,z}(\hil,\sigma)=:\L^{\alpha}(\hil,\sigma)$.
\end{defin}

\begin{rem}\label{rem:L alpha trace class}
Note that for $\alpha>1$, $\sigma^{\frac{\alpha-1}{2z}}\in\B(\hil)$, and if 
$z\ge 1$ then $\L^{z}(\hil)$ is an ideal in $\B(\hil)$.
Thus, by \ref{rho alpha1} of Lemma \ref{lemma:rho alpha}, if $\rho\in \L^{\alpha,z}(\hil,\sigma)$ then 
$\rho^{\frac{\alpha}{z}}\in \L^{z}(\hil)$, 
or equivalently, $\rho\in\L^{\alpha}(\hil)$.
Therefore,
\begin{align*}
\L^{\alpha,z}(\hil,\sigma)\subseteq\B^{\alpha,z}(\hil,\sigma)\cap \L^{\alpha}(\hil),\ds\ds\alpha>1,\ds z\ge 1.
\end{align*}

Assume now that $\sigma$ is trace-class
and $\rho\in \L^{\alpha,z}(\hil,\sigma)$ for some 
$(\alpha,z)\in(1,+\infty)\times(0,+\infty)$.
Then, by Lemma \ref{rho alpha1} of \ref{lemma:rho alpha} and the operator 
H\"older inequality, 
$\Tr\bz \rho^{\frac{\alpha}{z}}\jz^r<+\infty$, where
$\frac{1}{r}=\frac{\alpha-1}{2z}+\frac{1}{z}+\frac{\alpha-1}{2z}=\frac{\alpha}{z}$,
or equivalently, $\rho\in\L^1(\hil)$. 
Thus, we get
\begin{align*}
\sigma\text{ trace-class }\ds\imp\ds
\L^{\alpha,z}(\hil,\sigma)\subseteq
\B^{\alpha,z}(\hil,\sigma)\cap \L^{1}(\hil),\ds\ds \alpha>1,\ds z>0.
\end{align*}
It is easy to see that the above inclusion is strict. Indeed, let 
$\sigma\in\B(l^2(\bN))$ be diagonal in the canonical basis of 
$l^2(\bN)$, i.e., 
$\sigma=\sum_{k\in\bN}s(k)\pr{\egy_{\{k\}}}$ for some 
$s:\,\bN\to(0,+\infty)$ such that $\sum_{k\in\bN}s(k)<+\infty$ (i.e., $\sigma$ is trace-class) and 
$\sum_{k\in\bN}s(k)^{\frac{\alpha-1}{z}}<+\infty$.
Define $\rho:=\sum_{k\in\bN}s(k)^{\frac{\alpha-1}{z}}\pr{\egy_{\{k\}}}$.
Then $\rho$ is trace-class, and for any sequence $(P_n)_{n\in\bN}$ as in 
Lemma \ref{lemma:rho alpha},
$\sigma_n^{\frac{1-\alpha}{2z}}\rho \sigma_n^{\frac{1-\alpha}{2z}}
=\sum_{k=1}^{m_n}\pr{\egy_{\{k\}}}$, which goes to $I$ in the strong operator topology.
Hence, $\rho\in\B^{\alpha,z}(\hil,\sigma)\cap \L^{1}(\hil)$, but 
$\rho_{\sigma,\alpha,z}=I$, and therefore
$\rho\notin\L^{\alpha,z}(\hil)$.
\end{rem}
\medskip

The following is an extension of the R\'enyi $(\alpha,z)$-divergences 
\cite{AD}
to the case of infinite-dimensional PSD operators. 
It is also a special case of Jen\v cov\'a's definition of the sandwiched R\'enyi 
divergence \cite{Jencova_NCLp} when $\rho$ and $\sigma$ are trace-class, and 
$z=\alpha$, and
it is a natural extension of it otherwise. 

\begin{defin}\label{def:Q alpha}
For $\rho,\sigma\in\B(\hil)\p$ and $(\alpha,z)\in(1,+\infty)\times(0,+\infty)$, let
\begin{align*}
Q_{\alpha,z}(\rho\|\sigma):=
\begin{cases}
\Tr\rho_{\sigma,\alpha,z}^{z},
&\rho\in \B^{\alpha,z}(\hil,\sigma),\\
+\infty,&\text{otherwise},
\end{cases}
\end{align*}
with $\rho_{\sigma,\alpha,z}$ as in Lemma \ref{lemma:rho alpha}.
The R\'enyi $(\alpha,z)$-divergence of $\rho$ and $\sigma$ is defined as
\begin{align*}
D_{\alpha,z}(\rho\|\sigma):=\frac{1}{\alpha-1}\log Q_{\alpha,z}(\rho\|\sigma).
\end{align*}
We use the notations $Q_{\alpha}\nw:=Q_{\alpha,\alpha}$ and
$D_{\alpha}\nw:=D_{\alpha,\alpha}$, and call the latter the 
\ki{sandwiched R\'enyi $\alpha$-divergence}.
\end{defin}

We also define the following variants of the R\'enyi $(\alpha,z)$-divergences for 
trace-class operators:
\begin{defin}\label{def:D alpha tilde}
For PSD trace-class operators $\rho,\sigma\in\L^1(\hil)\p$ and $(\alpha,z)\in(1,+\infty)\times(0,+\infty)$, let
\begin{align*}
\tilde D_{\alpha,z}(\rho\|\sigma):=D_{\alpha,z}(\rho\|\sigma)-\frac{1}{\alpha-1}\log\Tr\rho.
\end{align*}
We also use the notation
$\tilde D_{\alpha}\nw:=\tilde D_{\alpha,\alpha}$.
\end{defin}

\begin{rem}
For a convex function $f$ on $[0,+\infty)$, 
the \ki{quantum $f$-divergence} of a pair of positive normal functionals on a 
von Neumann algebra
is defined using the relative modular operator; see \cite{Hiai_fdiv_Springer,Petz_QE_vN}.
In particular, it is well-defined for a pair of positive trace-class operators 
$\rho,\sigma$ on a Hilbert space and $f_{\alpha}:=\id_{[0,+\infty)}^{\alpha}$
for any $\alpha>1$;
let it be denoted by $Q_{f_{\alpha}}(\rho\|\sigma)$.  
According to \cite[Theorem 3.6]{Hiai_fdiv_Springer},
\begin{align*}
Q_{f_{\alpha}}(\rho\|\sigma)=Q_{\alpha,1}(\rho\|\sigma),\ds\ds\ds\alpha>1.
\end{align*}
In particular, for PSD trace-class operators $\rho$ and $\sigma$, 
$D_{\alpha,1}(\rho\|\sigma)$ in Definition \ref{def:Q alpha} coincides with the 
\ki{Petz-type} or \ki{standard quantum R\'enyi $\alpha$-divergence} of $\rho$ and $\sigma$,
just as in the finite-dimensional case; see, e.g. \cite{AD}.
\end{rem}

\begin{rem}
Note that 
\begin{align*}
D_{\alpha,z}(\rho\|\sigma)<+\infty\ds\iff\ds
Q_{\alpha,z}(\rho\|\sigma)<+\infty\ds\iff\ds
\rho\in \L^{\alpha,z}(\hil,\sigma).
\end{align*}
\end{rem}

\begin{rem}\label{rem:scaling}
It is clear from their definitions that $Q_{\alpha,z}$, $D_{\alpha,z}$ 
and $\tilde D_{\alpha,z}$ 
satisfy the scaling properties
\begin{align}
Q_{\alpha,z}(\lambda\rho\|\eta\sigma)&=\lambda^{\alpha}\eta^{1-\alpha}Q_{\alpha,z}(\rho\|\sigma),\label{Q D scaling1}\\
D_{\alpha,z}(\lambda\rho\|\eta\sigma)&=
D_{\alpha,z}(\rho\|\sigma)+\frac{\alpha}{\alpha-1}\log\lambda-\log\eta,\label{Q D scaling2}\\
\tilde D_{\alpha,z}(\lambda\rho\|\eta\sigma)&=
D_{\alpha,z}(\rho\|\sigma)+\log\lambda-\log\eta,\label{Q D scaling3}
\end{align}
valid for any $\rho,\sigma\in\B(\hil)\p$ and $\lambda,\eta\in(0,+\infty)$.
% Note also that by definition, 
%\begin{align}
%\Dmax(\lambda\rho\|\eta\sigma)&=
%\Dmax(\rho\|\sigma)+\log\lambda-\log\eta.\label{Dmax scaling}
%\end{align}
\end{rem}

\begin{rem}
According to Lemma \ref{lemma:rho alpha}, if $\rho\in \B^{\alpha,z}(\hil,\sigma)$ then 
\begin{align*}
Q_{\alpha,z}(\rho\|\sigma)=
\Tr\oll{\sigma^{\frac{1-\alpha}{2z}}\rho^{\frac{\alpha}{z}}\sigma^{\frac{1-\alpha}{2z}}}^{\,z},
\end{align*}
which is a straightforward generalization of the formula for PSD operators on a finite-dimensional Hilbert space. Moreover, Lemma \ref{lemma:rho alpha} also yields the formula 
\begin{align}\label{Q alpha formal2}
Q_{\alpha,z}(\rho\|\sigma)=
\Tr \bz \oll{\rho^{\frac{\alpha}{2z}}\sigma^{\frac{1-\alpha}{2z}}} 
\bz\sigma^{\frac{1-\alpha}{2z}}\rho^{\frac{\alpha}{2z}} \jz\jz^{z},
\end{align}
which generalizes the finite-dimensional expression $\Tr\bz\rho^{\frac{\alpha}{2z}}\sigma^{\frac{1-\alpha}{z}}\rho^{\frac{\alpha}{2z}}\jz^z$.
Note that by Lemma \ref{lemma:rho alpha}, \eqref{Q alpha formal2} can also be written as
\begin{align*}
Q_{\alpha,z}(\rho\|\sigma)=\norm{\sigma^{\frac{1-\alpha}{2z}}\rho^{\frac{\alpha}{2z}} }_{2z}^{2z},
\end{align*}
where we use the notation $\norm{\valt}_z=(\Tr|\valt|^z)^{1/z}$ also for $z\in(0,1)$. 
\end{rem}

A further connection to the finite-dimensional formula is given by the following:

\begin{lemma}\label{lemma:Q alpha z as limit}
Let $\rho,\sigma\in\B(\hil)\p$ be such that $\rho^0\le\sigma^0$, and let 
$(\alpha,z)\in(1,+\infty)\times[1,+\infty)$. Then $\rho\in \L^{\alpha,z}(\hil,\sigma)$, or equivalently, $Q_{\alpha,z}(\rho\|\sigma)<+\infty)$, if and only if 
for any/some sequences 
$0<c_n<d_n$ with $c_n\to 0$, $d_n\to+\infty$,
$\bz\sigma_n^{\frac{1-\alpha}{2z}}\rho^{\frac{\alpha}{z}}
\sigma_n^{\frac{1-\alpha}{2z}}\jz_{n\in\bN}$ is a convergent sequence in 
$\L^{z}(\hil)$, 
where $\sigma_n:=\id_{(c_n,d_n)}(\sigma)$. 

Moreover, if $\rho\in \L^{\alpha,z}(\hil,\sigma)$ then 
\begin{align}\label{rho alpha as alpha limit}
\lim_{n\to+\infty}\norm{\rho_{\sigma,\alpha,z}-\sigma_n^{\frac{1-\alpha}{2z}}\rho^{\frac{\alpha}{z}}
\sigma_n^{\frac{1-\alpha}{2z}}}_z=0,
\end{align}
and if $\rho\in \B^{\alpha,z}(\hil,\sigma)$ then 
\begin{align}\label{rho alpha as alpha limit2}
Q_{\alpha,z}(\rho\|\sigma)
=
\lim_{n\to+\infty}\Tr\bz\sigma_n^{\frac{1-\alpha}{2z}}\rho^{\frac{\alpha}{z}}
\sigma_n^{\frac{1-\alpha}{2z}}\jz^z
\end{align}
for any sequences as above.
\end{lemma}
\begin{proof}
The ``if'' part follows 
since convergence in $z$-norm implies $\so$ convergence, 
whence $\rho_{\sigma,\alpha,z}$ exists as in Lemma \ref{lemma:rho alpha},
and the $\so$ limit coincides with the $z$-norm limit,
whence $\rho_{\sigma,\alpha,z}\in \L^z(\hil,\sigma)$. 

Assume now that $\rho\in \L^{\alpha,z}(\hil,\sigma)$. Then 
$\sigma_n^{\frac{1-\alpha}{2z}}\rho^{\frac{\alpha}{z}}
\sigma_n^{\frac{1-\alpha}{2z}}
=P_n\rho_{\sigma,\alpha,z}P_n$, with $P_n:=\egy_{(c_n,d_n)}(\sigma)$,
and the ``only if'' part, as well as \eqref{rho alpha as alpha limit}, 
%as well as \eqref{rho alpha as alpha limit} 
follows from 
Lemma \ref{lemma:Grumm}.

Note that \eqref{rho alpha as alpha limit} trivially implies 
\eqref{rho alpha as alpha limit2} when $\rho\in \L^{\alpha,z}(\hil,\sigma)$.
Assume thus that $\rho\in\B^{\alpha,z}(\hil,\sigma)\setminus\L^{\alpha,z}(\hil,\sigma)$, so that $Q_{\alpha,z}(\rho\|\sigma)=+\infty$.
Since $\sigma_n^{\frac{1-\alpha}{2z}}\rho
\sigma_n^{\frac{1-\alpha}{2z}}=P_n\rho_{\sigma,\alpha,z}P_n$ converges to 
$\rho_{\sigma,\alpha,z}$ in the weak operator topology,
Lemma \ref{lemma:pnorm lsc} yields that 
\begin{align*}
+\infty=Q_{\alpha,z}(\rho\|\sigma)=
\norm{\rho_{\sigma,\alpha,z}}_z^z
\le
\liminf_{n\to+\infty}\norm{\sigma_n^{\frac{1-\alpha}{2z}}\rho
\sigma_n^{\frac{1-\alpha}{2z}}}_z^z
=
\liminf_{n\to+\infty}\Tr\bz\sigma_n^{\frac{1-\alpha}{2z}}\rho
\sigma_n^{\frac{1-\alpha}{2z}}\jz^z,
\end{align*}
from which \eqref{rho alpha as alpha limit2} follows.
\end{proof}

\begin{prop}\label{prop:z mon}
For any $\rho,\sigma\in\B(\hil)\p$, and any $\alpha\in(1,+\infty)$, 
\begin{align*}
Q_{\alpha,z}(\rho\|\sigma),\s D_{\alpha,z}(\rho\|\sigma),\s
\tilde D_{\alpha,z}(\rho\|\sigma)\ds\ds\text{are decreasing in }\ds z.
%z\mapsto Q_{\alpha,z}(\rho\|\sigma)\ds\text{and}\ds
%z\mapsto D_{\alpha,z}(\rho\|\sigma)\ds\text{are decreasing.}
\end{align*}
In particular, for any $\sigma\in\B(\hil)\p$, $\L^{\alpha,z}(\hil,\sigma)$ is increasing in $z$, i.e., 
\begin{align*}
0<z\le z'\ds\imp\ds
\L^{\alpha,z}(\hil,\sigma)\subseteq \L^{\alpha,z'}(\hil,\sigma).
\end{align*}
\end{prop}
\begin{proof}
It is sufficient to prove that for any $0<z<z'$, 
$Q_{\alpha,z}(\rho\|\sigma)\ge Q_{\alpha,z'}(\rho\|\sigma)$ holds.
This is obvious when $Q_{\alpha,z}(\rho\|\sigma)=+\infty$, and hence for the rest we assume 
the contrary, i.e., that $\rho\in\L^{\alpha,z}(\hil,\sigma)$. 
By Lemma \ref{lemma:B alpha z increasing}, this implies that 
$\rho\in\B^{\alpha,z'}(\hil,\sigma)$. Thus, by Lemma \ref{lemma:Q alpha z as limit},
\begin{align}\label{z mon proof1}
Q_{\alpha,z'}(\rho\|\sigma)
=
\lim_{n\to+\infty}\Tr\bz\sigma_n^{\frac{1-\alpha}{2z'}}\rho^{\frac{\alpha}{z'}}
\sigma_n^{\frac{1-\alpha}{2z'}}\jz^{z'}.
\end{align}
According to Araki's inquality \cite[Theorem]{Araki},
$\Tr\vfi\bz B^{1/2}AB^{1/2}\jz^q\le\Tr\vfi\bz B^{q/2}AB^{q/2}\jz$ for any 
$A,B\in\B(\hil)\p$, $q\in[1,+\infty)$, and monotone increasing continuous function $\vfi$ on $[0,+\infty)$
such that $\vfi(0)=0$ and $t\mapsto \vfi(e^t)$ is convex on $\bR$. 
Applying this to $A:=\rho^{\frac{\alpha}{z'}}$, $B:=\sigma_n^{\frac{1-\alpha}{z'}}$, 
$q:=\frac{z'}{z}$, and $\vfi:=\id_{[0,+\infty)}^{z}$ yields
\begin{align*}
\Tr\bz\sigma_n^{\frac{1-\alpha}{2z'}}\rho^{\frac{\alpha}{z'}}
\sigma_n^{\frac{1-\alpha}{2z'}}\jz^{z'}=
\Tr\left[\bz\sigma_n^{\frac{1-\alpha}{2z'}}\rho^{\frac{\alpha}{z'}}
\sigma_n^{\frac{1-\alpha}{2z'}}\jz^{\frac{z'}{z}}\right]^z
\le
\Tr\bz\sigma_n^{\frac{1-\alpha}{2z}}\rho^{\frac{\alpha}{z}}
\sigma_n^{\frac{1-\alpha}{2z}}\jz^{z}
\end{align*}
for every $n\in\bN$. Thus, by \eqref{z mon proof1}, 
\begin{align*}
Q_{\alpha,z'}(\rho\|\sigma)
\le
\lim_{n\to+\infty}\Tr\bz\sigma_n^{\frac{1-\alpha}{2z}}\rho^{\frac{\alpha}{z}}
\sigma_n^{\frac{1-\alpha}{2z}}\jz^{z}
=
Q_{\alpha,z}(\rho\|\sigma),
\end{align*}
where the equality is again due to Lemma \ref{lemma:Q alpha z as limit}.
\end{proof}

\begin{rem}
As a special case of Proposition \ref{prop:z mon}, we get that for any $\rho,\sigma\in\B(\hil)\p$,
\begin{align*}
D_{\alpha}\nw(\rho\|\sigma)= D_{\alpha,\alpha}(\rho\|\sigma)
\le
D_{\alpha,1}(\rho\|\sigma),
\end{align*}
i.e., the sandwiched R\'enyi $\alpha$-divergence cannot be larger than the 
Petz-type R\'enyi $\alpha$-divergence. 
This has been proved for positive normal functionals on a von Neumann algebra
(positive trace-class operators in our case) in 
\cite[Theorem 12]{BST} and \cite[Corollary 3.6]{Jencova_NCLp}
using different methods than in the proof of Proposition \ref{prop:z mon} above.
\end{rem}

\begin{rem}\label{rem:smaller alpha finite}
Assume that $Q_{\alpha}^*(\rho\|\sigma)<+\infty$, i.e., $\rho\in\L^{\alpha}(\rho\|\sigma)$ for some $\alpha>1$, and $1<\alpha'<\alpha$. Then, by \eqref{smaller alpha},
\begin{align*}
Q_{\alpha'}\nw(\rho\|\sigma)&=\Tr\rho_{\sigma,\alpha'}^{\alpha'}
=
\norm{\rho_{\sigma,\alpha'}}_{\alpha'}^{\alpha'}
\le
\norm{\sigma^{\frac{1}{2\alpha'}-\frac{1}{2\alpha}}}_{\frac{2\alpha\alpha'}{\alpha-\alpha'}}^{\alpha'}
\norm{\rho_{\sigma,\alpha}}_{\alpha}^{\alpha'}
\norm{\sigma^{\frac{1}{2\alpha'}-\frac{1}{2\alpha}}}_{\frac{2\alpha\alpha'}{\alpha-\alpha'}}^{\alpha'}
=
(\Tr\sigma)^{1-\frac{\alpha'}{\alpha}}Q_{\alpha}\nw(\rho\|\sigma)^{\frac{\alpha'}{\alpha}},
\end{align*}
where the inequality follows by the operator H\"older inequality. 
In particular, if $\sigma$ is trace-class then $Q_{\alpha'}\nw(\rho\|\sigma)<+\infty$. 
If $\Tr\sigma=1$ then a simple rearrangement yields
\begin{align*}
\frac{\alpha'-1}{\alpha'}D_{\alpha'}\nw(\rho\|\sigma)\le
\frac{\alpha-1}{\alpha}D_{\alpha}\nw(\rho\|\sigma).
\end{align*}
Note that this is weaker than $D_{\alpha'}\nw(\rho\|\sigma)\le D_{\alpha}\nw(\rho\|\sigma)$,
which was proved in \cite[Proposition 3.7]{Jencova_NCLpII}.
\end{rem}

\begin{rem}
Since we do not assume the second operator to be trace-class, the expression 
$-D_{\alpha,z}(\rho\|I)$ makes sense, and we recover the following identity for the R\'enyi $\alpha$-entropy of a state $\rho\in\S(\hil)$, which is well-known in the finite-dimensional case:
\begin{align}\label{Renyi ent}
S_{\alpha}(\rho):=\frac{1}{1-\alpha}\log\Tr\rho^{\alpha}=-D_{\alpha,z}(\rho\|I),\ds\ds\ds
\alpha>1.
\end{align}
(In fact, this makes sense for arbitrary PSD operator $\rho$).

More importantly, allowing non trace-class operators enables the definition of conditional 
$(\alpha,z)$-entropies. Following \cite{BTH_Renyicond}, 
we define two different notions of conditional $(\alpha,z)$-entropy between systems 
$A$ and $B$ in a state $\rho_{AB}\in\S(\hil_A\otimes\hil_B)$ as 
\begin{align}
S_{\alpha,z}(A|B)^{\downarrow}&:=-D_{\alpha,z}(\rho_{AB}\|I_A\otimes\rho_B),
\label{Renyi cond1}\\
S_{\alpha,z}(A|B)^{\uparrow}&:=-\inf_{\omega_B\in\S(\hil_B)}D_{\alpha,z}(\rho_{AB}\|I_A\otimes\omega_B),
\label{Renyi cond2}
\end{align}
where $\rho_B=\Tr_A\rho_{AB}$ denotes the marginal of $\rho_{AB}$ on system $B$. 
Again, \eqref{Renyi cond1}--\eqref{Renyi cond2} make sense even when $\rho_{AB}$ is only assumed to be PSD. 
Note that while the R\'enyi entropies \eqref{Renyi ent} can be defined directly for $\rho$ 
without reference to any R\'enyi divergences, this is not the case for the conditional R\'enyi 
entropies \eqref{Renyi cond1}--\eqref{Renyi cond2}, and the ability to take non-trace-class operators at least in the second argument of the divergence is crucial for the definition.

According to Proposition \ref{prop:z mon}, for any fixed $\rho_{AB}\in\S(\hil_{A}\otimes\hil_B)$, and any $\alpha>1$,
\begin{align*}
S_{\alpha,z}(A|B)^{\downarrow}\ds\text{and}\ds
S_{\alpha,z}(A|B)^{\uparrow}\ds\text{are monotone increasing in}\ds z.
\end{align*}
In particular, either version of the sandwiched conditional R\'enyi entropy is at least as large as the corresponding version of the Petz-type conditional R\'enyi entropy.
\end{rem}

\begin{lemma}\label{lemma:Q poz}
For any $\rho,\sigma\in\B(\hil)\p$ and  $(\alpha,z)\in(1,+\infty)\times(0,+\infty)$,  
\begin{align}\label{Q poz}
Q_{\alpha,z}(\rho\|\sigma)>0,\ds\ds\ds
D_{\alpha,z}(\rho\|\sigma)>-\infty.
\end{align}
\end{lemma}
\begin{proof}
The assertion is trivial when $\rho\notin\L^{\alpha,z}(\hil,\sigma)$, and hence we assume the contrary. Then 
\begin{align*}
Q_{\alpha,z}(\rho\|\sigma)=\Tr\rho_{\sigma,\alpha,z}^z=0\ds\iff\ds
\rho_{\sigma,\alpha,z}=0\ds\imp\ds
\rho^{\frac{\alpha}{z}}=\sigma^{\frac{\alpha-1}{2z}}\rho_{\sigma,\alpha,z}\sigma^{\frac{\alpha-1}{2z}}=0,
\end{align*}
contrary to the assumption that $\rho\in\B(\hil)\p$.
Hence, the inequalities in \eqref{Q poz} hold.
\end{proof}

\begin{rem}
Stronger bounds than the ones in \eqref{Q poz} are given below in 
Corollary \ref{cor:trace mon} for trace-class operators.
\end{rem}

\begin{lemma}\label{lemma:tensor product}
Let $\rho_k,\sigma_k\in\B(\hil_k)\p$, $k=1,2$. For any $(\alpha,z)\in(1,+\infty)\times(0,+\infty)$,  
\begin{align}
&\rho_1\otimes\rho_2\in\B^{\alpha,z}(\hil_1\otimes\hil_2,\sigma_1\otimes\sigma_2)\ds\iff\ds
\rho_k\in\B^{\alpha,z}(\hil_k,\sigma_k),\ds\ds k=1,2,
\label{tensor product1}\\
&\rho_1\otimes\rho_2\in\L^{\alpha,z}(\hil_1\otimes\hil_2,\sigma_1\otimes\sigma_2)\ds\iff\ds
\rho_k\in\L^{\alpha,z}(\hil_k,\sigma_k),\ds\ds k=1,2,\label{tensor product2}
\end{align}
and $(\rho_1\otimes\rho_2)_{\sigma_1\otimes\sigma_2,\alpha,z}=
(\rho_1)_{\sigma_1,\alpha,z}\otimes(\rho_2)_{\sigma_2,\alpha,z}$.
As a consequence,
\begin{align}
&Q_{\alpha,z}(\rho_1\otimes\rho_2\|\sigma_1\otimes\sigma_2)=
Q_{\alpha,z}(\rho_1\|\sigma_1)Q_{\alpha,z}(\rho_2\|\sigma_2),\label{Q tensor}\\
&D_{\alpha,z}(\rho_1\otimes\rho_2\|\sigma_1\otimes\sigma_2)=
D_{\alpha,z}(\rho_1\|\sigma_1)+D_{\alpha,z}(\rho_2\|\sigma_2).\label{D tensor}
\end{align}
\end{lemma}
\begin{proof}
The right to left implications in \eqref{tensor product1}-\eqref{tensor product2} are obvious 
from choosing $R:=
(\rho_1)_{\sigma_1,\alpha,z}\otimes(\rho_2)_{\sigma_2,\alpha,z}$
in \ref{rho alpha1} of Lemma \ref{lemma:rho alpha}.
Assume that $\rho_1\otimes\rho_2\in\B^{\alpha,z}(\hil_1\otimes\hil_2,\sigma_1\otimes\sigma_2)$. By \ref{rho alpha6} of Lemma \ref{lemma:rho alpha}, there exists a $\lambda\ge 0$ such that 
\begin{align*}
\rho_1^{\frac{\alpha}{z}}\otimes\rho_2^{\frac{\alpha}{z}}
=
\bz\rho_1\otimes\rho_2\jz^{\frac{\alpha}{z}}
\le
\lambda\bz\sigma_1\otimes\sigma_2\jz^{\frac{\alpha-1}{z}}
=
\lambda\sigma_1^{\frac{\alpha-1}{z}}\otimes\sigma_2^{\frac{\alpha-1}{z}}.
\end{align*}
Choose any $\psi_2\notin\ker(\rho_2)$. For any $\psi_1\in\hil_1$, we get
\begin{align*}
\inner{\psi_1}{\rho_1^{\frac{\alpha}{z}}\psi_1}
\underbrace{\inner{\psi_2}{\rho_2^{\frac{\alpha}{z}}\psi_2}}_{=:\kappa_1>0}
\le
\lambda
\inner{\psi_1}{\sigma_1^{\frac{\alpha-1}{z}}\psi_1}
\underbrace{\inner{\psi_2}{\sigma_2^{\frac{\alpha-1}{z}}\psi_2}}_{=:\kappa_2}.
\end{align*}
Thus, $\rho_1^{\frac{\alpha}{z}}\le\lambda(\kappa_2/\kappa_1)\sigma_1^{\frac{\alpha-1}{z}}$, and again by \ref{rho alpha6} of Lemma \ref{lemma:rho alpha}, 
$\rho_1\in\B^{\alpha,z}(\hil_1,\sigma_1)$. An exactly analogous argument gives
$\rho_2\in\B^{\alpha,z}(\hil_2,\sigma_2)$.
This proves the left to right implication in \eqref{tensor product1},
and we also get 
\begin{align*}
(\sigma_1\otimes\sigma_2)^{\frac{1-\alpha}{2z}}
(\rho_1\otimes\rho_2)^{\frac{\alpha}{2z}}
=
\bz\sigma_1^{\frac{1-\alpha}{2z}}\rho_1^{\frac{\alpha}{2z}}\jz
\otimes
\bz\sigma_2^{\frac{1-\alpha}{2z}}\rho_2^{\frac{\alpha}{2z}}\jz,
\end{align*}
from which $(\rho_1\otimes\rho_2)_{\sigma_1\otimes\sigma_2,\alpha,z}=
(\rho_1)_{\sigma_1,\alpha,z}\otimes(\rho_2)_{\sigma_2,\alpha,z}$, according to 
\eqref{rho alpha eq1}, and thus \eqref{Q tensor} and \eqref{D tensor} follow
due to the multiplicativity of the trace.
The left to right implication in \eqref{tensor product2}
follows immediately from the above.
\end{proof}

\subsection{Variational formulas}

The following variational representations of $Q_{\alpha,z}$ and $D_{\alpha,z}$ are very useful to establish their fundamental properties. 
We will use these variational formulas to prove monotonicity of $Q_{\alpha,z}$ under restrictions of the operators to subspaces (Lemma \ref{lemma:projection monotonicity}, Corollary \ref{cor:increasing sequence}) and to give a lower bound on the strong converse exponent (Lemma \ref{lemma:sc optimality}).

For $z=\alpha$ (the case of the sandwiched R\'enyi divergence), the variational formula 
in \eqref{Q variational} was given first in 
\cite{FL} for finite-dimensional PSD operators, and was extended 
to the case of pairs of positive normal functionals on a general von Neumann algebra in \cite{Jencova_NCLpII} (see also \cite[Lemma 3.19]{Hiai_fdiv_Springer} for the case $\alpha<1$), while the variational formula in 
\eqref{D variational} can be obtained as an intermediate step 
in the proof of the first variational formula, and it was given in 
\cite{BFT_variational} in the finite-dimensional case.

For finite-dimensional invertible PSD operators and arbitrary 
$(\alpha,z)\in(1,+\infty)\times(0,+\infty)$,
both variational formulas \eqref{Q variational}--\eqref{D variational}
follow as special cases of \cite[Theorem 3.3]{Zhang2018}.

The version below is an extension of the above when 
$(\alpha,z)\in(1,+\infty)\times(0,+\infty)$ are arbitrary, 
and the operators $\rho,\sigma$ can be PSD operators on an infinite-dimensional Hilbert space 
satisfying the conditions in Lemma \ref{lemma:variational}.
Our proof follows essentially that of \cite[Theorem 3.3]{Zhang2018}.

For any $\sigma\in\B(\hil)\p$, let 
\begin{align*}
\B(\hil)_{\sigma,\alpha,z}&:=
\left\{H\in\B(\hil)_{\ge 0}:\,
\Tr\bz H^{1/2}\sigma^{\frac{\alpha-1}{z}}H^{1/2}\jz^{\frac{z}{\alpha-1}}<+\infty\right\},\\
\B(\hil)_{\sigma,\alpha,z}^+&:=
\left\{H\in\B(\hil)_{\ge 0}:\,
0<\Tr\bz H^{1/2}\sigma^{\frac{\alpha-1}{z}}H^{1/2}\jz^{\frac{z}{\alpha-1}}<+\infty\right\}.
\end{align*}

\begin{lemma}\label{lemma:variational}
Let $\rho,\sigma\in\B(\hil)\p$ and $(\alpha,z)\in(1,+\infty)\times(0,+\infty)$, and assume 
that one of the following holds:
a) $\rho\notin\B^{\alpha,z}(\hil,\sigma)$;
b) $\rho\in \B^{\alpha,z}(\hil,\sigma)\setminus \L^{\alpha,z}(\hil,\sigma)$ and $\sigma$ is compact;
c) $\rho\in \L^{\alpha,z}(\hil,\sigma)$. Then 
\begin{align}
Q_{\alpha,z}(\rho\|\sigma)
&=
\sup_{H\in\B(\hil)_{\sigma,\alpha,z}}\left\{\alpha
\Tr \bz H^{1/2}\rho^{\frac{\alpha}{z}} H^{1/2}\jz^{\frac{z}{\alpha}}+
(1-\alpha)
\Tr\bz H^{1/2}\sigma^{\frac{\alpha-1}{z}}H^{1/2}\jz^{\frac{z}{\alpha-1}}\right\},
\label{Q variational}\\
\log Q_{\alpha,z}(\rho\|\sigma)
&=
\sup_{H\in\B(\hil)_{\sigma,\alpha,z}^+}\left\{
\alpha\log\Tr \bz H^{1/2}\rho^{\frac{\alpha}{z}} H^{1/2}\jz^{\frac{z}{\alpha}}
+(1-\alpha)\log\Tr\bz H^{1/2}\sigma^{\frac{\alpha-1}{z}}H^{1/2}\jz^{\frac{z}{\alpha-1}}\right\}.\label{D variational}
\end{align}
The equality in \eqref{Q variational} still holds if the supremum is taken over 
$\B(\hil)_{\sigma,\alpha,z}^+$. Moreover, in cases a) and b), and in case c) if $\sigma$ is compact, the $H$ operators in \eqref{Q variational} and \eqref{D variational} 
may additionally be required to be of finite rank.
\end{lemma}

\begin{proof}
For any $H\in\B(\hil)\p$, let 
\begin{align}
F(H)&:=\Tr \bz H^{1/2}\rho^{\frac{\alpha}{z}} H^{1/2}\jz^{\frac{z}{\alpha}},\ds\ds\ds
G(H):=\Tr\bz H^{1/2}\sigma^{\frac{\alpha-1}{z}}H^{1/2}\jz^{\frac{z}{\alpha-1}}.
\end{align}

Assume first that $\rho\notin\B^{\alpha,z}(\hil,\sigma)$, and hence 
$Q_{\alpha,z}(\rho\|\sigma)=\log Q_{\alpha,z}(\rho\|\sigma)=+\infty$.
By \ref{rho alpha6} of Lemma 
\ref{lemma:rho alpha}, for every $\lambda>0$ there exists a vector 
$x_{\lambda}\in\hil$ such that 
\begin{align}\label{variational proof0}
\inner{x_{\lambda}}{\rho^{\frac{\alpha}{z}}x_{\lambda}}>
\lambda\langle x_{\lambda},\sigma^{\frac{\alpha-1}{z}}x_{\lambda}\rangle.
\end{align}
Clearly, for any $x\in\hil$, $H_x:=\pr{x}\in\B(\hil)_{\sigma,\alpha,z}\cap\B_f(\hil)$, and 
\begin{align*}
%\Tr \bz H_{x}^{1/2}\rho^{\frac{\alpha}{z}} H_{x}^{1/2}\jz^{\frac{z}{\alpha}}
F(H_x)=
\inner{x}{\rho^{\frac{\alpha}{z}}x}^{\frac{z}{\alpha}},\ds\ds\ds
%\Tr\bz H_{x}^{1/2}\sigma^{\frac{\alpha-1}{z}}H_{x}^{1/2}\jz^{\frac{z}{\alpha-1}}
G(H_x)=
\langle x,\sigma^{\frac{\alpha-1}{z}}x\rangle^{\frac{z}{\alpha-1}}.
\end{align*}

If $\langle x_{\lambda},\sigma^{\frac{\alpha-1}{z}}x_{\lambda}\rangle=0$ for some 
$\lambda>0$ then let $x_{\lambda,t}:=tx_{\lambda}+t\inv y$, $t>0$, where 
$y\in(\ker\sigma)^{\perp}\setminus\{0\}$ is some fixed vector. 
Then $H_{x_{\lambda,t}}\in\B(\hil)_{\sigma,\alpha,z}^+\cap\B_f(\hil)$, 
\eqref{variational proof0} implies that 
$\inner{x_{\lambda}}{\rho^{\frac{\alpha}{z}}x_{\lambda}}>0$,
and 
it is straightforward to verify that 
\begin{align*}
\lim_{t\to+\infty}F(H_{x_{\lambda,t}})=+\infty,\ds\ds
\lim_{t\to+\infty}G(H_{x_{\lambda,t}})=0.
\end{align*}
Thus, 
\begin{align*}
\lim_{t\to+\infty}\bz \alpha F(H_{x_{\lambda,t}})+(1-\alpha)G(H_{x_{\lambda,t}})\jz
&=+\infty=Q_{\alpha,z}(\rho\|\sigma)=\log Q_{\alpha,z}(\rho\|\sigma)\\
&=\lim_{t\to+\infty}\bz \alpha \log F(H_{x_{\lambda,t}})+(1-\alpha)\log G(H_{x_{\lambda,t}})\jz, 
\end{align*}
and therefore \eqref{Q variational}--\eqref{D variational} hold.

If $\langle x_{\lambda},\sigma^{\frac{\alpha-1}{z}}x_{\lambda}\rangle>0$ for every 
$\lambda>0$ then let $\tilde x_{\lambda}:=x_{\lambda}
\langle x_{\lambda},\sigma^{\frac{\alpha-1}{z}}x_{\lambda}\rangle^{-1/2}$. 
Then
\begin{align}\label{variational proof00}
\langle \tilde x_{\lambda},\sigma^{\frac{\alpha-1}{z}}\tilde x_{\lambda}\rangle
=1=
\langle \tilde x_{\lambda},\sigma^{\frac{\alpha-1}{z}}\tilde x_{\lambda}\rangle
^{\frac{z}{\alpha-1}}=G(H_{\tilde x_{\lambda}}),
\end{align}
and
\begin{align*}
F(H_{\tilde x_{\lambda}})
=
\inner{\tilde x_{\lambda}}{\rho^{\frac{\alpha}{z}}\tilde x_{\lambda}}^{\frac{z}{\alpha}}
>\bz\lambda\langle \tilde x_{\lambda},\sigma^{\frac{\alpha-1}{z}}\tilde x_{\lambda}\rangle\jz^{\frac{z}{\alpha}}
=\lambda^{\frac{z}{\alpha}},
\end{align*}
according to \eqref{variational proof0} and \eqref{variational proof00}.
Thus, 
\begin{align*}
\lim_{\lambda\to+\infty}\bz \alpha F(H_{\tilde x_{\lambda}})+
(1-\alpha)G(H_{\tilde x_{\lambda}})\jz
&=+\infty=Q_{\alpha,z}(\rho\|\sigma)=\log Q_{\alpha,z}(\rho\|\sigma)\\
&=\lim_{\lambda\to+\infty}\bz \alpha \log F(H_{\tilde x_{\lambda}})+
(1-\alpha)\log G(H_{\tilde x_{\lambda}})\jz, 
\end{align*}
and therefore \eqref{Q variational}--\eqref{D variational} hold, even with the 
optimizations restricted to finite-rank operators.

This completes the proof of case a), and hence for the rest we assume that 
b) or c) holds.

Consider any sequences $0<c_n<d_n$ with $c_n\to 0$, $d_n\to+\infty$, 
and let $P_n:=\egy_{(c_n,d_n)}(\sigma)$,
$\sigma_n:=\id_{(c_n,d_n)}(\sigma)=P_n\sigma P_n$, and
\begin{align*}
H_n:=\sigma_n^{\frac{1-\alpha}{2z}}
(\sigma_n^{\frac{1-\alpha}{2z}}
\rho^{\frac{\alpha}{z}}\sigma_n^{\frac{1-\alpha}{2z}})^{\alpha-1}
\sigma_n^{\frac{1-\alpha}{2z}}
=
\sigma_n^{\frac{1-\alpha}{2z}}(P_n\rho_{\sigma,\alpha,z}P_n)^{\alpha-1}\sigma_n^{\frac{1-\alpha}{2z}}.
\end{align*}
Then 
\begin{align}
F(H_n)&=\Tr \bz H_n^{1/2}\rho^{\frac{\alpha}{z}} H_n^{1/2}\jz^{\frac{z}{\alpha}}\nn\\
&=
\Tr \bz \rho^{\frac{\alpha}{2z}}H_n\rho^{\frac{\alpha}{2z}} \jz^{\frac{z}{\alpha}}\nn\\
&=
\Tr \bz \rho^{\frac{\alpha}{2z}}\sigma_n^{\frac{1-\alpha}{2z}}
(\sigma_n^{\frac{1-\alpha}{2z}}
\rho^{\frac{\alpha}{z}}\sigma_n^{\frac{1-\alpha}{2z}})^{\alpha-1}\sigma_n^{\frac{1-\alpha}{2z}}\rho^{\frac{\alpha}{2z}} \jz^{\frac{z}{\alpha}}\nn\\
&=
\Tr \Big( 
(\sigma_n^{\frac{1-\alpha}{2z}}
\rho^{\frac{\alpha}{z}}\sigma_n^{\frac{1-\alpha}{2z}})^{\frac{\alpha-1}{2}}
\sigma_n^{\frac{1-\alpha}{2z}}
\rho^{\frac{\alpha}{z}}\sigma_n^{\frac{1-\alpha}{2z}}
(\sigma_n^{\frac{1-\alpha}{2z}}
\rho^{\frac{\alpha}{z}}\sigma_n^{\frac{1-\alpha}{2z}})^{\frac{\alpha-1}{2}}
\Big)^{\frac{z}{\alpha}}\nn\\
&=
\Tr \Big(\sigma_n^{\frac{1-\alpha}{2z}}
\rho^{\frac{\alpha}{z}}\sigma_n^{\frac{1-\alpha}{2z}}\Big)^{z}\nn\\
&=
\Tr\bz P_n\rho_{\sigma,\alpha,z}P_n\jz^z,\label{variational proof1}
\end{align}
and similarly,
\begin{align}
G(H_n)&=\Tr\bz H_n^{1/2}\sigma^{\frac{\alpha-1}{z}}H_n^{1/2}\jz^{\frac{z}{\alpha-1}}
=
\Tr\bz\sigma^{\frac{\alpha-1}{2z}}H_n\sigma^{\frac{\alpha-1}{2z}}\jz^{\frac{z}{\alpha-1}}
\nn\\
&=
\Tr\Big(
\underbrace{\sigma^{\frac{\alpha-1}{2z}}
\sigma_n^{\frac{1-\alpha}{2z}}}_{=P_n}
(\sigma_n^{\frac{1-\alpha}{2z}}
\rho^{\frac{\alpha}{z}}\sigma_n^{\frac{1-\alpha}{2z}})^{\alpha-1}
\underbrace{\sigma_n^{\frac{1-\alpha}{2z}}
\sigma^{\frac{\alpha-1}{2z}}}_{=P_n}
\Big)^{\frac{z}{\alpha-1}}\nn\\
&=
\Tr \Big(\sigma_n^{\frac{1-\alpha}{2z}}
\rho^{\frac{\alpha}{z}}\sigma_n^{\frac{1-\alpha}{2z}}\Big)^{z}\nn\\
&=
\Tr(P_n\rho_{\sigma,\alpha,z}P_n)^z.\label{variational proof2}
\end{align}

We have
\begin{align*}
\Tr(P_n\rho_{\sigma,\alpha,z}P_n)^z=
\Tr\bz\rho_{\sigma,\alpha,z}^{1/2}P_n\rho_{\sigma,\alpha,z}^{1/2}\jz^z
\le
\Tr\rho_{\sigma,\alpha,z}^z
=
Q_{\alpha,z}(\rho\|\sigma); 
\end{align*}
in particular, 
\begin{align}\label{variational proof3}
\limsup_{n\to+\infty}\Tr(P_n\rho_{\sigma,\alpha,z}P_n)^z\le Q_{\alpha,z}(\rho\|\sigma).
\end{align}
Moreover, if $Q_{\alpha,z}(\rho\|\sigma)<+\infty$, i.e., in case c), or if 
$\sigma$ is compact (in which case $H_n$ and $P_n\rho_{\sigma,\alpha,z}P_n$ are of 
finite rank) then 
$F(H_n)=G(H_n)<+\infty$, whence
$H_n\in\B(\hil)_{\sigma,\alpha,z}$.
Since $\rho\in\B^{\alpha,z}(\hil,\sigma)$ implies $\rho^0\le\sigma^0$, it is also true that 
$0\ne \sigma_n^{\frac{1-\alpha}{2z}}
\rho^{\frac{\alpha}{z}}\sigma_n^{\frac{1-\alpha}{2z}}$, and hence
$H_n\in\B(\hil)_{\sigma,\alpha,z}^+$, for all large enough $n$.

When $z\ge 1$, Lemma \ref{lemma:pnorm lsc} yields
\begin{align}\label{variational proof4}
Q_{\alpha,z}(\rho\|\sigma)
=
\Tr\rho_{\sigma,\alpha,z}^z=\norm{\rho_{\sigma,\alpha,z}}_z^z
\le
\liminf_{n\to+\infty}\norm{P_n\rho_{\sigma,\alpha,z}P_n}_z^z
=
\liminf_{n\to+\infty}
\Tr(P_n\rho_{\sigma,\alpha,z}P_n)^z.
\end{align}
When $z\in(0,1)$, $\id_{[0,+\infty)}^z$ is operator concave, and hence
$(P_n\rho_{\sigma,\alpha,z}P_n)^z\ge P_n\rho_{\sigma,\alpha,z}^zP_n$, whence
\begin{align}\label{variational proof5}
\liminf_{n\to+\infty}\Tr(P_n\rho_{\sigma,\alpha,z}P_n)^z
\ge
\liminf_{n\to+\infty}\Tr P_n\rho_{\sigma,\alpha,z}^zP_n=\Tr \rho_{\sigma,\alpha,z}^z=
Q_{\alpha,z}(\rho\|\sigma).
\end{align}
Combining \eqref{variational proof1}--\eqref{variational proof5} gives
\begin{align*}
Q_{\alpha,z}(\rho\|\sigma)
&=
\lim_{n\to+\infty}\Tr \Big(\sigma_n^{\frac{1-\alpha}{2z}}
\rho^{\frac{\alpha}{z}}\sigma_n^{\frac{1-\alpha}{2z}}\Big)^{z}\\
&=
\lim_{n\to+\infty}\bz
\alpha\Tr \bz H_n^{1/2}\rho^{\frac{\alpha}{z}} H_n^{1/2}\jz^{\frac{z}{\alpha}}+
(1-\alpha)
\Tr\bz H_n^{1/2}\sigma^{\frac{\alpha-1}{z}}H_n^{1/2}\jz^{\frac{z}{\alpha-1}}\jz,\\
\log Q_{\alpha}\nw(\rho\|\sigma)
&=
\lim_{n\to+\infty}\log\Tr \Big(\sigma_n^{\frac{1-\alpha}{2z}}
\rho^{\frac{\alpha}{z}}\sigma_n^{\frac{1-\alpha}{2z}}\Big)^{z}\\
&=
\lim_{n\to+\infty}\bz
\alpha\log\Tr \bz H_n^{1/2}\rho^{\frac{\alpha}{z}} H_n^{1/2}\jz^{\frac{z}{\alpha}}+
(1-\alpha)
\log\Tr\bz H_n^{1/2}\sigma^{\frac{\alpha-1}{z}}H_n^{1/2}\jz^{\frac{z}{\alpha-1}}\jz.
\end{align*}
This completes the proof when $Q_{\alpha,z}(\rho\|\sigma)=+\infty$, i.e., 
in case b). 

Assume for the rest that case c) holds. 
By the above considerations, we have LHS$\le$RHS in \eqref{Q variational}--\eqref{D variational}, and hence we only have to show the converse inequalities.
By Lemma \ref{lemma:rho alpha} and Definition \ref{def:L alpha},
$\rho^{\frac{\alpha}{z}}=\sigma^{\frac{\alpha-1}{2z}}\rho_{\sigma,\alpha,z}\sigma^{\frac{\alpha-1}{2z}}$
with $\rho_{\sigma,\alpha,z}\in \L^{\alpha,z}(\hil)$.
For any $H\in\B(\hil)_{\sigma,\alpha,z}$, we have 
\begin{align}
\Tr \bz H^{1/2}\rho^{\frac{\alpha}{z}} H^{1/2}\jz^{\frac{z}{\alpha}}
&=
\Tr \bz \rho^{\frac{\alpha}{2z}}H\rho^{\frac{\alpha}{2z}}\jz^{\frac{z}{\alpha}}
=
\Tr\abs{H^{1/2}\rho^{\frac{\alpha}{2z}}}^{\frac{2z}{\alpha}}
=
\Tr\abs{H^{1/2}\sigma^{\frac{\alpha-1}{2z}}\sigma^{\frac{1-\alpha}{2z}}\rho^{\frac{\alpha}{2z}}}^{\frac{2z}{\alpha}}\label{variational lower bound1}
\\
&=
\norm{H^{1/2}\sigma^{\frac{\alpha-1}{2z}}\sigma^{\frac{1-\alpha}{2z}}\rho^{\frac{\alpha}{2z}}}_{\frac{2z}{\alpha}}^{\frac{2z}{\alpha}}
\le
\norm{H^{1/2}\sigma^{\frac{\alpha-1}{2z}}}_{\frac{2z}{\alpha-1}}^{\frac{2z}{\alpha}}
\norm{\sigma^{\frac{1-\alpha}{2z}}\rho^{\frac{\alpha}{2z}}}_{2z}^{\frac{2z}{\alpha}}
\label{variational lower bound2}\\
&=
\left[\Tr\bz\sigma^{\frac{\alpha-1}{2z}}H \sigma^{\frac{\alpha-1}{2z}}\jz^{\frac{z}{\alpha-1}}\right]^{\frac{\alpha-1}{\alpha}}
\Bigg[\underbrace{\Tr\bz 
\bz \sigma^{\frac{1-\alpha}{2z}}\rho^{\frac{\alpha}{2z}}\jz^*
\sigma^{\frac{1-\alpha}{2z}}\rho^{\frac{\alpha}{2z}}
\jz^{z}}_{=Q_{\alpha,z}(\rho\|\sigma)}\Bigg]^{\frac{1}{\alpha}}
\label{variational lower bound3}\\
&\le
\frac{\alpha-1}{\alpha}\Tr\bz\sigma^{\frac{\alpha-1}{2z}}H \sigma^{\frac{\alpha-1}{2z}}\jz^{\frac{z}{\alpha-1}}
+
\frac{1}{\alpha}Q_{\alpha,z}(\rho\|\sigma),\label{variational lower bound4}
\end{align}
where we used that 
$\ran\rho^{\frac{\alpha}{2z}}\subseteq
\dom \sigma^{\frac{1-\alpha}{2z}}$ and $\sigma^{\frac{1-\alpha}{2z}}\rho^{\frac{\alpha}{2z}}\in\B(\hil)$, according to Lemma \ref{lemma:rho alpha}, and the 
expression in \eqref{Q alpha formal2} for $Q_{\alpha,z}(\rho\|\sigma)$.
The first inequality above is due to the operator H\"older inequality, and 
the second inequality is trivial from the convexity of the exponential function.
A simple rearrangement yields that LHS$\ge$RHS in \eqref{Q variational}--\eqref{D variational}, completing the proof.
\end{proof}

\begin{rem}
It is interesting that
one can formally take the logarithm of each term in \eqref{Q variational} to obtain \eqref{D variational}.
\end{rem}

\begin{rem}
The variational formulas in \eqref{Q variational}--\eqref{D variational} hold 
for the sandwiched quantities ($z=\alpha$)
when $\rho\in \B^{\alpha}(\hil,\sigma)\setminus \L^{\alpha}(\hil,\sigma)$
even if $\sigma$ is not compact \cite{Hiai_private}. However, we won't need this fact in the rest of the paper.
\end{rem}

\begin{rem}
Note that the case $z=1$ corresponds to the Petz-type R\'enyi divergences. 
By the above,
$\rho\in\B^{\alpha,1}(\hil,\sigma)$ if and only if $\rho^{\alpha}\le\lambda\sigma^{\alpha-1}$
with some $\lambda\ge 0$, in which case
\begin{align*}
Q_{\alpha}(\rho\|\sigma):=Q_{\alpha,1}(\rho\|\sigma)
=
\Tr\oll{\sigma^{\frac{1-\alpha}{2}}\rho^{\alpha}\sigma^{\frac{1-\alpha}{2}}}.
\end{align*}
(See \cite[Theorem 3.6]{Hiai_fdiv_Springer} for a generalization of the above in the setting of von Neumann algebras, and also for an analogous formula in the case $\alpha\in(0,1)$.)
Moreover, we have the variational formulas
\begin{align*}
Q_{\alpha}(\rho\|\sigma)
&=
\sup_{H\in\B(\hil)_{\sigma,\alpha,1}^+}\left\{\alpha
\Tr \bz H^{1/2}\rho^{\alpha} H^{1/2}\jz^{\frac{1}{\alpha}}+
(1-\alpha)
\Tr\bz H^{1/2}\sigma^{\alpha-1}H^{1/2}\jz^{\frac{1}{\alpha-1}}\right\},\\
\log Q_{\alpha}(\rho\|\sigma)
&=
\sup_{H\in\B(\hil)_{\sigma,\alpha,1}^+}\left\{
\alpha\log\Tr \bz H^{1/2}\rho^{\alpha} H^{1/2}\jz^{\frac{1}{\alpha}}
+(1-\alpha)\log\Tr\bz H^{1/2}\sigma^{\alpha-1}H^{1/2}\jz^{\frac{1}{\alpha-1}}\right\},
\end{align*}
where $\B(\hil)_{\sigma,\alpha,1}^+=
\left\{H\in\B(\hil)\p:\,
0<\Tr\bz H^{1/2}\sigma^{\alpha-1}H^{1/2}\jz^{\frac{1}{\alpha-1}}<+\infty\right\}$.

These variational expressions for the Petz-type R\'enyi divergences do not seem to have appeared in the literature before, even for finite-dimensional operators, although in that case they follow easily from the results of \cite{Zhang2018}.
\end{rem}

The variational formulas in Lemma \ref{lemma:variational} can be used to prove the following 
important properties of the R\'enyi $(\alpha,z)$-divergences.

\begin{cor}\label{cor:trace mon}
Let $\rho,\sigma\in\B(\hil)\p$ be such that $\sigma$ is trace-class. For every 
$(\alpha,z)\in(1,+\infty)\times(0,+\infty)$, 
\begin{align}\label{Q tr mon}
Q_{\alpha,z}(\rho\|\sigma)
\ge (\Tr\rho)^{\alpha}(\Tr\sigma)^{1-\alpha}
\ge 
\alpha\Tr\rho+(1-\alpha)\Tr\sigma.
\end{align}
If, moreover, $\rho$ is trace-class then we have
\begin{align}\label{trace mon eq1}
Q_{\alpha,z}(\rho\|\sigma)=(\Tr\rho)^{\alpha}(\Tr\sigma)^{1-\alpha}\ds\ds\iff\ds\ds
\sigma=\eta\rho \text{ for some }\eta\in(0,+\infty),
\end{align}
and
\begin{align}\label{trace mon eq2}
Q_{\alpha,z}(\rho\|\sigma)=\alpha\Tr\rho+(1-\alpha)\Tr\sigma\ds\ds\iff\ds\ds
\sigma=\rho.
\end{align}
\end{cor}
\begin{proof}
The second inequality in \eqref{Q tr mon} follows simply from the convexity of $\id_{[0,+\infty)}^{\alpha}$. 
The first inequality is obvious when $Q_{\alpha,z}(\rho\|\sigma)=+\infty$, and hence 
we may assume the contrary, i.e, that $\rho\in\L^{\alpha,z}(\hil,\sigma)$. 
The assumption that $\sigma$ is trace-class yields that 
$I\in\B(\hil)_{\sigma,\alpha,z}^+$, and the 
variational formula in \eqref{D variational}
with $H:=I$ yields the first inequality in \eqref{Q tr mon}.
In fact, we don't need the ``full power'' of the variational formula
in \eqref{D variational}
to obtain the first inequality in \eqref{Q tr mon}, as it follows simply from the H\"older inequality as in \eqref{variational lower bound1}--\eqref{variational lower bound4},
with $H=I$. 

Assume for the rest that $\rho$ is trace-class. The right to left implications are straightforward to verify in both \eqref{trace mon eq1} and \eqref{trace mon eq2}.
Assume now that the equality on the LHS of \eqref{trace mon eq1} holds. It implies that 
$Q_{\alpha,z}(\rho\|\sigma)$ is finite, i.e., $\rho\in\L^{\alpha,z}(\hil,\sigma)$, and 
the inequality in \eqref{variational lower bound2} holds as an equality for $H=I$. Thus, by the 
characterization of the equality case in H\"older's inequality 
(Lemma \ref{lemma:Holder}), 
$\sigma=\lambda\abs{\bz \sigma^{\frac{1-\alpha}{2z}}\rho^{\frac{\alpha}{2z}}\jz^*}^{2z}=\lambda\rho_{\sigma,\alpha,z}^z$ for some $\lambda>0$.
From this we get 
$\sigma^{\frac{1}{z}}=\lambda^{\frac{1}{z}}\rho_{\sigma,\alpha,z}$, and
\begin{align*}
\sigma_n^{\frac{1}{z}}
=
P_n\sigma^{\frac{1}{z}} P_n
=
\lambda^{\frac{1}{z}}P_n\rho_{\sigma,\alpha,z} P_n
=
\lambda^{\frac{1}{z}}\sigma_n^{\frac{1-\alpha}{2z}}\rho^{\frac{\alpha}{z}}\sigma_n^{\frac{1-\alpha}{2z}}
\end{align*}
for every $n\in\bN$, where $P_n$ is as in Lemma \ref{lemma:rho alpha}.
Rearranging yields 
$\sigma_n^{\frac{\alpha}{z}}=\lambda^{\frac{1}{z}}P_n\rho^{\frac{\alpha}{z}}P_n$. Thus,
\begin{align*}
\sigma^{\frac{\alpha}{z}}
=
\so\lim_n \sigma_n^{\frac{\alpha}{z}}
=
\lambda^{\frac{1}{z}}\so\lim_n P_n\rho^{\frac{\alpha}{z}}P_n
=
\lambda^{\frac{1}{z}}\rho^{\frac{\alpha}{z}}.
\end{align*}
(In the last equality we use that 
$\rho\in\L^{\alpha,z}(\hil,\sigma)$ implies $\rho^0\le\sigma^0$). 
Hence, $\sigma=\lambda^{\frac{1}{\alpha}}\rho$, i.e, the RHS of 
\eqref{trace mon eq1} holds true.

Finally, assume that the equality on the LHS of \eqref{trace mon eq2} is true. By \eqref{Q tr mon}, 
this implies that the equality on the LHS of \eqref{trace mon eq1} is true, and hence, by the above, $\sigma=\eta\rho$ for some $\eta\in(0,+\infty)$. Moreover, the second equality in
\eqref{Q tr mon} holds as an equality, whence $\Tr\rho=\Tr\sigma$, so we get $\rho=\sigma$ as given 
on the RHS of \eqref{trace mon eq2}.
\end{proof}

\begin{cor}
For any $(\alpha,z)\in(1,+\infty)\times(0,+\infty)$, 
the R\'enyi $(\alpha,z)$-divergence
$D_{\alpha,z}$ is strictly positive in the sense that for any two density operators
$\rho,\sigma\in\S(\hil)$, 
\begin{align*}
D_{\alpha,z}(\rho\|\sigma)\ge 0,\ds\ds\text{with equality if and only if}\ds\ds
\rho=\sigma.
\end{align*}
\end{cor}
\begin{proof}
Immediate from Corollary \ref{cor:trace mon}.
\end{proof}

\begin{rem}
Non-negativity of the R\'enyi $(\alpha,z)$-divergences has been proved in 
\cite{MO-cqconv-cc} in the finite-dimensional case, by different methods.
Strict positivity of the sandwiched R\'enyi $\alpha$-divergences with $\alpha>1$ has been 
proved in the general von Neumann algebra case in \cite{Jencova_NCLp}.
\end{rem}

Finally, we prove the lower semi-continuity of $Q_{\alpha,z}$ and $D_{\alpha,z}$ on pairs of trace-class
operators from the variational formula; we will use this later in the proof of Lemma \ref{lemma:Q alpha as limit}.
\begin{cor}\label{cor:lsc}
For any $\alpha>1$ and $z\ge\alpha$, $Q_{\alpha,z}$ and $D_{\alpha,z}$
are lower semi-continuous on $\L^1(\hil)\times\L^1(\hil)$.
\end{cor}
\begin{proof}
Let $\rho_n,\sigma_n\in\L^1(\hil)$, $n\in\bN$, be convergent sequences 
in trace-norm, with 
$\rho:=\lim_{n\to+\infty}\rho_n$, 
$\sigma:=\lim_{n\to+\infty}\sigma_n$. 
Then
\begin{align*}
\norm{\rho_n^{\frac{\alpha}{z}}}_{\frac{z}{\alpha}}=
(\Tr\rho_n)^{\frac{\alpha}{z}}
=\norm{\rho_n}_1^{\frac{\alpha}{z}}
\xrightarrow[n\to+\infty]{}\norm{\rho}_1^{\frac{\alpha}{z}}=(\Tr\rho)^{\frac{\alpha}{z}}
=\norm{\rho^{\frac{\alpha}{z}}}_{\frac{z}{\alpha}}.
\end{align*}
Since $\norm{\rho_n-\rho}_{\infty}\le\norm{\rho_n-\rho}_1\to 0$, 
the continuity of the functional calculus implies
$\norm{\rho_n^{\frac{\alpha}{z}}-\rho^{\frac{\alpha}{z}}}_{\infty}\to 0$; 
in particular, $\rho_n^{\frac{\alpha}{z}}\to \rho^{\frac{\alpha}{z}}$
in the strong operator topology. 
Hence, by Lemma \ref{lemma:Grumm2},
$\lim_n\norm{\rho_n^{\frac{\alpha}{z}}-\rho^{\frac{\alpha}{z}}}_{\frac{z}{\alpha}}=0$.
Thus, for any $H\in\B(\hil)\p$, 
\begin{align*}
\abs{\norm{H^{1/2}\rho_n^{\frac{\alpha}{z}} H^{1/2}}_{\frac{z}{\alpha}}
-
\norm{H^{1/2}\rho^{\frac{\alpha}{z}} H^{1/2}}_{\frac{z}{\alpha}}
}
&\le
\norm{H^{1/2}\rho_n^{\frac{\alpha}{z}} H^{1/2}
-
H^{1/2}\rho^{\frac{\alpha}{z}} H^{1/2}}_{\frac{z}{\alpha}}\\
&\le
\norm{H}_{\infty}^2\norm{\rho_n^{\frac{\alpha}{z}}-\rho^{\frac{\alpha}{z}}}_{\frac{z}{\alpha}}\xrightarrow[n\to+\infty]{}0.
\end{align*}
This shows that for any $H\in\B(\hil)\p$,
$\rho\mapsto\Tr\bz H^{1/2}\rho^{\frac{\alpha}{z}} H^{1/2}\jz^{\frac{z}{\alpha}}
=
\norm{H^{1/2}\rho^{\frac{\alpha}{z}} H^{1/2}}_{\frac{z}{\alpha}}^{\frac{z}{\alpha}}$ is continuous on $\L^1(\hil)$, and continuity of 
$\sigma\mapsto\Tr\bz H^{1/2}\sigma^{\frac{\alpha-1}{z}} H^{1/2}\jz^{\frac{z}{\alpha-1}}$
on $\L^1(\hil)$ can be proved in the same way.
Thus, by Lemma \ref{lemma:variational}, $\L^1(\hil)\times\L^1(\hil)\ni(\rho,\sigma)\mapsto Q_{\alpha,z}(\rho,\sigma)$ is the supremum of continuous functions, and hence 
it is upper semi-continuous. The assertion about the lower semi-continuity of $D_{\alpha,z}$
follows trivially from this.
\end{proof}

\begin{rem}
Lower semi-continuity of the sandwiched R\'enyi $\alpha$-divergences for $\alpha>1$
(i.e., $z=\alpha>1$) was given in 
\cite[Proposition 3.10]{Jencova_NCLp} in the general von Neumann algebra setting, with a different proof. 
\end{rem}

\subsection{Finite-dimensional approximations}

Our next goal is to investigate the relation between the sandwiched R\'enyi divergences of 
finite-dimensional restrictions of the operators and those of the unrestricted operators. 
We start with the following:

\begin{lemma}\label{lemma:projection monotonicity}
Let $\rho,\sigma\in\B(\hil)\p$, and $K\in\B(\hil,\kil)_{\rho,\sigma}^+$ be a contraction. For any 
$(\alpha,z)\in(1,+\infty)\times(0,+\infty)$ with $\max\{\alpha-1,\alpha/2\}\le z\le \alpha$,
\begin{align}\label{projection monotonicity}
Q_{\alpha,z}(K\rho K^*\|K\sigma K^*)\le Q_{\alpha,z}(\rho\|\sigma).
\end{align}
\end{lemma}
\begin{proof}
By assumption, $\id_{[0,+\infty)}^{\frac{\alpha}{z}}$ is operator convex and
$\id_{[0,+\infty)}^{\frac{\alpha-1}{z}}$ is operator concave, whence
\begin{align}
(K\rho K^*)^{\frac{\alpha}{z}}\le
K\rho^{\frac{\alpha}{z}} K^*,\ds\ds\ds\ds
(K\sigma K^*)^{\frac{\alpha-1}{z}}
\ge
K\sigma^{\frac{\alpha-1}{z}}K^*,
\label{mart proof0}
\end{align}
according to the operator Jensen inequality \cite[Theorem 11]{BrownKosaki}.

If $\rho\notin \L^{\alpha,z}(\hil,\sigma)$ then 
$Q_{\alpha,z}(\rho\|\sigma)=+\infty$ by definition, and \eqref{projection monotonicity} holds trivially. 
Hence, for the rest we assume that $\rho\in \L^{\alpha,z}(\hil,\sigma)$.
Lemma \ref{lemma:rho alpha} yields the existence of some $\lambda\ge 0$ such that 
$\rho^{\frac{\alpha}{z}}\le \lambda \sigma^{\frac{\alpha-1}{z}}$. Thus,
\begin{align*}
(K\rho K^*)^{\frac{\alpha}{z}}\le
K\rho^{\frac{\alpha}{z}} K^*
\le
\lambda K\sigma^{\frac{\alpha-1}{z}}K^*\le
\lambda(K\sigma K^*)^{\frac{\alpha-1}{z}},
\end{align*}
where the first and the last inequalities are due to 
\eqref{mart proof0}. Hence, again by Lemma \ref{lemma:rho alpha},
$K\rho K^*\in \L^{\alpha,z}(\hil,K\sigma K^*)$; in particular, 
the variational formulas in Lemma \ref{lemma:variational} 
hold for $K\rho K^*$ and $K\sigma K^*$ in place of 
$\rho$ and $\sigma$, respectively. 

For any $H\in\B(\kil)_{K\sigma K^*,\alpha,z}$, 
\begin{align}
+\infty>\Tr\bz H^{1/2}(K\sigma K^*)^{\frac{\alpha-1}{z}}H^{1/2}\jz^{\frac{z}{\alpha-1}}
&\ge
\Tr\bz H^{1/2}K\sigma^{\frac{\alpha-1}{z}} K^* H^{1/2}\jz^{\frac{z}{\alpha-1}}
\nonumber\\
&=
\Tr\bz \sigma^{\frac{\alpha-1}{2z}} K^* HK\sigma^{\frac{\alpha-1}{2z}} 
\jz^{\frac{z}{\alpha-1}}\nonumber\\
&=
\Tr\bz (K^* HK)^{1/2}\sigma^{\frac{\alpha-1}{z}} (K^* HK)^{1/2}\jz^{\frac{z}{\alpha-1}},\label{mart proof1}
\end{align}
where the second inequality is due to \eqref{mart proof0}.
In particular, $K^*HK\in\B(\hil)_{\sigma,\alpha,z}$.
Similarly,
\begin{align}
\Tr\bz H^{1/2}(K\rho K^*)^{\frac{\alpha}{z}}H^{1/2}\jz^{\frac{z}{\alpha}}
&\le
\Tr\bz H^{1/2}K\rho^{\frac{\alpha}{z}} K^*H^{1/2}\jz^{\frac{z}{\alpha}}\nonumber\\
&=
\Tr\bz \rho^{\frac{\alpha}{2z}} K^*HK\rho^{\frac{\alpha}{2z}}\jz^{\frac{z}{\alpha}}\nonumber\\
&=
\Tr\bz (K^*HK)^{1/2}\rho^{\frac{\alpha}{z}}(K^*HK)^{1/2}\jz^{\frac{z}{\alpha}}.
\label{mart proof2}
\end{align}
Plugging \eqref{mart proof1}--\eqref{mart proof2} into the variational formula yields
\begin{align*}
&Q_{\alpha,z}(K\rho K^*\|K\sigma K^*)\\
&=
\sup_{H\in\B(\kil)_{K\sigma K^*,\alpha,z}}\left\{\alpha\Tr \bz H^{1/2}(K\rho K^*)^{\frac{\alpha}{z}} H^{1/2}\jz^{\frac{z}{\alpha}}+
(1-\alpha)
\Tr\bz H^{1/2}(K\sigma K^*)^{\frac{\alpha-1}{z}}H^{1/2}\jz^{\frac{z}{\alpha-1}}\right\}\\
&\ds\le
\sup_{H\in\B(\kil)_{K\sigma K^*,\alpha,z}}\left\{\alpha\Tr \bz (K^*HK)^{1/2}\rho^{\frac{\alpha}{z}} (K^*HK)^{1/2}\jz^{\frac{z}{\alpha}}+
(1-\alpha)
\Tr\bz (K^*HK)^{1/2}\sigma^{\frac{\alpha-1}{z}}(K^*HK)^{1/2}\jz^{\frac{z}{\alpha-1}}\right\}\\
&\ds\le
\sup_{H\in\B(\hil)_{\sigma,\alpha,z}}\left\{\alpha\Tr \bz H^{1/2}\rho^{\frac{\alpha}{z}} H^{1/2}\jz^{\frac{z}{\alpha}}+
(1-\alpha)
\Tr\bz H^{1/2}\sigma^{\frac{\alpha-1}{z}}H^{1/2}\jz^{\frac{z}{\alpha-1}}\right\}\\
&\ds=
Q_{\alpha,z}(\rho\|\sigma).
\end{align*}
\end{proof}

\begin{rem}
When $z=\alpha$ and $K=P$ is a projection in Lemma \ref{lemma:projection monotonicity},
one could appeal to the monotonicity of $Q_{\alpha}\nw$ under positive trace-preserving maps, and its additivity on direct sums \cite[Proposition 3.11]{Jencova_NCLp},
to obtain the inequality \eqref{projection monotonicity} as
\begin{align*}
Q_{\alpha}\nw(\rho\|\sigma)\ge
Q_{\alpha}\nw(P\rho P+P^{\perp}\rho P^{\perp}\|P\sigma P+P^{\perp}\sigma P^{\perp})
&=
Q_{\alpha}\nw(P\rho P\|P\sigma P)+
Q_{\alpha}\nw( P^{\perp}\rho P^{\perp}\| P^{\perp}\sigma P^{\perp})\\
&\ge
Q_{\alpha}\nw(P\rho P\|P\sigma P).
\end{align*}
Note, however, that these properties were only proved in \cite{Jencova_NCLp} for positive normal functionals, i.e., 
positive trace-class operators in our setting, and hence this argument gives 
\eqref{projection monotonicity} in a restricted setting compared to that of Lemma \ref{lemma:projection monotonicity}, even when we only consider 
$z=\alpha$ and reductions by projections. 
\end{rem}

Recall that the set of projections on $\hil$ is an upward directed partially ordered set w.r.t.~the PSD order. Lemma \ref{lemma:projection monotonicity} yields the following:

\begin{cor}\label{cor:increasing sequence}
Let $\rho,\sigma\in\B(\hil)\p$ and 
$(\alpha,z)$ as in Lemma \ref{lemma:projection monotonicity}.
For any contraction $K\in\B(\hil)_{\rho,\sigma}^+$, and any projection $P\in\bP(\hil)$ such that 
$|K|^0\le P$, 
\begin{align}\label{increasing 1}
Q_{\alpha,z}(K\rho K^*\|K\sigma K^*)\le 
Q_{\alpha,z}(P\rho P\|P\sigma P).
\end{align}
In particular, 
\begin{align}\label{increasing 2}
\bP(\hil)_{\rho,\sigma}^+\ni P\mapsto Q_{\alpha,z}(P\rho P\|P\sigma P)\ds\ds\text{is increasing.}
\end{align}
\end{cor}
\begin{proof}
Since $K(P\rho P)K^*=K\rho K^*$ and $K(P\sigma P)K^*=K\sigma K^*$,
\eqref{increasing 1} follows immediately by replacing 
$\rho$ with $P\rho P$ and $\sigma$ with $P\sigma P$ in Lemma \ref{lemma:projection monotonicity}. The monotonicity in \eqref{increasing 2} follows immediately from this.
\end{proof}

\begin{defin}\label{def:finitedim appr}
For any $\rho,\sigma\in\B(\hil)\p$ and $(\alpha,z)\in(1,+\infty)\times(0,+\infty)$, let 
\begin{align*}
Q_{\alpha,z}(\rho\|\sigma)\fa
&:=\sup_{P\in\bP_f(\hil)_{\rho,\sigma}^+}Q_{\alpha,z}(P\rho P\|P\sigma P),\\
D_{\alpha,z}(\rho\|\sigma)\fa
&:=\sup_{P\in\bP_f(\hil)_{\rho,\sigma}^+}D_{\alpha,z}(P\rho P\|P\sigma P)
=\frac{1}{\alpha-1}\log Q_{\alpha,z}(\rho\|\sigma)\fa,
\end{align*}
be the \ki{finite-dimensional approximations} of $Q_{\alpha,z}(\rho\|\sigma)$ and 
$D_{\alpha,z}(\rho\|\sigma)$, respectively. If, moreover, $\rho$ is trace-class then we also define
\begin{align*}
\tilde D_{\alpha,z}(\rho\|\sigma)\fa
&:=
D_{\alpha,z}(\rho\|\sigma)\fa-\frac{1}{\alpha-1}\log\Tr\rho.
\end{align*}
\end{defin}

\begin{lemma}\label{lemma:fa equiv}
Let $\rho,\sigma\in\B(\hil)\p$ and $(\alpha,z)$ as in Lemma \ref{lemma:projection monotonicity}. Then 
\begin{align}\label{fa smaller}
Q_{\alpha,z}(\rho\|\sigma)\fa\le Q_{\alpha,z}(\rho\|\sigma),
\end{align}
and 
\begin{align*}
Q_{\alpha,z}(\rho\|\sigma)\fa
&=
\sup\left\{Q_{\alpha,z}(T\rho T\|T\sigma T):\,T\in\B(\hil)_{[0,I]}\cap\B_f(\hil)_{\rho,\sigma}^+\right\}\\
&=
\sup\left\{Q_{\alpha,z}(K\rho K^*\|K\sigma K^*):\,K\in\B_f(\hil)_{\rho,\sigma}^+,\,\norm{K}\le 1\right\}.
\end{align*}
\end{lemma}
\begin{proof}
Immediate from Corollary \ref{cor:increasing sequence}.
\end{proof}
\medskip

Our next goal is to see when equality in \eqref{fa smaller} holds. 

\begin{lemma}\label{lemma:Q alpha convergence}
Let $\rho,\sigma\in\B(\hil)_+$ and $1<\alpha\le z$ be such that $\rho\in\B^{\alpha,z}(\hil,\sigma)$, let $0<c_n<d_n$, $n\in\bN$, be sequences such that 
$c_n\to 0$, $d_n\to+\infty$, and $P_n:=\egy_{(c_n,d_n)}(\sigma)$.
Then 
\begin{align}
Q_{\alpha,z}(\rho\|\sigma)\le \liminf_{n\to+\infty}Q_{\alpha,z}(P_n\rho P_n\|P_n\sigma P_n).
\label{Q alpha convergence}
\end{align}
\end{lemma}
\begin{proof}
Note that by assumption, $\rho^0\le\sigma^0$, and for every large enough $n$, 
$P_n\in\B(\hil)_{\rho,\sigma}^+$.
By Lemma \ref{lemma:rho alpha}, 
\begin{align}\label{Q alpha convergence proof1}
\rho_{\sigma,\alpha,z}=\wo\lim_{n\to+\infty}(P_n\sigma P_n)^{\frac{1-\alpha}{2z}}
\underbrace{(P_n\rho^{\frac{\alpha}{z}} P_n)}_{\le (P_n\rho P_n)^{{\frac{\alpha}{z}}}}
(P_n\sigma P_n)^{\frac{1-\alpha}{2z}},
\end{align}
where the inequality follows from the 
operator Jensen inequality \cite[Theorem 11]{BrownKosaki} due to the fact that 
$\alpha/z\in(0,1)$ by assumption. 
Hence, 
\begin{align*}
Q_{\alpha,z}(\rho\|\sigma)
=\norm{\rho_{\sigma,\alpha,z}}_{z}^{z}
&\le
\liminf_{n\to+\infty}
\norm{(P_n\sigma P_n)^{\frac{1-\alpha}{2z}}(P_n\rho^{\frac{\alpha}{z}} P_n)
(P_n\sigma P_n)^{\frac{1-\alpha}{2z}}}_{z}^{z}\\
&\le
\liminf_{n\to+\infty}
\norm{(P_n\sigma P_n)^{\frac{1-\alpha}{2z}}(P_n\rho P_n)^{\frac{\alpha}{z}}
(P_n\sigma P_n)^{\frac{1-\alpha}{2z}}}_{z}^{z}
=
\liminf_{n\to+\infty}
Q_{\alpha,z}(P_n\rho P_n\|P_n\sigma P_n),
\end{align*}
where the first inequality is due to Lemma \ref{lemma:pnorm lsc}, and the second inequality 
follows from \eqref{Q alpha convergence proof1}.
\end{proof}

The range of $(\alpha,z)$ pairs to which both Lemma \ref{lemma:fa equiv} and Lemma \ref{lemma:Q alpha convergence} apply is 
$1<\alpha=z$, i.e., the case of the sandwiched R\'enyi divergences, and hence for the rest we restrict to this case. Fortunately, this is sufficient for the intended applications in the 
rest of the paper.

\begin{lemma}\label{lemma:Q alpha as limit}
Let $\rho,\sigma\in\B(\hil)\p$ be trace-class, and
$K_n\in\B(\hil,\kil)_{\rho,\sigma}^+$, $n\in\bN$, be contractions such that 
\begin{align*}
\exists\s\so\lim_{n\to+\infty}K_n=:K_{\infty}\in\B(\hil,\kil)_{\rho,\sigma}^+,\ds\ds\ds
\exists\s\so\lim_{n\to+\infty}K_n^*.
\end{align*}   
(That is, $(K_n)_{n\in\bN}$ converges in the strong$^*$ operator topology.)
%$(P_n)_{n\in\bN}$ be a sequence of projections strongly converging to $I$. 
Then 
\begin{align}
Q_{\alpha}\nw(K_{\infty}\rho K_{\infty}^*\|K_{\infty}\sigma K_{\infty}^*)
&\le
\liminf_{n\to+\infty}Q_{\alpha}\nw(K_n\rho K_n^*\|K_n\sigma K_n^*)
\label{Q alpha liminf}\\
&\le
\limsup_{n\to+\infty}Q_{\alpha}\nw(K_n\rho K_n^*\|K_n\sigma K_n^*)
\le 
Q_{\alpha}\nw(\rho\|\sigma).
\label{Q alpha limsup}
\end{align}
In particular, if $P_n\in\bP(\hil)_{\rho,\sigma}^+$, $n\in\bN$, is a sequence of projections strongly converging to some 
$P_{\infty}$ with $P_{\infty}\rho P_{\infty}=\rho$ and 
$P_{\infty}\sigma P_{\infty}=\sigma$ then 
\begin{align*}
\lim_{n\to+\infty}Q_{\alpha}\nw(P_n\rho P_n\|P_n\sigma P_n)
=
Q_{\alpha}\nw(\rho\|\sigma).
\end{align*}
\end{lemma}
\begin{proof}
The second inequality in \eqref{Q alpha limsup} is obvious from 
Lemma \ref{lemma:projection monotonicity}, and the first inequality is trivial. 
By the assumptions and Lemma \ref{lemma:Grumm},
\begin{align}\label{1-norm conv}
\lim_{n\to+\infty}\norm{K_n\rho K_n^*-K_{\infty}\rho K_{\infty}^*}_1=0=
\lim_{n\to+\infty}\norm{K_n\sigma K_n^*-K_{\infty}\sigma K_{\infty}^*}_1.
\end{align}
Since $Q_{\alpha}\nw$ is lower semi-continuous
on $\L^1(\hil)\times \L^1(\hil)$ (see Corollary \ref{cor:lsc}, or \cite[Proposition 3.10]{Jencova_NCLp}), we get
the inequality in \eqref{Q alpha liminf}.
The last assertion follows obviously.
\end{proof}

Lemmas \ref{lemma:fa equiv}, \ref{lemma:Q alpha convergence}, and
\ref{lemma:Q alpha as limit} imply immediately the following:

\begin{prop}\label{cor:fa equality}
Let $\rho,\sigma\in\B(\hil)\p$, and assume that $\rho$ and $\sigma$ are trace-class, 
or that $\sigma$ is compact and $\rho\in\B^{\alpha}(\hil,\sigma)$. Then 
\begin{align*}
Q_{\alpha}\nw(\rho\|\sigma)&=
Q_{\alpha}\nw(\rho\|\sigma)\fa
=
\lim_{\bP_f(\hil)_{\rho,\sigma}^+\ni P\nearrow I}
Q_{\alpha}\nw(P\rho P\|P\sigma P)
=
\lim_{n\to+\infty}
Q_{\alpha}\nw(P_n\rho P_n\|P_n\sigma P_n),%\label{fa equal1}
\\
D_{\alpha}\nw(\rho\|\sigma)&=
D_{\alpha}\nw(\rho\|\sigma)\fa
=
\lim_{\bP_f(\hil)_{\rho,\sigma}^+\ni P\nearrow I}
D_{\alpha}\nw(P\rho P\|P\sigma P)
=
\lim_{n\to+\infty}
D_{\alpha}\nw(P_n\rho P_n\|P_n\sigma P_n),%\label{fa equal2}
\end{align*}
for every $\alpha>1$,
where the convergence in the third expressions in each line is a net convergence in the strong operator topology, 
and the last equalities in each line hold for any sequence $(P_n)_{n\in\bN}$ as in Lemma \ref{lemma:Q alpha convergence}. If, moreover, $\rho$ is trace-class then 
\begin{align}\label{D alpha tilde lim}
\tilde D_{\alpha}\nw(\rho\|\sigma)&=
\tilde D_{\alpha}\nw(\rho\|\sigma)\fa
=
\lim_{\bP_f(\hil)_{\rho,\sigma}^+\ni P\nearrow I}
\tilde D_{\alpha}\nw(P\rho P\|P\sigma P)
=
\lim_{n\to+\infty}
\tilde D_{\alpha}\nw(P_n\rho P_n\|P_n\sigma P_n),\ds\ds\ds\alpha>1.
\end{align}
\end{prop}

\begin{rem}
Finite-dimensional approximability for the standard $f$-divergences 
was given in \cite[Theorem 4.5]{Hiai_fdiv_standard} 
in the general von Neumann algebra setting. In particular, it shows that 
for any two PSD trace-class operators on a Hilbert space,
the standard (or Petz-type) R\'enyi divergences satisfy
$D_{\alpha,1}(\rho\|\sigma)=
D_{\alpha,1}(\rho\|\sigma)\fa:=
\sup_{P\in\bP_f(\hil)_{\rho,\sigma}^+}D_{\alpha,1}(P\rho P\|P\sigma P)$ for
$\alpha\in[0,2]$. It is an open question whether finite-dimensional approximability holds
for $\alpha>1$ 
when $z\ne 1$ and $z\ne \alpha$.
\end{rem}

There are cases apart from the ones treated in Proposition \ref{cor:fa equality} where the inequality in \eqref{fa smaller} holds with equality. In particular, we have the following trivial case, which we will use in the proof of Proposition \ref{prop:sc lower f}.

\begin{lemma}\label{lemma:Q fa infty}
Let $\rho,\sigma\in\B(\hil)\p$. 
\begin{align}
\text{If}\ds\ds\rho^0\nleq\sigma^0\ds\ds\text{then}\ds\ds
Q_{\alpha,z}(\rho\|\sigma)\fa=+\infty=Q_{\alpha,z}(\rho\|\sigma),\ds\ds
\alpha\in(1,+\infty),\ds z\in(0,+\infty).
\end{align}
\end{lemma}
\begin{proof}
Assume that $\rho^0\nleq\sigma^0$, so that there exists a unit vector
$\psi\in\hil$ such that $\sigma^0\psi=0$, $\rho^0\psi\ne 0$. Let 
$\phi$ be any unit vector such that $\sigma^0\phi=\phi$, and for every 
$t\in[0,1]$, define $\psi_t:=\sqrt{1-t}\psi+\sqrt{t}\phi$, 
$P_t:=\pr{\psi_t}$. Then 
\begin{align*}
P_t\rho P_t=\pr{\psi_t}\inner{\psi_t}{\rho\psi_t}\xrightarrow[t\to 0]{}
\pr{\psi}\underbrace{\inner{\psi}{\rho\psi}}_{>0},\ds\ds
P_t\sigma P_t=\pr{\psi_t}\inner{\psi_t}{\sigma\psi_t}\xrightarrow[t\to 0]{}0,
\end{align*}
while $P_t\sigma P_t\ne 0$ for every $t\in(0,1]$. Thus, 
\begin{align*}
Q_{\alpha,z}(\rho\|\sigma)\fa
\ge
\lim_{t\searrow 0}Q_{\alpha,z}(P_t\rho P_t\|P_t\sigma P_t)
=
\lim_{t\searrow 0}\inner{\psi_t}{\rho\psi_t}^{\alpha}\inner{\psi_t}{\sigma\psi_t}^{1-\alpha}
=+\infty.
\end{align*}
Since $\rho^0\nleq\sigma^0$ implies that $\rho\notin\B^{\alpha,z}(\hil,\sigma)$
(see Lemma \ref{lemma:rho alpha}), we also get $Q_{\alpha,z}(\rho\|\sigma)=+\infty$.
\end{proof}
\medskip

The finite-dimensional approximability of the sandwiched R\'enyi divergences 
in Proposition \ref{cor:fa equality}
is the key property used in proving the main results of the paper, the 
equality of the sandwiched and the regularized measured R\'enyi divergences, and 
the determination of the strong converse exponent of state discrimination, in Sections 
\ref{sec:measured Renyi} and \ref{sec:sc Hanti}.

The following monotonicity result has been proved for finite-rank states in 
\cite{Renyi_new}, and for states of a general von Neumann algebra in 
\cite{Jencova_NCLp} and \cite{BST}. We give a different proof of it in our setting as an 
illustration of the use of the finite-dimensional approximability in extending 
finite-dimensional results to infinite dimension. We will give yet another proof 
in Section \ref{sec:measured Renyi}, using a different respresentation of the sandwiched R\'enyi divergences.

\begin{cor}\label{cor:D alpha tilde mon}
Let $\rho,\sigma\in\L^1(\hil)\p$ be PSD trace-class operators. Then 
\begin{align}\label{D alpha tilde mon}
(1,+\infty)\ni\alpha\mapsto\tilde D_{\alpha}\nw(\rho\|\sigma)\ds\ds\text{is increasing},
\end{align}
and
\begin{align}\label{D alpha tilde lim infty}
\lim_{\alpha\to+\infty}\tilde D_{\alpha}\nw(\rho\|\sigma)
=
\lim_{\alpha\to+\infty} D_{\alpha}\nw(\rho\|\sigma)
=
\Dmax(\rho\|\sigma).
\end{align}
\end{cor}
\begin{proof}
These are well-known when $\rho$ and $\sigma$ are finite-rank \cite{Renyi_new}.
Thus, by \eqref{D alpha tilde lim}, the monotonicity in \eqref{D alpha tilde mon} holds.
This also shows that the first limit in \eqref{D alpha tilde lim infty} exists, and 
it is trivial by definition that it is equal to the second limit. 
To show the last equality in \eqref{D alpha tilde lim infty}, it is sufficient to consider 
the case when $\rho$ and $\sigma$ are density operators, 
due to the scaling properties in Remark \ref{rem:scaling}.
Then $D_{\alpha}\nw(\rho\|\sigma)=\tilde D_{\alpha}\nw(\rho\|\sigma)$ for every $\alpha>1$, 
and
\begin{align*}
\lim_{\alpha\to+\infty} D_{\alpha}\nw(\rho\|\sigma)
&=
\lim_{\alpha\to+\infty} \tilde D_{\alpha}\nw(\rho\|\sigma)
=
\sup_{\alpha>1}\tilde D_{\alpha}\nw(\rho\|\sigma)
=
\sup_{\alpha>1} D_{\alpha}\nw(\rho\|\sigma)\\
&=
\sup_{\alpha>1}\sup_{P\in\bP_f(\hil)_{\rho,\sigma}^+}D_{\alpha}\nw(P\rho P\|P\sigma P)\\
&=
\sup_{P\in\bP_f(\hil)_{\rho,\sigma}^+}\sup_{\alpha>1}D_{\alpha}\nw(P\rho P\|P\sigma P)
\\
&=
\sup_{P\in\bP_f(\hil)_{\rho,\sigma}^+}\sup_{\alpha>1}
\left\{\tilde D_{\alpha}\nw(P\rho P\|P\sigma P)+\frac{1}{\alpha-1}\log\Tr P\rho P\right\}\\
&=
\sup_{P\in\bP_f(\hil)_{\rho,\sigma}^+}
D_{\max}(P\rho P\|P\sigma P)
=
\Dmax(\rho\|\sigma).
\end{align*}
Here, the first three equalities are trivial, and the fourth one follows by 
Proposition \ref{cor:fa equality}. The fifth equality is again trivial,
and the sixth one is by definition. In the seventh 
equality we use that both $\alpha\mapsto \tilde D_{\alpha}\nw(P\rho P\|P\sigma P)$ and 
$\alpha\mapsto\frac{1}{\alpha-1}\log\Tr P\rho P$ are increasing, and hence the supremum 
of their sum over $\alpha>1$ is the sum of their limits at $\alpha\to+\infty$, 
which is equal to $D_{\max}(P\rho P\|P\sigma P)$, according to the known behaviour in the 
finite-dimensional case. The last equality is straightforward to verify.
\end{proof}

\subsection{Regularized measured R\'enyi divergence}
\label{sec:measured Renyi}

A finite-outcome positive operator-valued measure (POVM) on a Hilbert space $\hil$
is a map $M:\,[r]\to\B(\hil)$, where $[r]:=\{1,\ldots,r\}$, 
all $M_i$ is PSD,
and 
$\sum_{i=1}^rM_i=I$. 
(We assume without loss of generality that the set of possible outcomes is a subset of $\bN$.)
We denote the set of such POVMs by $\povm(\hil,[r])$. 
For two PSD trace-class operators $\rho,\sigma\in \L^1(\hil)\p$, their \ki{measured R\'enyi divergence} is defined as
\begin{align*}
D_{\alpha}^{\meas}(\rho\|\sigma):=\sup_{r\in\bN}\sup_{M\in\povm(\hil,[r])}
D_{\alpha}\bz\bz\Tr \rho M_i\jz_{i\in[r]}\Big\|
\bz\Tr \sigma M_i\jz_{i\in[r]}\jz,
\end{align*}
where in the second expression we have the classical R\'enyi divergence \cite{Renyi} 
of the given non-negative functions on $[r]$. This is defined for 
$p,q\in[0,+\infty)^{[r]}$ as 
\begin{align*}
D_{\alpha}(p\|q):=\begin{cases}
\frac{1}{\alpha-1}\log\sum_{i\in[r]}p(i)^{\alpha}q(i)^{1-\alpha},&\supp p\subseteq\supp q,\\
+\infty,&\text{otherwise}.
\end{cases}
\end{align*}

One might consider more general POVMs for the definition, but that does not change the value 
of the measured R\'enyi divergence; see, e.g., \cite[Proposition 5.2]{Hiai_fdiv_Springer}.
The \ki{regularized measured R\'enyi divergence} of $\rho$ and $\sigma$ is then defined as
\begin{align*}
\oll{D}_{\alpha}^{\meas}(\rho\|\sigma):=
\sup_{n\in\bN}\frac{1}{n}D_{\alpha}^{\meas}\bz\rho^{\otimes n}\|\sigma^{\otimes n}\jz
=
\lim_{n\to\infty}\frac{1}{n}D_{\alpha}^{\meas}\bz\rho^{\otimes n}\|\sigma^{\otimes n}\jz.
\end{align*}

The following has been shown in \cite{MO}:

\begin{lemma}\label{lemma:reg measured Renyi}
For finite-rank PSD operators $\rho,\sigma\in\B(\hil)\p$,
\begin{align*}
\oll D_{\alpha}^{\meas}(\rho\|\sigma)
=
D_{\alpha}\nw(\rho\|\sigma),\ds\ds\ds\alpha>1.
\end{align*}
\end{lemma}
\medskip

In the proof of the next theorem, we will use the monotonicity of the sandwiched 
R\'enyi $\alpha$-divergences under finite-outcome measurements for $\alpha>1$. 
The more general statement of monotonicity under quantum operations has been proved in 
\cite[Theorem 14]{BST} and \cite[Theorem 3.14]{Jencova_NCLp}
in the general von Neumann algebra setting. We give a different proof 
for trace-class operators on a Hilbert space in Corollary \ref{cor:mon} below.

\begin{thm}\label{thm:reg meas Renyi}
Let $\rho,\sigma\in\B(\hil)\p$ be trace-class, and $\alpha>1$. 
Then 
\begin{align*}
\oll D_{\alpha}^{\meas}(\rho\|\sigma)=D_{\alpha}\nw(\rho\|\sigma).
\end{align*}
\end{thm}
\begin{proof}
The inequality $\oll D_{\alpha}^{\meas}(\rho\|\sigma)\le D_{\alpha}\nw(\rho\|\sigma)$ is trivial from the monotonicity of $D_{\alpha}\nw$ under 
quantum operations
% (\cite[Theorem 14]{BST}, \cite[Theorem 3.14]{Jencova_NCLp}; see also 
%Corollary \ref{cor:mon} below),
and its additivity under tensor products
(Lemma \ref{lemma:tensor product}), and hence we only need to prove the converse inequality. 
By Proposition \ref{cor:fa equality}, 
for any $c<D_{\alpha}\nw(\rho\|\sigma)$ there exists 
a finite-rank projection $P\in\bP_f(\hil)_{\rho,\sigma}^+$ such that 
\begin{align*}
c<D_{\alpha}\nw(P\rho P\|P\sigma P)\le D_{\alpha}\nw(\rho\|\sigma).
\end{align*}
By Lemma \ref{lemma:reg measured Renyi}, there exist $n\in\bN$, a number $r\in\bN$, and 
$M_i\in\B(\hil^{\otimes n})\p$, $M_i^0\le P^{\otimes n}$, $i\in[r]$, with 
$\sum_{i\in[r]}M_i=P^{\otimes n}$, such that 
\begin{align}\label{measured proof1}
c<
\frac{1}{n}D_{\alpha}\bz\bz\Tr(P\rho P)^{\otimes n}M_i\jz_{i\in[r]}\big\|
\bz\Tr(P\sigma P)^{\otimes n}M_i\jz_{i\in[r]}\jz.
\end{align}
Let us define $\tilde M_i:=M_i$, 
$i\in[r]$, and
$\tilde M_{r+1}:=I_{\hil^{\otimes n}}-P^{\otimes n}$. 
Then $(\tilde M_i)_{i\in[r+1]}$ is a POVM on $\hil^{\otimes n}$, and we have 
\begin{align*}
c&<
\frac{1}{n}D_{\alpha}\bz\bz\Tr(P\rho P)^{\otimes n}M_i\jz_{i\in[r]}\big\|
\bz\Tr(P\sigma P)^{\otimes n}M_i\jz_{i\in[r]}\jz\\
&=
\frac{1}{n}D_{\alpha}\bz\bz\Tr\rho^{\otimes n}\tilde M_i\jz_{i\in[r]}\Big\|
\bz\Tr\sigma^{\otimes n}\tilde M_i\jz_{i\in[r]}\jz\\
&\le
\frac{1}{n}D_{\alpha}\bz\bz\Tr \rho^{\otimes n}\tilde M_i\jz_{i\in[r+1]}\Big\|
\bz\Tr\sigma^{\otimes n}\tilde M_i\jz_{i\in[r+1]}\jz\\
&\le
\frac{1}{n}D_{\alpha}^{\meas}\bz\rho^{\otimes n}\|\sigma^{\otimes n}\jz\\
&\le
\oll D_{\alpha}^{\meas}(\rho\|\sigma),
\end{align*}
where the first inequality is by \eqref{measured proof1}, the equality and the second inequality are trivial, and the third and the fourth inequalities are by definition.
Thus, $c<\oll D_{\alpha}^{\meas}(\rho\|\sigma)$, and 
since the above holds for every 
$c<D_{\alpha}\nw(\rho\|\sigma)$, the assertion follows. 
\end{proof}

Their representation given in Theorem \ref{thm:reg meas Renyi} distinguishes the sandwiched R\'enyi divergences among all 
quantum generalizations of the classical R\'enyi divergences; in particular, 
it gives special importance to the $\alpha=z$ case in the family of 
R\'enyi $(\alpha,z)$-divergences, at least for $\alpha>1$. 
It also allows to deduce some important properties of the sandwiched R\'enyi divergences from those of the classical R\'enyi divergences; we present such an example in 
Corollary \ref{cor:D infty}.
Note that the properties in Corollary \ref{cor:D infty}
were also proved in \cite{BST,Jencova_NCLp} in the general von Neumann algebra setting, 
by different methods.
Yet another proof was given in our setting in Corollary \ref{cor:D alpha tilde mon}.

%\begin{cor}\label{cor:monotonicity in alpha}
%For any two density operators $\rho,\sigma\in\S(\hil)$, the function $\alpha\mapsto D_{\alpha}\nw(\rho\|\sigma)$ is increasing on $(1,+\infty)$. 
%\end{cor}
%\begin{proof}
%This is well-known and easy to verify for commuting finite-rank states (i.e., in the finite-dimensional classical setting). The assertion then follows immediately from this and 
%Theorem \ref{thm:reg meas Renyi}.
%\end{proof}
%
%\begin{cor}
%For any two PSD trace-class operators $\rho,\sigma\in\B(\hil)\p\cap\L^1(\hil)$,
%\begin{align*}
%\exists\s\lim_{\alpha\to+\infty}D_{\alpha}\nw(\rho\|\sigma)=
%\log\Tr\rho-\log\Tr\sigma+
%\sup_{\alpha>1}
%D_{\alpha}\nw\bz\frac{\rho}{\Tr\rho}\Big\|\frac{\sigma}{\Tr\sigma}\jz.
%\end{align*}
%\end{cor}
%\begin{proof}
%Immediate from Theorem \ref{thm:reg meas Renyi} and the scaling laws 
%in Remark \ref{rem:scaling}.
%\end{proof}
%
%\begin{rem}
%Corollary \ref{cor:monotonicity in alpha} is a special case of a stronger statement given in 
%\cite[Proposition 3.7]{Jencova_NCLp} in the general von Neumann algebra setting.  
%\end{rem}

\begin{cor}\label{cor:D infty}
Let $\rho,\sigma\in \L^1(\hil)\p$ be PSD trace-class operators. Then
\begin{align}\label{D alpha tilde increasing}
(1,+\infty)\ni\alpha\mapsto \tilde D_{\alpha}\nw(\rho\|\sigma)\ds\ds\text{is increasing},
\end{align}
and
\begin{align}
\sup_{\alpha>1}\tilde D_{\alpha}\nw(\rho\|\sigma)
=
\lim_{\alpha\to+\infty}\tilde D_{\alpha}\nw(\rho\|\sigma)
&=
\lim_{\alpha\to+\infty} D_{\alpha}\nw(\rho\|\sigma)\label{Dmax0}\\
&=
\Dmax(\rho\|\sigma)\\
&=
\log\inf\{\lambda> 0:\,\rho\le\lambda\sigma\}\label{Dmax1}\\
&=
\log\sup\left\{\frac{\Tr\rho T}{\Tr\sigma T}:\,T\in\B(\hil)_{[0,1]},\,\Tr\sigma T>0\right\}.
\label{Dmax2}
%&=
%\lim_{\alpha\to +\infty}D_{\alpha}\nw(\rho\|\sigma).\label{Dmax3}
\end{align}
\end{cor}
\begin{proof}
The increasing property in \eqref{D alpha tilde increasing} is
well-known and easy to verify for commuting finite-rank states (i.e., in the finite-dimensional classical setting). The general case follows immediately from this and 
Theorem \ref{thm:reg meas Renyi}.
The first equality in \eqref{Dmax0} is immediate from the increasing property in 
\eqref{D alpha tilde increasing}, and the second equality is trivial by definition. 

%It is easy to see that the above equalities hold for every $\rho,\sigma$ pair of 
%trace-class operators if and only if they hold for every pair of density operators, and hence
%for the rest we assume without loss of generality that $\rho,\sigma$ are density operators. 

Note that the equality in \eqref{Dmax1} is by definition (see \eqref{Dmax def}), and it is 
clear that $D_{\max}(\rho\|\sigma)$ is an upper bound on 
\eqref{Dmax2}. To prove the converse inequality, note first that 
\eqref{Dmax2} is equal to $+\infty$ if $\rho^0\nleq\sigma^0$, and hence for the rest we assume the contrary. 
Let $0<\lambda<\exp(D_{\max}(\rho\|\sigma))$.  By definition, there exists a unit vector 
$\psi\in\hil$ such that $\inner{\psi}{\rho\psi}>\lambda\inner{\psi}{\sigma\psi}$.
In particular, $\inner{\psi}{\rho\psi}>0$, and hence also 
$\inner{\psi}{\sigma\psi}>0$, due to the assumption that $\rho^0\le\sigma^0$. 
Choosing $T:=\pr{\psi}$ shows that \eqref{Dmax2} is lower bounded by 
$\log\lambda$ for
any such $\lambda$, 
and hence it is also lower bounded by $D_{\max}(\rho\|\sigma)$. 
Thus, we get the  equality in \eqref{Dmax2}.

It is also straightforward to verify that the expressions in 
\eqref{Dmax0} are upper bounded by $D_{\max}(\rho\|\sigma)$.
To prove the converse inequality, note that for any test $T$ as in \eqref{Dmax2},
\begin{align*}
D_{\alpha}\nw(\rho\|\sigma)=\oll D_{\alpha}^{\meas}(\rho\|\sigma)
&\ge 
D_{\alpha}\bz(\Tr\rho T,\Tr\rho (I-T)),(\Tr\sigma T,\Tr\sigma (I-T))\jz\\
&\ge
\frac{1}{\alpha-1}\log\left[(\Tr\rho T)^{\alpha}(\Tr\sigma T)^{1-\alpha}\right]
=
\frac{\alpha}{\alpha-1}\log\Tr\rho T-\log\Tr\sigma T,
\end{align*}
whence
\begin{align*}
\lim_{\alpha\to+\infty}D_{\alpha}\nw(\rho\|\sigma)\ge \log\frac{\Tr\rho T}{\Tr\sigma T}\,.
\end{align*}
Taking the supremum over $T$ yields that
$\lim_{\alpha\to+\infty}D_{\alpha}\nw(\rho\|\sigma)$
is lower bounded by \eqref{Dmax2}, which in turn is equal to 
$D_{\max}(\rho\|\sigma)$ by the above.
\end{proof}

\subsection{The Hoeffding anti-divergences}
\label{sec:Hoeffding}

For any $\rho,\sigma\in\B(\hil)\p$, let
\begin{align*}
\psi\nw(\rho\|\sigma|\alpha)&:=
\log Q_{\alpha}\nw(\rho\|\sigma)=(\alpha-1)D_{\alpha}\nw(\rho\|\sigma),\ds\ds\ds\alpha>1,\\
\tilde \psi\nw(\rho\|\sigma|u)&:=
(1-u)\psi\nw\bz\rho\|\sigma|(1-u)\inv\jz,\ds\ds\ds u\in(0,1).
\end{align*}
We will need these quantities to define the Hoeffding anti-divergences,
which will give the strong converse exponent of state discrimination in Section \ref{sec:sc}.

\begin{lemma}\label{lemma:restricted psi conv}
Let $\rho,\sigma\in\B(\hil)\p$ with $\rho^0\le\sigma^0$.
\smallskip

\noindent (i) For any finite-rank projection 
$P\in\bP_f(\hil)_{\rho,\sigma}^+$, 
$\psi\nw(P\rho P\|P\sigma P|\valt)$ 
and $\tilde\psi\nw(P\rho P\|P\sigma P|\valt)$ are finite-valued convex 
functions on $(1,+\infty)$ and $(0,1)$, respectively, and hence they are continuous.
Moreover, 
\begin{align}
\psi\nw(P\rho P\|P\sigma P|1)
:=
\tilde\psi\nw(P\rho P\|P\sigma P|0)
&:=\lim_{u\searrow 0}\tilde\psi\nw(P\rho P\|P\sigma P|u)
=\log\Tr P\rho P,\label{psi 0}\\
\psi\nw(P\rho P\|P\sigma P|+\infty)
:=
\tilde\psi\nw(P\rho P\|P\sigma P|1)
&:=\lim_{u\nearrow 1}\tilde\psi\nw(P\rho P\|P\sigma P|u)
=
\Dmax(P\rho P\|P\sigma P),\label{psi 1}
\end{align}
and the so extended functions $\psi\nw(P\rho P\|P\sigma P|\valt)$ and 
$\tilde\psi\nw(P\rho P\|P\sigma P|\valt)$ are convex and continuous on 
$[1,+\infty]$ and on $[0,1]$, respectively.
\smallskip

\noindent (ii) For every $\alpha\in[1,+\infty]$, 
$P\mapsto \psi\nw(P\rho P\|P\sigma P|u)$
is monotone increasing on $\bP_f(\hil)_{\rho,\sigma}^+$. 

\noindent (iii) For every $u\in[0,1]$, $P\mapsto \tilde\psi\nw(P\rho P\|P\sigma P|u)$
is monotone increasing on $\bP_f(\hil)_{\rho,\sigma}^+$. 
\end{lemma}
\begin{proof}
By \cite[Corollary 3.11]{MO}, $\psi\nw(P\rho P\|P\sigma P|\valt)$ is a finite-valued convex 
function on $(1,+\infty)$. Hence, it can be written as 
$\psi\nw(P\rho P\|P\sigma P|\alpha)=\sup_{i\in\I}\{c_i\alpha+d_i\}$, $\alpha\in(1,+\infty)$, 
with some 
$c_i,d_i\in\bR$ and an index set $\I$. This implies that 
$\tilde\psi\nw(P\rho P\|P\sigma P|u)=(1-u)\sup_{i\in\I}\{c_i(1-u)\inv+d_i\}=
\sup_{i\in\I}\{c_i+d_i(1-u)\}$, and therefore 
$\tilde\psi\nw(P\rho P\|P\sigma P|\valt)$ is also convex and finite-valued on $(0,1)$, 
and thus it is continuous as well. 
The limits in \eqref{psi 0}--\eqref{psi 1} follow by a straightforward computation, using in the second limit that 
$\lim_{\alpha\to+\infty}D_{\alpha}\nw(\omega\|\tau)=\Dmax(\omega\|\tau)$ for finite-rank states $\omega,\tau$ 
(see \cite{Renyi_new} or Corollary \ref{cor:D infty}).
Convexity and continuity of the extensions are obvious from the definitions.
Monotonicity in (ii) and (iii) are 
immediate from Corollary \ref{cor:increasing sequence}.
\end{proof}

\begin{cor}\label{cor:psi fa}
For any $\rho,\sigma\in\B(\hil)\p$, the functions
\begin{align*}
\psi\nw(\rho\|\sigma|\alpha)\fa&:=
\begin{cases}
\sup_{P\in\bP_f(\hil)_{\rho,\sigma}^+}\psi\nw(P\rho P\|P\sigma P|\alpha),&\ds\ds \alpha\in[1,+\infty],\\
+\infty,&\ds\ds  \alpha\in(-\infty,1),
\end{cases}\\
\tilde \psi\nw(\rho\|\sigma|u)\fa&:=
\begin{cases}
\sup_{P\in\bP_f(\hil)_{\rho,\sigma}^+}\tilde\psi\nw(P\rho P\|P\sigma P|u),&\ds\ds u\in[0,1],\\
+\infty,&\ds\ds  u\in\bR\setminus[0,1],
\end{cases}\\
\tilde \psi\nw(\rho\|\sigma|u)\ofa&:=
\begin{cases}
\sup_{n\in\bN}\frac{1}{n}\tilde \psi\nw(\rho^{\otimes n}\|\sigma^{\otimes n}|u)\fa,&
\ds\ds\ds u\in[0,1],\\
+\infty,&\ds\ds\ds  u\in\bR\setminus[0,1],
\end{cases}
\end{align*}
are convex and lower semi-continuous on $\bR$
(and on $\bR\cup\{+\infty\}$ in the case of $\psi\nw(\rho\|\sigma|\valt)\fa$).
\end{cor}
\begin{proof}
If $\rho^0\nleq\sigma^0$ then all three functions are easily seen to be constant $+\infty$
on $\bR$, and hence for the rest we assume that $\rho^0\le\sigma^0$.
Since the supremum of convex functions is again convex, and the supremum of lower semi-continuous functions is again lower semi-continuous, both properties hold for the above functions on $[1,+\infty]$,
$[0,1]$, and $[0,1]$, respectively, according to Lemma \ref{lemma:restricted psi conv}, and it is trivial to verify that the same is true on the whole of $\bR$.
\end{proof}

\begin{rem}\label{rem:psi tilde fa endpoint}
It is clear that 
\begin{align*}
\tilde \psi\nw(\rho\|\sigma|u)\fa&=
(1-u)\psi\nw(\rho\|\sigma|1/(1-u))\fa
=
(1-u)\log Q_{1/(1-u)}\nw(\rho\|\sigma)\fa,\ds\ds\ds u\in(0,1),
\end{align*}
and \eqref{psi 0}--\eqref{psi 1} yield
\begin{align}
\tilde\psi\nw(\rho \|\sigma |0)\fa
&=\log\Tr\rho=
\tilde\psi\nw(\rho \|\sigma |0)\ofa,\label{psi tilde fa 0}\\
\tilde\psi\nw(\rho \|\sigma |1)\fa
&=
\Dmax(\rho \|\sigma )
=
\tilde\psi\nw(\rho \|\sigma |1)\ofa.\label{psi tilde fa 1}
\end{align}
This motivates to define
\begin{align}
\tilde\psi\nw(\rho \|\sigma |0)
&:=\log\Tr\rho,\label{psi tilde 0}\\
\tilde\psi\nw(\rho \|\sigma |1)
&:=
\Dmax(\rho \|\sigma ).\label{psi tilde 1}
\end{align}
\end{rem}

\begin{rem}\label{rem:psi equalities}
By Corollary \ref{cor:psi fa}, if $\rho$ and $\sigma$ are such that 
$Q_{\alpha}\nw(\rho\|\sigma)=
Q_{\alpha}\nw(\rho\|\sigma)\fa$, $\alpha>1$, then 
$\psi\nw(\rho \|\sigma |\valt)$ and
$\tilde\psi\nw(\rho \|\sigma |\valt)$ are convex and lower semi-continuous 
on $(1,+\infty)$ and on $(0,1)$, respectively. In particular, this holds when 
both $\rho$ and $\sigma$ are trace-class, according to Proposition \ref{cor:fa equality}.
\end{rem}
\medskip

Recall the definition of the finite-dimensional approximation of the sandwiched R\'enyi divergences as a special case of Definition \ref{def:finitedim appr}:
For $\rho,\sigma\in\B(\hil)\p$, 
\begin{align*}
D_{\alpha}\nw(\rho\|\sigma)\fa
&:=
\sup_{P\in\bP_f(\hil)_{\rho,\sigma}^+}D_{\alpha}\nw(P\rho P\|P\sigma P)
=\frac{1}{\alpha-1}\psi\nw(\rho\|\sigma|\alpha)\fa.
\end{align*}
Analogously, we define
\begin{align*}
D_{\alpha}\nw(\rho\|\sigma)\ofa
&:=
\sup_{n\in\bN}\frac{1}{n}D_{\alpha}\nw(\rho^{\otimes n}\|\sigma^{\otimes n})\fa
=\frac{1}{\alpha-1}\psi\nw(\rho\|\sigma|\alpha)\ofa.
\end{align*}

\begin{defin}\label{def:Hanti def}
For $\rho,\sigma\in\B(\hil)\p$ and $r\in\bR$, let 
\begin{align}
H_r\nw(\rho\|\sigma)
&:=
\sup_{\alpha>1}\frac{\alpha-1}{\alpha}
\left[r-D_{\alpha}\nw(\rho\|\sigma)\right]
=
\sup_{u\in(0,1)}\left\{ur-\tilde\psi\nw(\rho\|\sigma|u)\right\},\label{Hanti def01}\\
H_r\nw(\rho\|\sigma)\fa
&:=
\sup_{\alpha>1}\frac{\alpha-1}{\alpha}
\left[r-D_{\alpha}\nw(\rho\|\sigma)\fa\right]
=
\sup_{u\in(0,1)}\left\{ur-\tilde\psi\nw(\rho\|\sigma|u)\fa\right\},\label{Hanti def02}\\
H_r\nw(\rho\|\sigma)\ofa
&:=
\sup_{\alpha>1}\frac{\alpha-1}{\alpha}
\left[r-D_{\alpha}\nw(\rho\|\sigma)\ofa\right]
=
\sup_{u\in(0,1)}\left\{ur-\tilde\psi\nw(\rho\|\sigma|u)\ofa\right\},\label{Hanti def03}\\
\hat H_r\nw(\rho\|\sigma)
&:=
\sup_{u\in[0,1]}\left\{ur-\tilde\psi\nw(\rho\|\sigma|u)\right\},\label{Hanti def2}\\
\hat H_r\nw(\rho\|\sigma)\fa
&:=
\max_{u\in[0,1]}\left\{ur-\tilde\psi\nw(\rho\|\sigma|u)\fa\right\}
=
\max_{u\in\bR}\left\{ur-\tilde\psi\nw(\rho\|\sigma|u)\fa\right\},\label{Hanti def3}\\
\hat H_r\nw(\rho\|\sigma)\ofa
&:=
\max_{u\in[0,1]}\left\{ur-\tilde\psi\nw(\rho\|\sigma|u)\ofa\right\}
=
\max_{u\in\bR}\left\{ur-\tilde\psi\nw(\rho\|\sigma|u)\ofa\right\}.\label{Hanti def4}
\end{align}
Here, $H_r\nw(\rho\|\sigma)$ and $\hat H_r\nw(\rho\|\sigma)$ are two different versions of 
the \ki{Hoeffding anti-divergence} of $\rho$ and $\sigma$ 
with parameter $r\in\bR$, and the rest of the quantities are different finite-dimensional approximations.
\end{defin}

\begin{rem}
$H_r\nw$ and $\hat H_r\nw$ are called anti-divergences because for trace-class operators 
they are monotone 
non-decreasing under quantum operations; 
this is immediate from the monotone non-increasing property of $D_{\alpha}\nw$ under
such maps for $\alpha> 1$; see 
\cite[Theorem 14]{BST}, \cite[Theorem 3.14]{Jencova_NCLp}, or 
Theorem \ref{thm:mon}.
\end{rem}

The Hoeffding anti-divergences are defined as Legendre-Fenchel transforms
(polar functions). For some of them this transformation can be reversed as follows; this will be used in Theorem \ref{thm:mon}.

\begin{lemma}\label{lemma:Hanti bipolar}
For any $\rho,\sigma\in\B(\hil)\p$, 
\begin{align}
\tilde\psi\nw(\rho\|\sigma|u)\fa&=
\sup_{r\in\bR}\left\{ur- H_r\nw(\rho\|\sigma)\fa\right\},
\ds\ds\ds u\in\bR\setminus\{0,1\},\label{Hanti bipolar0}\\
\tilde\psi\nw(\rho\|\sigma|u)\fa&=
\sup_{r\in\bR}\left\{ur-\hat H_r\nw(\rho\|\sigma)\fa\right\},
\ds\ds\ds u\in\bR\label{Hanti bipolar1}\\
\tilde\psi\nw(\rho\|\sigma|u)\ofa&=
\sup_{r\in\bR}\left\{ur-\hat H_r\nw(\rho\|\sigma)\ofa\right\},
\ds\ds\ds u\in\bR.\label{Hanti bipolar2}
\end{align}
\end{lemma}
\begin{proof}
By Corollary \ref{cor:psi fa}, $\tilde\psi\nw(\rho\|\sigma|\valt)\fa$ 
and $\tilde\psi\nw(\rho\|\sigma|\valt)\ofa$ are convex and lower semi-continuous on $\bR$, and hence \eqref{Hanti bipolar1}--\eqref{Hanti bipolar2} follow from 
\eqref{Hanti def3}--\eqref{Hanti def4} according to 
the bipolar theorem (see, e.g., \cite[Proposition 4.1]{EkelandTemam}).
Likewise, $r\mapsto H_r\nw(\rho\|\sigma)\fa$ is the polar function of 
$f(u):=\tilde\psi\nw(\rho\|\sigma|u)\fa+(+\infty)\egy_{\{0,1\}}(u)$, 
$u\in\bR$, and hence, by \cite[Proposition 4.1]{EkelandTemam}, its polar function is
the largest convex and lower semi-continuous minorant of $f$, which is exactly 
$\tilde\psi\nw(\rho\|\sigma|\valt)\fa$.
This proves \eqref{Hanti bipolar0}.
\end{proof}
\medskip

The different variants of the Hoeffding anti-divergence defined above appear naturally in different bounds on the strong converse exponents; see Section \ref{sec:sc}.
Our next goal is to explore their relations; in particular, to find sufficient conditions for some or all of them to coincide. Note that this is not always the case, as shown in 
Examples \ref{ex:infdim states}--\ref{ex:infdim states2}.

\begin{lemma}\label{lemma:fa equality}
Let $\rho,\sigma\in\B(\hil)\p$. For any $u\in[0,1]$, and any $r\in\bR$, 
\begin{align*}
\tilde \psi\nw(\rho\|\sigma|u)\fa\le 
\tilde \psi\nw(\rho\|\sigma|u)\ofa\le 
\tilde \psi\nw(\rho\|\sigma|u),\ds\ds\ds\ds
\begin{array}[t]{lllll}
H_r\nw(\rho\|\sigma)  &\le& H_r\nw(\rho\|\sigma)\ofa &\le & H_r\nw(\rho\|\sigma)\fa \\
\ds\ds    \vertleq        &   &   \ds\ds \vertleq &     & \ds\ds \vertleq   \\
\hat H_r\nw(\rho\|\sigma) &\le& \hat H_r\nw(\rho\|\sigma)\ofa & \le & \hat H_r\nw(\rho\|\sigma)\fa,
\end{array}
\end{align*}
and
\begin{align}
Q_{\alpha}\nw(\rho\|\sigma)\fa=Q_{\alpha}\nw(\rho\|\sigma),\ds\ds\alpha>1
&\ds\ds\iff\ds\ds
\tilde \psi\nw(\rho\|\sigma|u)\fa= \tilde \psi\nw(\rho\|\sigma|u),\ds\ds u\in[0,1]
\label{fa approximation eq1}\\
&\ds\ds\imp\ds\ds
\begin{cases}
H_r\nw(\rho\|\sigma)\fa = H_r\nw(\rho\|\sigma),\ds\ds r\in\bR,\\
\hat H_r\nw(\rho\|\sigma)\fa =\hat H_r\nw(\rho\|\sigma),\ds\ds r\in\bR.
\end{cases}
\label{fa approximation eq2}
\end{align}
In particular, 
if $\rho$ and $\sigma$ are trace-class, or $\sigma$ is compact and $\rho\in\B^{\infty}(\hil,\sigma)$, then all equalities in 
\eqref{fa approximation eq1}--\eqref{fa approximation eq2} hold.
\end{lemma}
\begin{proof}
The inequalities are immediate from \eqref{fa smaller} and the definitions
of the given quantities. 
The equivalence in \eqref{fa approximation eq1} is trivial by definition, as is the 
implication in \eqref{fa approximation eq2}.
The last assertion follows from Proposition \ref{cor:fa equality}.
\end{proof}

\begin{lemma}\label{lemma:fa equality2}
Let $\rho,\sigma\in\B(\hil)\p$ be such that 
$D_{\alpha_0}\nw(\rho\|\sigma)<+\infty$ 
(equivalently, $\rho\in\L^{\alpha_0}(\hil,\sigma)$)
for some $\alpha_0\in(1,+\infty)$. Then 
\begin{align}
H_r\nw(\rho\|\sigma)\fa=\hat H_r\nw(\rho\|\sigma)\fa,\ds\ds\ds r\in\bR.
\label{Hanti max}
\end{align}
\end{lemma}
\begin{proof}
It is enough to prove that 
\begin{align}
\sup_{u\in[0,1)}\left\{ur-\tilde\psi\nw(\rho\|\sigma|u)\fa\right\}
=H_r\nw(\rho\|\sigma)\fa
=
\sup_{u\in(0,1]}\left\{ur-\tilde\psi\nw(\rho\|\sigma|u)\fa\right\}.
\end{align}
We prove the first equality, as the second one follows the same way. 
If $\tilde\psi\nw(\rho\|\sigma|0)\fa=+\infty$ then there is nothing to prove, and hence we assume the contrary. Also by assumption, 
\begin{align*}
+\infty>(\alpha_0-1)D_{\alpha_0}\nw(\rho\|\sigma)=\psi\nw(\rho\|\sigma|\alpha_0)
\ge
\psi\nw(\rho\|\sigma|\alpha_0)\fa
=
\frac{\tilde\psi\nw(\rho\|\sigma|u_0)\fa}{1-u_0},
\end{align*}
where $u_0:=(\alpha_0-1)/\alpha_0$. 
By Corollary \ref{cor:psi fa}, $\tilde\psi\nw(\rho\|\sigma|\valt)\fa$ is 
convex on $[0,1]$, and finiteness at $0$ and $u_0$ implies
$\tilde\psi\nw(\rho\|\sigma|u)\fa<+\infty$, $u\in[0,u_0]$. 
By Lemma \ref{lemma:Q poz}, we also have 
$\tilde\psi\nw(\rho\|\sigma|u)\fa>-\infty$, $u\in[0,u_0]$. 
Hence, $u\mapsto ur-\tilde\psi\nw(\rho\|\sigma|u)\fa$ is a 
finite-valued concave and upper semi-continuous function on $[0,u_0]$, 
whence it is also continuous  on $[0,u_0]$. This proves the asserted equality.
\end{proof}

\begin{prop}\label{prop:Hanti equality}
Let $\rho,\sigma\in\B(\hil)\p$ be such that 
$D_{\alpha_0}\nw(\rho\|\sigma)<+\infty$ 
for some $\alpha_0\in(1,+\infty)$, and
$Q_{\alpha}\nw(\rho\|\sigma)\fa=Q_{\alpha}\nw(\rho\|\sigma)$, $\alpha>1$. 
Then, for every $r\in\bR$,
\begin{align}\label{Hanti all eq}
\begin{array}[t]{lllll}
H_r\nw(\rho\|\sigma)  &=& H_r\nw(\rho\|\sigma)\ofa &= & H_r\nw(\rho\|\sigma)\fa \\
\ds\ds    \verteq        &   &   \ds\ds \verteq &     & \ds\ds \verteq   \\
\hat H_r\nw(\rho\|\sigma) &=& \hat H_r\nw(\rho\|\sigma)\ofa &= & \hat H_r\nw(\rho\|\sigma)\fa.
\end{array}
\end{align}
%In particular, this holds if $\rho$ and $\sigma$ are trace-class, or $\sigma$ is compact and 
%$\rho\in\B^{\infty}(\hil,\sigma)$. 
\end{prop}
\begin{proof}
Immediate from Lemmas \ref{lemma:fa equality} and \ref{lemma:fa equality2}.
\end{proof}

\begin{rem}
Some further properties of, and relations among, 
the different Hoeffding anti-divergences are given in 
Appendix \ref{sec:further}. While these are not used in the rest of the paper, they 
might give some extra insight into the different bounds given in 
Proposition \ref{prop:sc achievability}.
\end{rem}
\medskip

We close this section with some statements on the possible values of the Hoeffding 
anti-divergences. For these, we will need the notion of the 
\ki{Umegaki relative entropy} \cite{Umegaki}. For two 
finite-rank PSD operators $\rho,\sigma\in\B(\hil)\p$, it is defined as
\begin{align*}
D(\rho\|\sigma):=\begin{cases}
\Tr\rho(\logn\rho-\logn\sigma),&\rho^0\le\sigma^0,\\
+\infty,&\text{otherwise},
\end{cases}
\end{align*}
where $\logn x:=x$, $x>0$, and $\logn 0:=0$.
For positive normal functionals on a von Neumann algebra, it may be defined using the 
relative modular operator \cite{Araki_relentr}. In the simple case of
PSD trace-class operators $\rho,\sigma$ on a separable Hilbert space $\hil$, their relative entropy may be expressed equivalently as \cite[Theorem 4.5]{Hiai_fdiv_standard}
\begin{align*}
D(\rho\|\sigma)=\lim_{\bP_f(\hil)_{\rho,\sigma}^+\ni P\nearrow I}D(P\rho P\|P\sigma P)
=
\lim_{n\to+\infty}D(P_n\rho P_n\|P_n\sigma P_n),
\end{align*}
where the second equality holds for any increasing sequence 
$P_n\in\bP_f(\hil)_{\rho,\sigma}^+$, $n\in\bN$, converging strongly to $I$. 
For non-zero PSD trace-class operators $\rho,\sigma$ and $\lambda,\eta\in(0,+\infty)$, the scaling laws 
\begin{align}
D_{\alpha}\nw(\lambda\rho\|\eta\sigma)&=D_{\alpha}\nw(\rho\|\sigma)+\frac{\alpha}{\alpha-1}\log\lambda-\log\eta,\label{Renyi scaling}\\
H_r\nw(\lambda\rho\|\eta\sigma)&=H_{r+\log\eta}\nw(\rho\|\sigma)-\log\lambda,
\label{Hanti scaling}\\
D(\lambda\rho\|\eta\sigma)&=
\lambda D(\rho\|\sigma)+(\Tr\rho)\lambda\log\frac{\lambda}{\eta},
\label{relentr scaling}\\
\Dmax(\lambda\rho\|\eta\sigma)&=
\Dmax(\rho\|\sigma)+\log\lambda-\log\eta,
\label{Dmax scaling2}
\end{align}
are easy to verify from the definitions (see also Remark \ref{rem:scaling}).
It was shown in \cite{BST,Jencova_NCLp} that 
\begin{align}\label{D alpha tilde limit1}
\exists\,\alpha_0>0:\ds D_{\alpha_0}\nw(\rho\|\sigma)<+\infty\ds\imp\ds
\lim_{\alpha\searrow 1}\tilde D_{\alpha}\nw(\rho\|\sigma)
=
\inf_{\alpha>1}\tilde D_{\alpha}\nw(\rho\|\sigma)
=
\frac{1}{\Tr\rho}D(\rho\|\sigma).
\end{align}

\begin{lemma}\label{lemma:Hanti bounds}
Let $\rho,\sigma\in\L^1(\hil)\p$ be PSD trace-class operators. 

\noindent (i)
For every $r\in\bR$,
\begin{align}\label{Hanti lower}
H_r\nw(\rho\|\sigma)\ge r-\Dmax(\rho\|\sigma).
\end{align}

\noindent (ii) If there exists an $\alpha_0\in(1,+\infty)$ such that 
$D_{\alpha_0}\nw(\rho\|\sigma)<+\infty$ then
\begin{align}
H_r\nw(\rho\|\sigma)=&H_r\nw(\rho\|\sigma)\ofa=H_r\nw(\rho\|\sigma)\fa
=\hat H_r\nw(\rho\|\sigma)=\hat H_r\nw(\rho\|\sigma)\ofa=\hat H_r\nw(\rho\|\sigma)\fa
\label{Hanti eq2}\\
&\begin{cases}
=-\log\Tr\rho,& r\le \frac{1}{\Tr\rho}D(\rho\|\sigma)-\log\Tr\rho,\\
\in\bz-\log\Tr\rho, r-\frac{1}{\Tr\rho}D(\rho\|\sigma)\jz,&r> \frac{1}{\Tr\rho}D(\rho\|\sigma)-\log\Tr\rho.\label{Hanti upper}
\end{cases}
\end{align}

\noindent (iii) If  
$D_{\alpha}\nw(\rho\|\sigma)=+\infty$ for every $\alpha\in(1,+\infty)$ then
\begin{align}
H_r\nw(\rho\|\sigma)=H_r\nw(\rho\|\sigma)\ofa=H_r\nw(\rho\|\sigma)\fa
&=
-\infty\\
&<-\log\Tr\rho=
\hat H_r\nw(\rho\|\sigma)=
\hat H_r\nw(\rho\|\sigma)\ofa=\hat H_r\nw(\rho\|\sigma)\fa.
\end{align}
\end{lemma}
\begin{proof}
\noindent (i) By the scaling laws \eqref{Renyi scaling}--\eqref{relentr scaling},
\begin{align}
H_r\nw(\rho\|\sigma)&=
H_{r+\log\Tr\sigma}\nw\bz\frac{\rho}{\Tr\rho}\Big\|\frac{\sigma}{\Tr\sigma}\jz-\log\Tr\rho
\nonumber\\
&=
\sup_{\alpha>1}\frac{\alpha-1}{\alpha}\left[r+\log\Tr\sigma-D_{\alpha}\nw\bz\frac{\rho}{\Tr\rho}\Big\|\frac{\sigma}{\Tr\sigma}\jz\right]-\log\Tr\rho.\label{Hanti normalized}
\end{align}
According to Corollary \ref{cor:D infty}, 
$\displaystyle{\lim_{\alpha\to+\infty}}$ $D_{\alpha}\nw\bz\frac{\rho}{\Tr\rho}\Big\|\frac{\sigma}{\Tr\sigma}\jz
=\Dmax\bz\frac{\rho}{\Tr\rho}\Big\|\frac{\sigma}{\Tr\sigma}\jz
=
\Dmax\bz\rho\|\sigma\jz-\log\Tr\rho+\log\Tr\sigma$, and hence,
\begin{align*}
H_r\nw(\rho\|\sigma)&\ge \lim_{\alpha\to+\infty}
\frac{\alpha-1}{\alpha}\left[r+\log\Tr\sigma-D_{\alpha}\nw\bz\frac{\rho}{\Tr\rho}\Big\|\frac{\sigma}{\Tr\sigma}\jz\right]-\log\Tr\rho
=
r-\Dmax\bz\rho\|\sigma\jz,
\end{align*}
proving \eqref{Hanti lower}.
\smallskip

\noindent (ii) 
The equalities in \eqref{Hanti eq2} follow from 
Proposition  \ref{prop:Hanti equality}.
Using the assumption $D_{\alpha_0}\nw(\rho\|\sigma)<+\infty$,
\eqref{D alpha tilde limit1} and \eqref{relentr scaling} give
\begin{align}
\inf_{\alpha>1}D_{\alpha}\nw\bz\frac{\rho}{\Tr\rho}\Big\|\frac{\sigma}{\Tr\sigma}\jz
=
\lim_{\alpha\searrow 1}D_{\alpha}\nw\bz\frac{\rho}{\Tr\rho}\Big\|\frac{\sigma}{\Tr\sigma}\jz
=
D\bz\frac{\rho}{\Tr\rho}\Big\|\frac{\sigma}{\Tr\sigma}\jz
=
\frac{1}{\Tr\rho}D(\rho\|\sigma)-\log\Tr\rho+\log\Tr\sigma.\label{normalized relentr}
\end{align}
In particular, the above limit is finite, and thus
\begin{align*}
-\log\Tr\rho
=
\lim_{\alpha\searrow 1}
\frac{\alpha-1}{\alpha}\left[r+\log\Tr\sigma-D_{\alpha}\nw\bz\frac{\rho}{\Tr\rho}\Big\|\frac{\sigma}{\Tr\sigma}\jz\right]-\log\Tr\rho
\le
H_r\nw(\rho\|\sigma),
\end{align*}
where in the second expression we used \eqref{Hanti normalized}, and the inequality is by definition. 
On the other hand, 
\eqref{Hanti normalized} shows that 
$H_r\nw(\rho\|\sigma)>-\log\Tr\rho$ holds if and only if
\begin{align}\label{large r}
r+\log\Tr\sigma>\inf_{\alpha>1}D_{\alpha}\nw\bz\frac{\rho}{\Tr\rho}\Big\|\frac{\sigma}{\Tr\sigma}\jz=D\bz\frac{\rho}{\Tr\rho}\Big\|\frac{\sigma}{\Tr\sigma}\jz=\frac{1}{\Tr\rho}D(\rho\|\sigma)-\log\Tr\rho+\log\Tr\sigma,
\end{align} 
where the equalities are due to \eqref{normalized relentr}.
Note that \eqref{large r} is exactly the condition in the second line of \eqref{Hanti upper}, and hence we obtain the first line in \eqref{Hanti upper}.
Assume now that $r$ is as in \eqref{large r}. Then 
\begin{align*}
\underbrace{\frac{\alpha-1}{\alpha}}_{\in(0,1)}
\underbrace{\left[r+\log\Tr\sigma-D_{\alpha}\nw\bz\frac{\rho}{\Tr\rho}\Big\|\frac{\sigma}{\Tr\sigma}\jz\right]}_{\le\,
r+\log\Tr\sigma-D\bz\frac{\rho}{\Tr\rho}\Big\|\frac{\sigma}{\Tr\sigma}\jz\,\in(0,+\infty)
}-\log\Tr\rho
<r-\frac{1}{\Tr\rho}D(\rho\|\sigma),
\end{align*}
proving the second line of \eqref{Hanti upper}.
\smallskip

\noindent (iii) 
By Lemma \ref{lemma:fa equality},
$\tilde\psi\nw(\rho\|\sigma|u)=\tilde\psi\nw(\rho\|\sigma|u)\ofa=\tilde\psi\nw(\rho\|\sigma|u)\fa=+\infty$ for every 
$u\in(0,1)$, whence
$H_r\nw(\rho\|\sigma)=H_r\nw(\rho\|\sigma)\ofa=H_r\nw(\rho\|\sigma)\fa=
-\infty$. On the other hand, 
$\tilde\psi\nw(\rho\|\sigma|0)=\tilde\psi\nw(\rho\|\sigma|0)\ofa=\tilde\psi\nw(\rho\|\sigma|0)\fa=\log\Tr\rho$,
according to Remark \ref{rem:psi tilde fa endpoint}, and 
$\tilde\psi\nw(\rho\|\sigma|1)=\tilde\psi\nw(\rho\|\sigma|1)\ofa=\tilde\psi\nw(\rho\|\sigma|1)\fa=\Dmax(\rho\|\sigma)=+\infty$, where the last equality follows from 
Corollary \ref{cor:D infty}. Hence,
$\hat H_r\nw(\rho\|\sigma)=\hat H_r\nw(\rho\|\sigma)\ofa=\hat H_r\nw(\rho\|\sigma)\fa=-\log\Tr\rho$. 
\end{proof}

\begin{ex}\label{ex:infdim states}
Let $(e_n)_{n\in\bN}$ be an orthonormal basis in $\hil$, and 
$\rho:=c_1\sum_{n=1}^{+\infty}n^{-\beta}\pr{e_n}$,
$\sigma:=c_2\sum_{n=1}^{+\infty}n^{-n^{\gamma}}\pr{e_n}$,
with some $\beta>1$ and $\gamma>0$, where 
$c_1$ and $c_2$ are choosen so that $\rho$ and $\sigma$ are 
density operators. 
Obviously, $\rho$ and $\sigma$ are commuting (classical).
For $P_N:=\sum_{n=1}^N\pr{e_n}$, we have
\begin{align*}
Q_{\alpha}\nw(P_N\rho P_N\|P_N\sigma P_N)=c_1^{\alpha}c_2^{1-\alpha}
\sum_{n=1}^N n^{-\alpha\beta-(1-\alpha)n^{\gamma}}
\xrightarrow[N\to+\infty]{}+\infty,\ds\ds\ds \alpha\in(1,+\infty),
\end{align*}
whence 
\begin{align*}
D_{\alpha}\nw(\rho\|\sigma)=+\infty,\ds\ds\alpha\in(1,+\infty)
\ds\imp\ds
&H_r\nw(\rho\|\sigma)=H_r\nw(\rho\|\sigma)\ofa=H_r\nw(\rho\|\sigma)\fa
=
-\infty\\
&<-\log\Tr\rho=
\hat H_r\nw(\rho\|\sigma)=
\hat H_r\nw(\rho\|\sigma)\ofa=\hat H_r\nw(\rho\|\sigma)\fa,
\end{align*}
according to Lemma \ref{lemma:Hanti bounds}.
Note also that 
\begin{align*}
&\tilde\psi\nw(\rho\|\sigma|u)=\tilde\psi\nw(\rho\|\sigma|u)\fa
=\tilde\psi\nw(\rho\|\sigma|u)\ofa=+\infty,\ds\ds u\in(0,1),\\
&\tilde\psi\nw(\rho\|\sigma|0)\fa
=\tilde\psi\nw(\rho\|\sigma|0)\ofa=\log\Tr\rho=0
<\lim_{u\searrow 0}\tilde\psi\nw(\rho\|\sigma|u)\fa
,\\
&\tilde\psi\nw(\rho\|\sigma|1)\fa
=\tilde\psi\nw(\rho\|\sigma|1)\ofa=\Dmax(\rho\|\sigma)=+\infty
=
\lim_{u\nearrow 1}\tilde\psi\nw(\rho\|\sigma|u)\fa.
\end{align*}

For the relative entropy we get 
\begin{align*}
D(\rho\|\sigma)
=
c_1\sum_{n=1}^{+\infty}\frac{1}{n^{\beta}}\log\frac{c_1n^{n^{\gamma}}}{c_2n^{\beta}}
=
\log\frac{c_1}{c_2}+c_1\sum_{n=1}^{+\infty}\frac{(n^{\gamma}-\beta)\log n}{n^{\beta}}
<+\infty,
\end{align*}
if $\beta>\gamma+1$.
Hence, assuming that $D(\rho\|\sigma)<+\infty$ is not sufficient for 
Lemma \ref{lemma:fa equality2} and Lemma \ref{lemma:Hanti bounds}.

This also gives an example where 
\begin{align*}
D(\rho\|\sigma)<+\infty=\lim_{\alpha\searrow 1}D_{\alpha}\nw(\rho\|\sigma),
\end{align*}
which is contrary to the case where $D_{\alpha_0}\nw(\rho\|\sigma)<+\infty$ for some 
$\alpha_0>1$; see \eqref{D alpha tilde limit1}.
This kind of behaviour was already pointed out in \cite[Remark 5.4]{Hiai_fdiv_standard}.
\end{ex}

\begin{ex}\label{ex:infdim states2}
Let $\rho=\sigma\in\B(\hil)\p$ be such that $\rho$ is not trace-class.
Then $\rho=\rho^{\frac{\alpha-1}{2\alpha}}\rho^{\frac{1}{\alpha}}\rho^{\frac{\alpha-1}{2\alpha}}$, whence 
$\rho\in\B^{\alpha}(\hil,\sigma)$ with $\rho_{\sigma,\alpha}=\rho^{\frac{1}{\alpha}}\notin\L^{\alpha}(\hil)$, and
\begin{align*}
&\tilde\psi\nw(\rho\|\sigma|u)=\tilde\psi\nw(\rho\|\sigma|u)\fa
=\tilde\psi\nw(\rho\|\sigma|u)\ofa=+\infty,\ds\ds u\in(0,1),\\
&\tilde\psi\nw(\rho\|\sigma|0)\fa
=\tilde\psi\nw(\rho\|\sigma|0)\ofa=\log\Tr\rho=+\infty
=\lim_{u\searrow 0}\tilde\psi\nw(\rho\|\sigma|u)\fa
,\\
&\tilde\psi\nw(\rho\|\sigma|1)\fa
=\tilde\psi\nw(\rho\|\sigma|1)\ofa=\Dmax(\rho\|\sigma)=0
<
\lim_{u\nearrow 1}\tilde\psi\nw(\rho\|\sigma|u)\fa.
\end{align*}
Thus,
\begin{align*}
H_r\nw(\rho\|\sigma)=H_r\nw(\rho\|\sigma)\ofa=H_r\nw(\rho\|\sigma)\fa=-\infty<r
=\hat H_r\nw(\rho\|\sigma)=\hat H_r\nw(\rho\|\sigma)\ofa=\hat H_r\nw(\rho\|\sigma)\fa,\ds\ds r\in\bR.
\end{align*}
In particular, this holds also when $\rho$ is compact, and obviously 
$\rho\in\B^{\infty}(\hil,\rho)$. This shows that the assumption 
$D_{\alpha_0}\nw(\rho\|\sigma)<+\infty$ for some $\alpha_0<+\infty$ is also 
important in this case of Proposition \ref{prop:Hanti equality}.

Note also that this is an example where 
\begin{align*}
\exists\,\lim_{\alpha\to +\infty}D_{\alpha}\nw(\rho\|\sigma)\ds(=+\infty)\ds \ne
\Dmax(\rho\|\sigma).
\end{align*}
This cannot happen when $\rho$ and $\sigma$ are both trace-class, according to 
Corollary \ref{cor:D alpha tilde mon} or
Corollary \ref{cor:D infty}.
\end{ex}

\section{The strong converse exponent} 
\label{sec:sc}

\subsection{The strong converse exponents and the Hoeffding anti-divergences}
\label{sec:sc Hanti}

Before restricting our attention to the i.i.d.~case in the main result,
we first consider a generalization of the binary state discrimination problem described in the Introduction. First, we do not assume the hypotheses to be represented by density operators, but by general positive semi-definite operators. Second, we do not assume the problem to be i.i.d. In the most general case, a \ki{simple asymptotic binary operator discrimination problem} is specified by a sequence of Hilbert spaces $\hil_n$, $n\in\bN$, and for each $n\in\bN$, a pair 
$\rho_n,\sigma_n\in\B(\hil_n)\p$, representing the null and the alternative hypotheses, 
respectively.
Since the operators are not assumed to be trace-class, the expressions in 
\eqref{errors intro} may not make sense, and need to be modified as
\begin{align*}
\tos_n(T_n|\rho_n)&:=\Tr\, (T_n^{1/2}\rho_n T_n^{1/2})
=\Tr\,(\rho_n^{1/2}T_n\rho_n^{1/2}),\\
\beta_n(T_n|\sigma_n)&:=\Tr\, (T_n^{1/2}\sigma_n T_n^{1/2})=
\Tr\, (\sigma_n^{1/2} T_n\sigma_n^{1/2}),
\end{align*} 
to define the generalized type I success and type II errors, respectively.
These expressions are equal to those in \eqref{errors intro} when 
$\rho_n$ and $\sigma_n$ are trace-class. 

\begin{definition}
Let $\vec{\rho}:=(\rho_n)_{n\in\bN}$, $\vec{\sigma}:=(\sigma_n)_{n\in\bN}$ be as above.
The \ki{strong converse exponents} of the simple asymptotic binary operator discrimination problem $H_0:\,\vec{\rho}$ \s vs.~$H_1:\,\vec{\sigma}$ \s with type II exponent $r\in\bR$ are defined as
\begin{align*}
\scli_r(\vec{\rho}\|\vec{\sigma})&:=
\inf\left\{\liminf_{n\to+\infty}-\frac{1}{n}\log\tos_n(T_n|\rho_n):\,
\liminf_{n\to+\infty}-\frac{1}{n}\log \beta_n(T_n|\sigma_n)\ge r
\right\},\\
\scls_r(\vec{\rho}\|\vec{\sigma})&:=
\inf\left\{\limsup_{n\to+\infty}-\frac{1}{n}\log\tos_n(T_n|\rho_n):\,
\liminf_{n\to+\infty}-\frac{1}{n}\log \beta_n(T_n|\sigma_n)\ge r
\right\},\\
\scl_r(\vec{\rho}\|\vec{\sigma})&:=
\inf\left\{\lim_{n\to+\infty}-\frac{1}{n}\log\tos_n(T_n|\rho_n):\,
\liminf_{n\to+\infty}-\frac{1}{n}\log \beta_n(T_n|\sigma_n)\ge r
\right\},
\end{align*}
where the infima are taken along all test sequences $T_n\in\B(\hil_n)_{[0,I]}$,
$n\in\bN$, 
satisfying the indicated condition, and in the last expression also that the limit exists. 
\end{definition}

We will need an extension of the notion of the Hoeffding anti-divergence in the above setting. 
Let 
\begin{align*}
\psi\nw(\vec{\rho}\|\vec{\sigma}|\alpha)&:=\limsup_{n\to+\infty}\frac{1}{n}\psi\nw(\rho_n\|\sigma_n|\alpha),\ds\ds\ds \alpha\in(1,+\infty),\\
\tilde \psi\nw(\vec{\rho}\|\vec{\sigma}|u)&:=
\limsup_{n\to+\infty}\frac{1}{n}\tilde\psi\nw(\rho_n\|\sigma_n|u)=
\begin{cases}
(1-u)\psi\nw\bz\vec{\rho}\|\vec{\sigma}|(1-u)\inv\jz,& u\in(0,1),\\
\limsup_{n\to+\infty}\frac{1}{n}\log\Tr\rho_n,&u=0,\\
\limsup_{n\to+\infty}\frac{1}{n}\Dmax(\rho_n\|\sigma_n),&u=1,
\end{cases}
\end{align*}
where we used \eqref{psi tilde 0}--\eqref{psi tilde 1}, 
and
\begin{align*}
\hat H_r\nw\bz\vec{\rho}\|\vec{\sigma}\jz:=
\sup_{u\in[0,1]}\{ur-\tilde \psi\nw(\vec{\rho}\|\vec{\sigma}|u)\}.
\end{align*}

The inequality in the following lemma is called the optimality part of the Hoeffding bound. 
For trace-class operators, it can be easily obtained from the monotonicity of the sandwiched 
R\'enyi divergence under measurements; see \cite{N,MO,BST}. 
If we do not assume $\rho_n$ and $\sigma_n$ to be trace-class, we can still obtain it using the 
variational formula in \eqref{D variational}, as we show below.

\begin{prop}\label{lemma:sc optimality}
For every $r\in\bR$,
\begin{align}\label{sc lower}
\hat H_r\nw(\vec{\rho}\|\vec{\sigma})\le \scli_r(\vec{\rho}\|\vec{\sigma})\le \scls_r(\vec{\rho}\|\vec{\sigma})\le \scl_r(\vec{\rho}\|\vec{\sigma}).
\end{align}
\end{prop}
\begin{proof}
All the inequalities are trivial by definition, except for the first one.
Thus, we need to show that for any $r\in\bR$ and any $u\in[0,1]$, 
\begin{align}\label{Hoeffding opt proof1}
\scli_r(\vec{\rho}\|\vec{\sigma})\ge ur-\tilde\psi\nw(\vec{\rho}\|\vec{\sigma}|u).
\end{align}
Let us fix $r\in\bR$ for the rest. 
First, note that for any test $T_n$,
\begin{align*}
\tos_n(T_n|\rho_n)=\Tr\,(\rho_n^{1/2}T_n\rho_n^{1/2})\le\Tr\rho_n.
\end{align*}
Thus, for any sequence of tests $(T_n)_{n\in\bN}$, 
\begin{align*}
\liminf_{n\to+\infty}-\frac{1}{n}\log\tos_n(T_n|\rho_n)\ge -\limsup_{n\to+\infty}
\frac{1}{n}\log\Tr\rho_n=0\cdot r-\tilde \psi\nw(\vec{\rho}\|\vec{\sigma}|0),
\end{align*}
proving \eqref{Hoeffding opt proof1} for $u=0$.
For $u=1$, \eqref{Hoeffding opt proof1} is trival when $\psi\nw(\vec{\rho}\|\vec{\sigma}|1)=+\infty$, and hence we assume the contrary; in particular, $\Dmax(\rho_n\|\sigma_n)<+\infty$ for any large enough $n$. Let $(T_n)_{n\in\bN}$ be a test sequence such that 
\begin{align*}
\liminf_{n\to+\infty}-\frac{1}{n}\log\beta_n(T_n|\sigma_n)\ge r. 
\end{align*}
Then for any $r'<r$ and any large enough $n$, 
$\beta_n(T_n|\sigma_n)\le \exp(-nr')$, whence
\begin{align*}
\tos(T_n|\rho_n)=
\Tr\, (T_n^{1/2}\rho_n T_n^{1/2})
\le
\exp({\Dmax(\rho_n\|\sigma_n)})\Tr\, (T_n^{1/2}\sigma_n T_n^{1/2})
\le
\exp(\Dmax(\rho_n\|\sigma_n)-nr').
\end{align*}
Thus, 
\begin{align*}
\liminf_{n\to+\infty}-\frac{1}{n}\log\tos_n(T_n|\rho_n)\ge
r'-\limsup_{n\to+\infty}\frac{1}{n}\Dmax(\rho_n\|\sigma_n)
=
r'-\tilde \psi\nw(\vec{\rho}\|\vec{\sigma}|1).
\end{align*}
This gives \eqref{Hoeffding opt proof1} for $u=1$.

For the rest, 
let us fix an $u\in(0,1)$, and corresponding 
$\alpha=1/(1-u)>1$.
If $\tilde \psi\nw(\vec{\rho}\|\vec{\sigma}|u)=+\infty$ then 
\begin{align}\label{sc lower proof2}
ur-\tilde\psi\nw(\vec{\rho}\|\vec{\sigma}|u)=-\infty\le \scli_r(\vec{\rho}\|\vec{\sigma})
\end{align}
holds trivially. Hence, we assume that $\tilde\psi\nw(\rho_n\|\sigma_n|u)<+\infty$, 
or equivalently,  $\rho_n\in \L^{\alpha}(\hil_n,\sigma_n)$ for every large enough $n$. 
In particular,
the variational formula  \eqref{D variational} holds (with $z=\alpha)$.

Consider now a sequence of tests $(T_n)_{n\in\bN}$ such that 
$\liminf_{n\to+\infty}-\frac{1}{n}\log \Tr\,(T_n^{1/2}\sigma_n T_n^{1/2})\ge r$.
Then $\Tr\,(T_n^{1/2}\sigma_n T_n^{1/2})<+\infty$ for every large enough $n$,
and we have
\begin{align}\label{sc lower proof1}
\Tr\bz T_n^{1/2}\sigma_n^{\frac{\alpha-1}{\alpha}}T_n^{1/2}\jz^{\frac{\alpha}{\alpha-1}}
\le 
\Tr\bz T_n^{1/2}\sigma_nT_n^{1/2}\jz
<+\infty,
\end{align}
where the first inequality is due to the operator Jensen inequality 
\cite[Theorem 11]{BrownKosaki}.
Hence, $T_n\in\B(\hil_n)_{\sigma_n,\alpha,\alpha}$.
If $\Tr\bz T_n^{1/2}\sigma_n^{\frac{\alpha-1}{\alpha}}T_n^{1/2}\jz^{\frac{\alpha}{\alpha-1}}>0$ then
the variational formula \eqref{D variational} yields
\begin{align*}
\psi\nw(\rho_n\|\sigma_n|\alpha)
&\ge 
\alpha\log\Tr \,(T_n^{1/2}\rho_n T_n^{1/2})
+(1-\alpha)\log\Tr\bz T_n^{1/2}\sigma_n^{\frac{\alpha-1}{\alpha}}T_n^{1/2}\jz^{\frac{\alpha}{\alpha-1}}\\
&\ge
\alpha\log\Tr\,( T_n^{1/2}\rho_n T_n^{1/2})
+(1-\alpha)\log\Tr\bz T_n^{1/2}\sigma_nT_n^{1/2}\jz,
\end{align*}
where the second inequality is due to \eqref{sc lower proof1}.
In particular, we also have $\Tr\,(T_n^{1/2}\rho_n T_n^{1/2})<+\infty$.
By a simple rearrangement, we get 
\begin{align}\label{sc lower proof3}
-\frac{1}{n}\log\Tr\,(T_n^{1/2}\rho_n T_n^{1/2})
&\ge
\frac{\alpha-1}{\alpha}\bz-\frac{1}{n}\jz
\log\Tr\bz T_n^{1/2}\sigma_nT_n^{1/2}\jz
-\frac{1}{\alpha}\frac{1}{n}\psi\nw(\rho_n\|\sigma_n|\alpha).
\end{align}
If $\Tr\bz T_n^{1/2}\sigma_n^{\frac{\alpha-1}{\alpha}}T_n^{1/2}\jz^{\frac{\alpha}{\alpha-1}}=0$ then $T_n^{1/2}\sigma_n^{\frac{\alpha-1}{\alpha}}T_n^{1/2}=0$. Since $\rho_n\in \L^{\alpha}(\hil_n,\sigma_n)$, this implies 
$T_n^{1/2}\rho_n T_n^{1/2}=0$, 
according to Lemma \ref{lemma:rho alpha}, and therefore
\eqref{sc lower proof3} holds trivially, with both sides equal to $+\infty$. 

Taking the liminf in \eqref{sc lower proof3} yields
\begin{align*}
\liminf_{n\to+\infty}-\frac{1}{n}\log\Tr\,(T_n^{1/2}\rho_n T_n^{1/2})
&\ge
\frac{\alpha-1}{\alpha}r-\frac{1}{\alpha}\psi\nw(\vec{\rho}\|\vec{\sigma}|\alpha)
=
ur-\tilde\psi\nw(\vec{\rho}\|\vec{\sigma}|u).
\end{align*}
Since this holds for every test sequence as above, we get 
$ur-\tilde\psi\nw(\vec{\rho}\|\vec{\sigma}|u)\le \scli(\rho\|\sigma)$,
as required.
\end{proof}
\medskip

For the rest, we restrict our attention to the i.i.d.~case, where 
\begin{align*}
\hil_n=\hil^{\otimes n},\ds\ds
\rho_n=\rho^{\otimes n},\ds\ds
\sigma_n=\sigma^{\otimes n},\ds\ds n\in\bN,
\end{align*}
for some Hilbert space $\hil$ and $\rho,\sigma\in\B(\hil)\p$. 
Note that by Lemma \ref{lemma:tensor product},
$\tilde\psi\nw\bz(\rho^{\otimes n})_{n\in\bN}\|(\sigma^{\otimes n})_{n\in\bN}|u\jz
=
\tilde\psi\nw(\rho\|\sigma|u)$, $u\in(0,1)$, $n\in\bN$, 
and the same identity is straightforward to verify for $u=0,1$, 
whence
\begin{align*}
H_r\nw\bz(\rho^{\otimes n})_{n\in\bN}\|(\sigma^{\otimes n})_{n\in\bN}\jz
=
H_r\nw(\rho\|\sigma),\ds\ds\ds r\in\bR.
\end{align*}

We replace the notations $\vec{\rho}$ and $\vec{\sigma}$ with 
$\rho$ and $\sigma$, respectively, in the strong converse exponents introduced above.
Let
\begin{align*}
\scli(\rho\|\sigma)\f,\ds\ds\ds \scls(\rho\|\sigma)\f,
\ds\ds\text{and}\ds\ds
\scl(\rho\|\sigma)\f
\end{align*}
be defined the same way as $\scli(\rho\|\sigma)$, $\scls(\rho\|\sigma)$, and 
$\scl(\rho\|\sigma)$, respectively, 
but with the restrictions that only finite-rank tests are used. Obviously,
\begin{align*}
\scli(\rho\|\sigma)\le \scli(\rho\|\sigma)\f,\ds\ds\ds
\scls(\rho\|\sigma)\le \scls(\rho\|\sigma)\f,\ds\ds\ds
\scl(\rho\|\sigma)\le \scl(\rho\|\sigma)\f.
\end{align*}

The following lower bound follows by a straightforward adaptation of Nagaoka's method \cite{N}.

\begin{prop}\label{prop:sc lower f}
Let $\rho,\sigma\in\B(\hil)\p$. For every $r\in\bR$,
\begin{align*}
\hat H_r\nw(\rho\|\sigma)\ofa\le\scli_r(\rho\|\sigma)\f.
\end{align*}
\end{prop}
\begin{proof}
Let us fix an $r\in\bR$. 
We need to prove that for every $u\in[0,1]$, 
\begin{align}\label{sc r fa lower}
\scli_r(\rho\|\sigma)\f\ge ur-\tilde\psi\nw(\rho\|\sigma|u)\ofa\s .
\end{align}
The cases $u=0$ and $u=1$ can be proved exactly the same way as in the proof of 
Proposition \ref{lemma:sc optimality} above. 
For the rest, let us fix an $u\in(0,1)$, with corresponding $\alpha=1/(1-u)>1$.
If $\tilde\psi\nw(\rho\|\sigma|u)\ofa=+\infty$ then \eqref{sc r fa lower} holds trivially, and hence for the rest we assume that $\tilde\psi\nw(\rho\|\sigma|u)\ofa<+\infty$.
In particular, we have $\rho^0\le\sigma^0$, according to Lemma \ref{lemma:Q fa infty}.

Let $T_n\in\B(\hil^{\otimes n})_{[0,I]}$, $n\in\bN$, be a sequence of finite-rank tests such that 
\begin{align*}
\liminf_{n\to+\infty}-\frac{1}{n}\log \Tr T_n^{1/2}\sigma^{\otimes n}T_n^{1/2}\ge r.
\end{align*}
Assume first that $T_n^{1/2}\rho^{\otimes n}T_n^{1/2}\ne 0$, whence, by the assumption that 
$\rho^0\le\sigma^0$, we also have $T_n^{1/2}\sigma^{\otimes n}T_n^{1/2}\ne 0$.
By Lemma \ref{lemma:fa equiv}, 
\begin{align*}
Q_{\alpha}\nw(\rho^{\otimes n}\|\sigma^{\otimes n})\fa
\ge
Q_{\alpha}\nw(T_n^{1/2}\rho^{\otimes n}T_n^{1/2}\|T_n^{1/2}\sigma^{\otimes n}T_n^{1/2})
\ge
\bz\Tr T_n^{1/2}\rho^{\otimes n}T_n^{1/2}\jz^{\alpha}
\bz\Tr T_n^{1/2}\sigma^{\otimes n}T_n^{1/2}\jz^{1-\alpha},
\end{align*}
where the second inequality follows from Corollary \ref{cor:trace mon}.
A simple rearrangement yields
\begin{align*}
-\frac{1}{n}\log \Tr T_n^{1/2}\rho^{\otimes n}T_n^{1/2}
&\ge
\frac{\alpha-1}{\alpha}\bz-\frac{1}{n}\log \Tr T_n^{1/2}\sigma^{\otimes n}T_n^{1/2}\jz
-\frac{1}{\alpha}\frac{1}{n}\psi\nw(\rho^{\otimes n}\|\sigma^{\otimes n}|\alpha)\fa\\
&=
u\bz-\frac{1}{n}\log \Tr T_n^{1/2}\sigma^{\otimes n}T_n^{1/2}\jz
-\underbrace{\frac{1}{n}\tilde\psi\nw(\rho^{\otimes n}\|\sigma^{\otimes n}|u)\fa}_{
\le \tilde\psi\nw(\rho\|\sigma|u)\ofa}\\
&\ge
u\bz-\frac{1}{n}\log \Tr T_n^{1/2}\sigma^{\otimes n}T_n^{1/2}\jz
-\tilde\psi\nw(\rho\|\sigma|u)\ofa.
\end{align*}
These inequalities also hold (trivially, with the leftmost expression being $+\infty$) when 
$T_n^{1/2}\rho^{\otimes n}T_n^{1/2}=0$. Thus, we get 
\begin{align*}
\liminf_{n\to+\infty}-\frac{1}{n}\log \Tr T_n^{1/2}\rho^{\otimes n}T_n^{1/2}
&\ge
ur-\tilde\psi\nw(\rho\|\sigma|u)\ofa.
\end{align*}
Since this holds for every test sequences as above, \eqref{sc r fa lower} follows.
\end{proof}

\begin{lemma}\label{lemma:sc finiterank}
For finite-rank PSD operators $\rho,\sigma$ on a Hilbert space,
with $0\ne \rho^0\le\sigma^0$, we have
\begin{align}\label{sc upper}
\scl_r(\rho\|\sigma)\le H_r\nw(\rho\|\sigma),\ds\ds\ds r\in\bR.
\end{align}
\end{lemma}
\begin{proof}
The inequality in \eqref{sc upper} was proved in \cite[Theorem 4.10]{MO} for finite-rank density operators, under the implicit assumption that 
$D(\rho\|\sigma)\ne \Dmax(\rho\|\sigma)$, and it was proved in \cite{HM_sc_opalg} in the case 
$D(\rho\|\sigma)= \Dmax(\rho\|\sigma)$.
The case of general PSD operators follows easily by replacing $\rho$ and $\sigma$
with $\rho/\Tr\rho$ and $\sigma/\Tr\sigma$, respectively, and using the scaling laws
\eqref{Hanti scaling} and 
$\scls_r(\lambda\rho\|\eta\sigma)=\scls_{r+\log\eta}(\rho\|\sigma)-\log\lambda$.
\end{proof}

\begin{prop}\label{prop:sc achievability}
Let $\rho,\sigma\in\B(\hil)\p$ be such that 
$\rho^0\le\sigma^0$. For every $r\in\bR$,
\begin{align*}
\begin{array}{lllllllll}
\hat H_r\nw(\rho\|\sigma) &\le& 
\scli_r(\rho\|\sigma) &\le& \scls_r(\rho\|\sigma) &\le& \scl_r(\rho\|\sigma) 
\\
 & &   \ds\vertleq &  & \ds\vertleq&  & \ds\vertleq \\
\hat H_r\nw(\rho\|\sigma)\ofa &\le & \scli_r(\rho\|\sigma)\f &\le & \scls_r(\rho\|\sigma)\f 
&\le & \scl_r(\rho\|\sigma)\f &\le & \hat H_r\nw(\rho\|\sigma)\fa.
\end{array}
\end{align*}
\end{prop}
\begin{proof}
By propositions \ref{lemma:sc optimality} and \ref{prop:sc lower f}, 
we only need to prove $\scl_r(\rho\|\sigma)\f \le \hat H_r\nw(\rho\|\sigma)\fa$.
Let $P\in\bP_f(\hil)_{\rho,\sigma}^+$.
According to 
Lemma \ref{lemma:sc finiterank}, there exists a sequence of tests 
$(S_{P,n})_{n\in\bN}$ such that $S_{P,n}\le P^{\otimes n}$, and
\begin{align}
\liminf_{n\to +\infty}-\frac{1}{n}\log\Tr(P\sigma P)^{\otimes n}S_{P,n}&\ge r,
\label{Hoeffding ac proof1}\\
\lim_{n\to +\infty}-\frac{1}{n}\log\Tr(P\rho P)^{\otimes n}S_{P,n}
&\le
H_r\nw(P\rho P\|P\sigma P)=
\max_{u\in[0,1]}\left\{ur-\tilde\psi\nw(P\rho P\|P\sigma P|u)\right\},
\label{Hoeffding ac proof2}
\end{align}
where the equality is due to Propositions \ref{prop:Hanti equality} and \ref{cor:fa equality}.
Note that 
\begin{align*}
\Tr(P\sigma P)^{\otimes n}S_{P,n}
=
\Tr\sigma^{\otimes n} 
\underbrace{\bz P^{\otimes n}S_{P,n}P^{\otimes n}\jz}_{=S_{P,n}},\ds\ds\ds
\Tr(P\rho P)^{\otimes n}S_{P,n}
=
\Tr\rho^{\otimes n} 
\underbrace{\bz P^{\otimes n}S_{P,n}P^{\otimes n}\jz}_{=S_{P,n}},
\end{align*}
and therefore \eqref{Hoeffding ac proof1}--\eqref{Hoeffding ac proof2} yield
\begin{align*}
\scl_r(\rho\|\sigma)\f\le
\max_{u\in[0,1]}\left\{ur-\tilde\psi\nw(P\rho P\|P\sigma P|u)\right\}.
\end{align*}
Thus,
\begin{align}\label{sc proof1}
\scl_r(\rho\|\sigma)\f&\le
\inf_{\bP_f(\hil)_{\rho,\sigma}^+}\max_{u\in[0,1]}\left\{ur-\tilde\psi\nw(P\rho P\|P\sigma P|u)\right\}.
\end{align}
By Lemma \ref{lemma:restricted psi conv}, 
$u\mapsto ur-\tilde\psi\nw(P\rho P\|P\sigma P|u)$ is continuous 
on the compact set $[0,1]$ for every $P\in\bP_f(\hil)_{\rho,\sigma}^+$.
On the other hand, 
$\bP_f(\hil)_{\rho,\sigma}^+$ is an upward directed partially ordered set with respect to the PSD order, and 
for any $u\in[0,1]$, 
$P\mapsto ur-\tilde\psi\nw(P\rho P\|P\sigma P|u)$ is monotone decreasing 
on $\bP_f(\hil)$, again by Lemma \ref{lemma:restricted psi conv}.
Hence, by Lemma \ref{lemma:minimax2}, we may exchange the 
inf and the max in \eqref{sc proof1}. Thus, we get the upper bound
\begin{align*}
\scl_r(\rho\|\sigma)\f&\le
\max_{u\in[0,1]}\inf_{\bP_f(\hil)_{\rho,\sigma}^+}\left\{ur-\tilde\psi\nw(P\rho P\|P\sigma P|u)\right\}\\
&=
\max_{u\in[0,1]}\left\{ur-\sup_{P_f\in\bP(\hil)}\tilde\psi\nw(P\rho P\|P\sigma P|u)\right\}\\
&=
\max_{u\in[0,1]}\left\{ur-\tilde\psi\nw(\rho\|\sigma |u)\fa\right\}
=
\hat H_r\nw(\rho\|\sigma),
\end{align*}
as required.
\end{proof}

\begin{thm}\label{thm:sc}
Let $\rho,\sigma\in\B(\hil)\p$ be such that 
$D_{\alpha}\nw(\rho\|\sigma)<+\infty$ for some $\alpha\in(1,+\infty)$, and 
$Q_{\alpha}\nw(\rho\|\sigma)\fa=Q_{\alpha}\nw(\rho\|\sigma)$, $\alpha>1$.
Then 
\begin{align*}
\scli_r(\rho\|\sigma)&=\scls_r(\rho\|\sigma)=\scl_r(\rho\|\sigma)=
\scli_r(\rho\|\sigma)\f=\scls_r(\rho\|\sigma)\f=\scl_r(\rho\|\sigma)\f\\
&=
H_r\nw(\rho\|\sigma)=H_r\nw(\rho\|\sigma)\ofa=H_r\nw(\rho\|\sigma)\fa
=\hat H_r\nw(\rho\|\sigma)=\hat H_r\nw(\rho\|\sigma)\fa=\hat H_r\nw(\rho\|\sigma)\ofa,
\ds\ds\ds\ds\ds\ds\ds\ds r\in\bR.
\end{align*}
On the other hand, if $\rho,\sigma$ are trace-class and $D_{\alpha}\nw(\rho\|\sigma)=+\infty$
for all $\alpha\in(1,+\infty)$, then 
\begin{align*}
\scli_r(\rho\|\sigma)&=\scls_r(\rho\|\sigma)=\scl_r(\rho\|\sigma)=
\scli_r(\rho\|\sigma)\f=\scls_r(\rho\|\sigma)\f=\scl_r(\rho\|\sigma)\f\\
&=
-\log\Tr\rho
=\hat H_r\nw(\rho\|\sigma)=\hat H_r\nw(\rho\|\sigma)\ofa=\hat H_r\nw(\rho\|\sigma)\fa\\
&>-\infty=H_r\nw(\rho\|\sigma)=H_r\nw(\rho\|\sigma)\ofa=H_r\nw(\rho\|\sigma)\fa,
\ds\ds\ds\ds\ds\ds\ds\ds\ds\ds\ds\ds\ds\ds\ds\ds\ds\ds r\in\bR.
\end{align*}
\end{thm}
\begin{proof}
Immediate from Propositions
\ref{prop:sc achievability}, \ref{prop:Hanti equality}, and 
Lemma \ref{lemma:Hanti bounds}.
\end{proof}

As a special case of Theorem \ref{thm:sc}, we get the exact characterization of the strong converse exponent of discriminating quantum states on a separable Hilbert space, as follows:

\begin{cor}\label{cor:state sc}
Let $\rho,\sigma\in\B(\hil)\p$ be density operators. For every $r\in\bR$,  
\begin{align}\label{state sc exponent}
\scli_r(\rho\|\sigma)&=\scls_r(\rho\|\sigma)=\scl_r(\rho\|\sigma)=
\scli_r(\rho\|\sigma)\f=\scls_r(\rho\|\sigma)\f=\scl_r(\rho\|\sigma)\f
=
\hat H_r\nw(\rho\|\sigma)\ge 0,
\end{align}
and 
\begin{align}\label{Hr pos}
\hat H_r\nw(\rho\|\sigma)
>0\ds\iff\ds\exists\,\alpha>1:\,D_{\alpha}\nw(\rho\|\sigma)<+\infty\ds\text{and}\ds
r>D(\rho\|\sigma).
\end{align}
\end{cor}
\begin{proof}
The equalities in \eqref{state sc exponent} are immediate from Theorem \ref{thm:sc}, and the characterization of positivity in \eqref{Hr pos} follows from Lemma \ref{lemma:Hanti bounds}.
\end{proof}

\begin{rem}\label{rem:Stein}
Let $\rho$ and $\sigma$ be density operators. 
According to the direct part of the quantum Stein's lemma \cite{HP,JOPS}, for every $r<D(\rho\|\sigma)$ there exists a test sequence $T_n\in\B(\hil^{\otimes n})_{[0,1]}$, $n\in\bN$, such that 
\begin{align}\label{Stein}
\lim_{n\to+\infty}\Tr\rho^{\otimes n}(I-T_n)=0,\ds\ds\text{and}\ds\ds
\liminf_{n\to+\infty}-\frac{1}{n}\log\Tr\sigma^{\otimes n}T_n\ge r.
\end{align}
It was shown in \cite{N,ON} that in the finite-dimensional case, for any test sequence
$T_n\in\B(\hil^{\otimes n})_{[0,1]}$, $n\in\bN$,
\begin{align}\label{Stein sc}
r:=\liminf_{n\to+\infty}-\frac{1}{n}\log\Tr\sigma^{\otimes n}T_n>D(\rho\|\sigma)
\ds\imp\ds
\lim_{n\to+\infty}\Tr\rho^{\otimes n}(I-T_n)=1.
\end{align}
That is, if the type II error decreases with an exponent larger than the relative entropy
then the type I error goes to $1$; this is called the \ki{strong converse} to Stein's lemma.
The optimal (lowest) speed of convergence to $1$ is exponential, with the exponent being equal to the 
Hoeffding anti-divergence $H_r\nw(\rho\|\sigma)$, according to \cite{MO}. 
Corollary \ref{cor:state sc} generalizes this to the infinite-dimensional case, with one important difference. While in the finite-dimensional case finiteness of the relative entropy 
implies strict positivity of $H_r\nw(\rho\|\sigma)$ for every $r>D(\rho\|\sigma)$, and hence 
the strong converse property, in the infinite-dimensional case it might happen that 
$D(\rho\|\sigma)<+\infty$, yet $H_r\nw(\rho\|\sigma)=0$ for every $r\in\bR$, and hence 
the type I error sequence does not converge to $1$ \ki{with an exponential speed}
along a test sequence $(T_n)_{n\in\bN}$, even if 
$\liminf_{n\to+\infty}-\frac{1}{n}\log\Tr\sigma^{\otimes n}T_n>D(\rho\|\sigma)$.
According to Corollary \ref{cor:state sc}, this happens if and only if 
$D_{\alpha}\nw(\rho\|\sigma)=+\infty$ for every $\alpha>1$. 
It is an open question what kind of behaviour can occur in this case; 
if the type II exponent is above the relative entropy,
do the type I error 
probabilities still go to $1$ (strong converse property) but with a sub-exponential speed, 
or may it happen that the strong converse property does not hold, 
i.e., \eqref{Stein sc} is not satisfied?
Note that the monotonicity of the relative entropy under measurements implies that 
\begin{align*}
\lim_{n\to+\infty}\Tr\rho^{\otimes n}(I-T_n)=0\ds\imp\ds
\limsup_{n\to+\infty}-\frac{1}{n}\log\Tr\sigma^{\otimes n}T_n\le D(\rho\|\sigma);
\end{align*}
see, e.g., the proof of (2.4) in \cite{HP}, or \cite[Proposition 5.2]{HMO2}.
In particular, \eqref{Stein} cannot hold with $r>D(\rho\|\sigma)$ for any test sequence.
\end{rem}

\subsection{Generalized cutoff rates}
\label{sec:cutoff}

Corollary \ref{cor:state sc} gives an operational interpretation to the Hoeffding anti-divergences, but 
not directly to the sandwiched R\'enyi divergences. To get such an operational interpretation, one can consider the following quantity, introduced originally in \cite{Csiszar} for the 
finite-dimensional classical case:
\begin{defin}
Let $\rho,\sigma\in\B(\hil)\p$ and $\kappa\in(0,1)$. The \ki{generalized $\kappa$-cutoff rate} 
$C_{\kappa}(\rho\|\sigma)$ is defined to be the infimum of all $r_0\in\bR$ such that 
$\scli_r(\rho\|\sigma)\ge \kappa(r-r_0)$ holds for every $r\in\bR$. Analogously, 
$C_{\kappa}(\rho\|\sigma)\fa$ is defined to be the infimum of all $r_0\in\bR$ such that 
$\scli_r(\rho\|\sigma)\f\ge \kappa(r-r_0)$ holds for every $r\in\bR$.
\end{defin}

\begin{prop}\label{prop:cutoff}
Let $\rho,\sigma\in\B(\hil)\p$.
\smallskip

\noindent (i)
For any $\kappa\in(0,1)$,
\begin{align}\label{cutoff upper}
C_{\kappa}(\rho\|\sigma)\fa\le C_{\kappa}(\rho\|\sigma)\le D_{\frac{1}{1-\kappa}}\nw(\rho\|\sigma).
\end{align}

\noindent (ii)
If $\kappa$ is such that there exist $0<\kappa_1<\kappa<\kappa_2<1$ for which  
$D_{\frac{1}{1-\kappa_j}}\nw(\rho\|\sigma)\fa<+\infty$, $j=1,2$, then
\begin{align}\label{cutoff lower}
D_{\frac{1}{1-\kappa}}\nw(\rho\|\sigma)\fa
\le
C_{\kappa}(\rho\|\sigma)\fa\le C_{\kappa}(\rho\|\sigma)\le D_{\frac{1}{1-\kappa}}\nw(\rho\|\sigma).
\end{align}
If, moreover, 
$D_{\frac{1}{1-\kappa}}\nw(\rho\|\sigma)\fa=D_{\frac{1}{1-\kappa}}\nw(\rho\|\sigma)$,
then all the inequalities in \eqref{cutoff lower} hold as equalities.
\end{prop}
\begin{proof}
(i) The first inequality in \eqref{cutoff upper} is trivial by definition. 
If $D_{\frac{1}{1-\kappa}}\nw(\rho\|\sigma)=+\infty$ then the second inequality in
\eqref{cutoff upper} holds trivially, and hence we assume the contrary. 
By Proposition \ref{lemma:sc optimality},
\begin{align*}
\scli_r(\rho\|\sigma)
\ge
\hat H_r\nw(\rho\|\sigma)
=
\sup_{u\in[0,1]}\{u r-\tilde\psi\nw(\rho\|\sigma|u)\}
\ge
\kappa r-\tilde\psi\nw(\rho\|\sigma|\kappa)
=
\kappa\Big( r-\underbrace{\frac{1}{\kappa}\tilde\psi\nw(\rho\|\sigma|\kappa)}_{=D_{\frac{1}{1-\kappa}}\nw(\rho\|\sigma)}\Big),
\end{align*}
from which the second inequality in \eqref{cutoff upper} follows by definition.

(ii)
By the assumptions,
$\tilde\psi\nw(\rho\|\sigma|\kappa_j)\fa<+\infty$, $j=1,2$, and hence
$-\infty<\derleft\tilde\psi\nw(\rho\|\sigma|\kappa)\fa
\le
\derright\tilde\psi\nw(\rho\|\sigma|\kappa)\fa<+\infty$, 
due to the convexity of $\tilde\psi\nw(\rho\|\sigma|\valt)\fa$, established in 
Lemma \ref{cor:psi fa}. Moreover, $\rho^0\le\sigma^0$, according to Lemma \ref{lemma:Q fa infty}.
For any $r\in [\derleft\tilde\psi\nw(\rho\|\sigma|\kappa)\fa,\derright\tilde\psi\nw(\rho\|\sigma|\kappa)\fa]$, 
\begin{align*}
\scl_r(\rho\|\sigma)\f&\le
\hat H_r\nw(\rho\|\sigma)\fa=\max_{u\in[0,1]}\{u r-\tilde\psi\nw(\rho\|\sigma|u)\fa\}
=
\kappa r- \tilde\psi\nw(\rho\|\sigma|\kappa)\fa
=
\kappa\Big( r-\underbrace{\frac{1}{\kappa}\tilde\psi\nw(\rho\|\sigma|\kappa)\fa}_{=D_{\frac{1}{1-\kappa}}\nw(\rho\|\sigma)\fa}\Big)\,,
\end{align*}
where the first inequality is due to Proposition \ref{prop:sc achievability}.
This yields the first inequality in \eqref{cutoff lower}, and the rest have already been proved in the previous point.
\end{proof}

Proposition \ref{prop:cutoff} and Corollary \ref{cor:fa equality} yield immediately the following:

\begin{thm}
Let $\rho,\sigma\in\B(\hil)\p$, be such that
$\rho$ and $\sigma$ are trace-class, or $\sigma$ is compact and $\rho\in\B^{\infty}(\hil,\sigma)$. Let $\kappa\in(0,1)$, and assume that $D_{\alpha}\nw(\rho\|\sigma)<+\infty$ for $\alpha$ in a neighborhood of $\alpha_0:=1/(1-\kappa)$. Then 
\begin{align*}
C_{\kappa}(\rho\|\sigma)\fa
=
C_{\kappa}(\rho\|\sigma)= D_{\frac{1}{1-\kappa}}\nw(\rho\|\sigma),
\ds\ds\ds\ds\text{or equivalently,}\ds\ds\ds\ds
D_{\alpha_0}\nw(\rho\|\sigma)=C_{\frac{\alpha_0-1}{\alpha_0}}(\rho\|\sigma)=C_{\frac{\alpha_0-1}{\alpha_0}}(\rho\|\sigma)\fa.
\end{align*}
\end{thm}

\subsection{Monotonicity of the R\'enyi divergences}
\label{sec:mon}

The operational representation of the Hoeffding anti-divergences in Section 
\ref{sec:sc Hanti} can be used to obtain the monotonicity of the sandwiched 
R\'enyi divergences under quantum operations.

In the Heisenberg picture, a quantum operation from a system with Hilbert space $\hil$ to a
system with Hilbert space $\kil$ is given by a unital normal completely positive map
$\map:\,\B(\kil)\to\B(\hil)$, which can be written as
\begin{align}\label{Kraus}
\map:\,\B(\kil)\ni A\mapsto V^*(A\otimes I_E)V=\sum_{i\in\I} V_i^*AV_i, 
\end{align}
where $V:\,\hil\to\kil\otimes\hil_E$ is an isometry,
$V_i:=(I_{\kil}\otimes\bra{e_i})V$ for some ONB $(e_i)_{i\in\I}$ in the auxiliary Hilbert space $\hil_E$, and the sum in \eqref{Kraus} converges in the strong operator topology \cite{HellwigKraus,Stinespring}. 
As in everywhere in the paper, we assume that $\hil,\kil$ are separable, in which case
%When $\hil$ is separable, 
the auxiliary Hilbert space $\hil_E$ can be chosen to be separable, and the index set $\I$ in \eqref{Kraus} countable. 

In the Schr\"odinger picture, a density operator $\rho\in\S(\hil)$ is transformed by the dual map
\begin{align*}
\map^*(\rho):=\sum_{i\in\I}V_i\rho V_i^*
=
\sum_{i\in\I}(I_{\kil}\otimes\bra{e_i})V\rho V^*(I_{\kil}\otimes\ket{e_i})
=
\Tr_EV\rho V^*,
\end{align*}
where the sum converges in trace-norm, and the result is a density operator on $\kil$. 
If $\rho$ is PSD but not trace-class then the above sum need not converge (in the weak, equivalently, in the strong operator topology), but it may,
in which case we say that $\map^*$ is defined on $\rho$, and define
$\map^*(\rho):=\sum_{i\in\I}V_i\rho V_i^*$.
A trivial case where $\map^*$ is defined on every $\rho\in\B(\hil)\p$
is when $\map$ has only finitely many operators in its Kraus decomposition, or equivalently, $\hil_E$ is finite-dimensional.

\begin{lemma}\label{lemma:transpose map}
Let $\rho\in\B(\hil)\p$ and $\map:\,\B(\kil)\to\B(\hil)$ be a unital normal completely positive map. 
If $\map^*$ is defined on $\rho$ then 
\begin{align*}
\Tr A^{1/2}\map^*(\rho)A^{1/2}=\Tr\map(A)^{1/2}\rho\map(A)^{1/2},\ds\ds\ds A\in\B(\kil)\p.
\end{align*}
\end{lemma}
\begin{proof}
Let $(f_j)_{j\in\J}$ be an orthonormal basis in $\kil$. Then 
\begin{align*}
\Tr A^{1/2}\map^*(\rho)A^{1/2}&=
\sum_{j\in\J}\underbrace{\inner{A^{1/2}f_j}{\map^*(\rho)A^{1/2}f_j}}_{
=\sum_{i\in\I}\inner{A^{1/2}f_j}{V_i\rho V_i^*A^{1/2}f_j}}
=
\sum_{i\in\I}\underbrace{\sum_{j\in\J}\inner{A^{1/2}f_j}{V_i\rho V_i^*A^{1/2}f_j}}_{=
\Tr A^{1/2}V_i\rho V_i^*A^{1/2}=\Tr\rho^{1/2}V_i^*AV_i\rho^{1/2}}\\
&=
\sum_{i\in\I}\Tr\rho^{1/2}V_i^*AV_i\rho^{1/2}
=
\sum_{i\in\I}\sum_{j\in\J}\inner{\rho^{1/2}f_j}{V_i^*AV_i\rho^{1/2}f_j}\\
&=
\sum_{j\in\J}\underbrace{\sum_{i\in\I}\inner{\rho^{1/2}f_j}{V_i^*AV_i\rho^{1/2}f_j}}_{
=\inner{\rho^{1/2}f_j}{\map(A)\rho^{1/2}f_j}}
=
\Tr \rho^{1/2}\map(A)\rho^{1/2}=\Tr\map(A)^{1/2}\rho\map(A)^{1/2}.
\end{align*}
\end{proof}

The transformation on multiple systems is given by 
\begin{align}\label{Kraus2}
\map^{\otimes n}:\,\B(\kil^{\otimes n})\ni A\mapsto (V^{\otimes n})^*(A\otimes I_{E^n})V^{\otimes n}=\sum_{\vecc{i}\in\I^n} (V_{i_1}\ootimes V_{i_n})^*A
(V_{i_1}\ootimes V_{i_n}). 
\end{align}
If $\map^*$ is defined on $\rho\in\B(\hil)\p$ then $(\map^{\otimes n})^*$ is defined on $\rho^{\otimes n}$, and 
$(\map^{\otimes n})^*(\rho^{\otimes n})=(\map^*(\rho))^{\otimes n}$.

In the context of operator discrimination, a transformation $\map$ effectively reduces the available tests for discriminating
$\rho$ and $\sigma$, thereby increasing the strong converse exponent, as expressed by the following:
\begin{lemma}\label{lemma:sc mon}
Let $\rho,\sigma\in\B(\hil)\p$ and $\map:\,\B(\kil)\to\B(\hil)$ be a unital normal completely positive map such that $\map^*$ is defined on $\rho$ and $\sigma$. Then 
\begin{align*}
\scli_r(\map^*(\rho)\|\map^*(\sigma))\ge
\scli_r(\rho\|\sigma),\ds\ds\ds
\scls_r(\map^*(\rho)\|\map^*(\sigma))\ge
\scls_r(\rho\|\sigma),\ds\ds\ds
r\in\bR.
\end{align*}
\end{lemma} 
\begin{proof}
We only prove the assertion for $\scls_r$, as the proof for $\scli_r$ goes the same way.
We have
\begin{align*}
\scls_r(\map^*(\rho)\|\map^*(\sigma))&=
\inf\left\{\liminf_{n\to+\infty}-\frac{1}{n}\log
\Tr T_n^{1/2}(\map^*(\rho))^{\otimes n}T_n^{1/2}:\,
\liminf_{n\to+\infty}-\frac{1}{n}\log
\Tr T_n^{1/2}(\map^*(\sigma))^{\otimes n}T_n^{1/2}\ge r
\right\}\\
&=\inf\left\{\liminf_{n\to+\infty}-\frac{1}{n}\log\Tr (\map^{\otimes n}(T_n))^{1/2}\rho^{\otimes n}(\map^{\otimes n}(T_n))^{1/2}:\,\right.\\
&\ds\ds\ds\ds\ds\left.\liminf_{n\to+\infty}-\frac{1}{n}\log\Tr (\map^{\otimes n}(T_n))^{1/2}\sigma^{\otimes n}(\map^{\otimes n}(T_n))^{1/2}\ge r
\right\}\\
&\ge
\inf\left\{\liminf_{n\to+\infty}-\frac{1}{n}\log
\Tr S_n^{1/2}\rho^{\otimes n}S_n^{1/2}:\,
\liminf_{n\to+\infty}-\frac{1}{n}\log
\Tr S_n^{1/2}\sigma^{\otimes n}S_n^{1/2}\ge r
\right\}\\
&=
\scls_r(\rho\|\sigma),
\end{align*}
where the first two infima are taken over sequences of tests
$T_n\in\B(\kil^{\otimes n})_{[0,I]}$, $n\in\bN$, satisfying the given conditions, 
the third infimum is taken over sequences of tests
$S_n\in\B(\hil^{\otimes n})_{[0,I]}$, $n\in\bN$, satisfying the given condition,
the first equality is by definition, the second equality follows from 
Lemma \ref{lemma:transpose map}, the inequality is obvious from the fact that 
$T_n\in\B(\kil^{\otimes n})_{[0,I]}\imp\map^{\otimes n}(T_n)\in\B(\hil^{\otimes n})_{[0,I]}$, 
and the last equality is again by definition.
\end{proof}

The proof of the following monotonicity result is similar to the proof of the analogous 
result given in \cite[Remark 2]{Nagaoka} for the monotonicity of the Petz-type
R\'enyi divergences and finite-dimensional density operators. 
The main ideas in the proof are using the bounds on the strong converse 
exponents given in Proposition \ref{prop:sc achievability}, the monotonicity of the 
strong converse exponents given in Lemma \ref{lemma:sc mon},
and the fact that the R\'enyi divergences can be expressed from the Hoeffding 
anti-divergences by Legendre-Fenchel transformation, i.e., Lemma \ref{lemma:Hanti bipolar}.

\begin{thm}\label{thm:mon}
Let $\rho,\sigma\in\B(\hil)\p$, be such that 
\begin{align}\label{mon cond1}
D_{\alpha}\nw(\rho\|\sigma)\fa=
D_{\alpha}\nw(\rho\|\sigma),\ds\ds\ds \alpha>1,
\end{align}
and let $\map:\,\B(\kil)\to\B(\hil)$ be a unital normal completely positive linear that is defined on  $\rho$ and $\sigma$. 
Then
\begin{align}\label{sandwiched mon}
D_{\alpha}\nw(\map^*(\rho)\|\map^*(\sigma))\fa\le D_{\alpha}\nw(\rho\|\sigma),\ds\ds\ds
\alpha>1.
\end{align}
\end{thm}
\begin{proof}
By assumption, 
$\hat H_r\nw(\rho\|\sigma)\fa=\hat H_r\nw(\rho\|\sigma)$, $r\in\bR$, and 
\begin{align*}
\hat H_r\nw(\rho\|\sigma)\fa=\hat H_r\nw(\rho\|\sigma)\le\scls_r(\rho\|\sigma)\le
\scls_r(\map^*(\rho)\|\map^*(\sigma))
\le 
\hat H_r\nw(\map^*(\rho)\|\map^*(\sigma))\fa\,,\ds\ds\ds r\in\bR,
\end{align*}
where the first and the last inequalities follow from 
Proposition \ref{prop:sc achievability}, 
and the second inequality from Lemma \ref{lemma:sc mon}.
Hence, by Lemma \ref{lemma:Hanti bipolar}, 
\begin{align*}
\tilde\psi\nw(\rho\|\sigma|u)\fa
=
\sup_{r\in\bR}\left\{ur-\hat H_r\nw(\rho\|\sigma)\fa\right\}
\ge 
\sup_{r\in\bR}\left\{ur-\hat H_r\nw(\map^*(\rho)\|\map^*(\sigma))\fa\right\}
=
\tilde\psi\nw(\map^*(\rho)\|\map^*(\sigma)|u)\fa,\ds\ds u\in\bR,
\end{align*}
which is equivalent to \eqref{sandwiched mon}.
\end{proof}

\begin{cor}\label{cor:mon}
Let $\rho,\sigma\in\B(\hil)\p$, and let 
$\map:\,\B(\kil)\to\B(\hil)$ be a unital normal completely positive map.
Assume that
\smallskip

\noindent a) $\rho$ and $\sigma$ are both trace-class, 
\smallskip

\noindent or
\smallskip

\noindent b) 
$\map^*$ is defined on $\rho$ and $\sigma$, $\sigma$ and $\map^*(\sigma)$ are compact,
and $\rho\in\B^{\infty}(\hil,\sigma)$.
\smallskip

\noindent Then 
\begin{align}\label{sandwiched mon2}
D_{\alpha}\nw(\map^*(\rho)\|\map^*(\sigma))\le D_{\alpha}\nw(\rho\|\sigma),\ds\ds\ds
\alpha>1.
\end{align}
\end{cor}
\begin{proof}
If $\rho$ and $\sigma$ are both trace-class then $\map^*$ is automatically defined on them, 
and $\map^*(\rho)$ and $\map^*(\sigma)$ are both trace-class.
Note that $\rho\in\B^{\infty}(\hil,\sigma)$ $\iff$ $\rho\le \lambda\sigma$
for some $\lambda\ge 0$, whence
$\map^*(\rho)\le\lambda\map^*(\sigma)$, i.e., 
$\map^*(\rho)\in\B^{\infty}(\kil,\map^*(\sigma))$. 
By Remark \ref{rem:smaller alpha},
$\rho\in\B^{\alpha}(\hil,\sigma)$ and 
$\map^*(\rho)\in\B^{\alpha}(\kil,\map^*(\sigma))$, $\alpha>1$.
Thus, by Lemma \ref{cor:fa equality}, the assumptions guarantee that 
$D_{\alpha}\nw(\rho\|\sigma)\fa=D_{\alpha}\nw(\rho\|\sigma)$ and 
$D_{\alpha}\nw(\map^*(\rho)\|\map^*(\sigma))\fa=D_{\alpha}\nw(\map^*(\rho)\|\map^*(\sigma))$,
$\alpha>1$,
and therefore \eqref{sandwiched mon2} follows immediately from 
Theorem \ref{thm:mon}.
\end{proof}

\begin{rem}
Monotonicity of the form \eqref{sandwiched mon2} in the case where both $\rho$ and $\sigma$ are trace-class is a special case of
\cite[Theorem 14]{BST} and \cite[Theorem 3.14]{Jencova_NCLp},
where monotonicity was proved in the more general setting of normal positive linear functionals on a von Neumann algebra. Our proof above is completely different from the proofs given in 
\cite{BST} and \cite{Jencova_NCLp}.
\end{rem}

\section{Conclusion}

We have shown that for any $\alpha>1$, the sandwiched R\'enyi $\alpha$-divergence of infinite-dimensional density operators has the same operational interpretation in the context of state discrimination as 
in the finite-dimensional case, and also that it coincides with the regularized measured 
R\'enyi $\alpha$-divergence, again analogously to the finite-dimensional case. 
Our results can be extended to more general operator algebraic 
settings, as shown in \cite{HM_sc_opalg}.

It is worth noting that while in \cite{MO} the equality of the sandwiched R\'enyi divergence and 
the regularized measured R\'enyi divergence was an important ingredient of showing the equality of the strong converse exponent and the Hoeffding anti-divergence, the extensions to the 
infinite-dimensional case can be done separately, building in each problem only on the corresponding finite-dimensional result and the recoverability of the sandwiched R\'enyi divergences from finite-dimensional restrictions.  

We also considered the extension of the sandwiched R\'enyi divergences 
(and more generally, R\'enyi $(\alpha,z)$-divergences)
to pairs of 
not necessarily trace-class positive semi-definite operators, and established 
some properties of this extension.
Related to this, we considered a generalization of the state discrimination problem, where 
the hypotheses may be represented by general positive semi-definite operators. 
We gave bounds on the strong converse exponent in this problem, and showed that 
at least in some cases, the equality between the strong converse exponent and the Hoeffding anti-divergence still holds in this generalized setting. 

There are a number of interesting problems left open in the paper. 
Probably the most important is clarifying 
whether 
$Q_{\alpha}\nw(\rho\|\sigma)=Q_{\alpha}\nw(\rho\|\sigma)\fa$
holds for every pair of PSD operators $\rho,\sigma$ and every $\alpha>1$, 
and if not then whether there exist other examples for which it holds, apart from the ones
given in Proposition \ref{cor:fa equality} and Lemma \ref{lemma:Q fa infty}.
Such examples would extend the applicability of Proposition \ref{prop:sc achievability}
and Theorem \ref{thm:mon}, among others.
While less relevant for the problem of operator discrimination, the same question may be 
asked for more general R\'enyi $(\alpha,z)$-divergences, which 
seems interesting from the matrix analysis point of view.
Finally, from the point of view of quantum information theory, the most important 
question seems to be to clarify the optimal asymptotics of the type I 
error probability when the type II exponent is strictly above the
relative entropy of the two states (assuming that the latter is finite), while all their sandwiched R\'enyi $\alpha$-divergences 
are $+\infty$ for $\alpha>1$; see Remark \ref{rem:Stein}.

\section*{Acknowledgments}

This work was partially funded by the
National Research, Development and 
Innovation Office of Hungary via the research grants K124152 and KH129601, and
by the Ministry of Innovation and
Technology and the National Research, Development and Innovation
Office within the Quantum Information National Laboratory of Hungary. 
The author is grateful to P\'eter Vrana for discussions at the early stage of the project, 
to Ludovico Lami for comments on the strong converse property that led to 
Remark \ref{rem:Stein},
and to Fumio Hiai for numerous helpful comments.

\section*{Data availability}
Data sharing not applicable to this article as no datasets were generated or analyzed during the current study.

\appendix

\section{Some further properties of the Hoeffding anti-divergences}
\label{sec:further}

\begin{lemma}\label{lemma:superadd}
For any $\rho,\sigma\in\B(\hil)\p$, and for every $u\in\bR$,
\begin{align*}
&n\mapsto \tilde\psi\nw(\rho^{\otimes n} \|\sigma^{\otimes n} |u)\fa\ds\ds\ds\ds\ds\text{is superadditive},\\
&n\mapsto 
\frac{1}{2^n}\tilde\psi\nw(\rho^{\otimes 2^n} \|\sigma^{\otimes 2^n} |u)\fa
\ds\ds\text{is monotone increasing},
\end{align*}
and
\begin{align}
\tilde\psi\nw(\rho \|\sigma |u)\ofa
&=
\lim_{n\to+\infty}\frac{1}{n}\tilde\psi\nw(\rho^{\otimes n} \|\sigma^{\otimes n} |u)\fa
\label{ofa as limit}\\
&=
\lim_{n\to+\infty}\frac{1}{2^n}\tilde\psi\nw(\rho^{\otimes 2^n} \|\sigma^{\otimes 2^n} |u)\fa
=
\sup_{n\in\bN}\frac{1}{2^n}\tilde\psi\nw(\rho^{\otimes 2^n} \|\sigma^{\otimes 2^n} |u)\fa.
\nonumber
\end{align}
\end{lemma}
\begin{proof}
The assertions are trivial when $u\in\bR\setminus[0,1]$, because then all quantities above 
are equal to $+\infty$. For $u\in(0,1)$, 
superadditivity is obvious from restricting to projections of the form 
$P_1\otimes P_2$, $P_1\in\bP_f(\hil^{\otimes n})_{\rho^{\otimes n},\sigma^{\otimes n}}^+$, 
$P_2\in\bP_f(\hil^{\otimes m})_{\rho^{\otimes m},\sigma^{\otimes m}}^+$ in the definition of 
$\tilde\psi\nw(\rho^{\otimes (n+m)} \|\sigma^{\otimes (n+m)} |u)\fa$, 
according to Lemma \ref{lemma:tensor product},
and the cases 
$u\in\{0,1\}$ follow by taking limits (in fact, even additivity holds there, according to 
\eqref{psi tilde fa 0}--\eqref{psi tilde fa 1}). The rest of the assertions are straightforward 
consequences. 
\end{proof}

\begin{lemma}\label{lemma:psi ofa additive}
For any $\rho,\sigma\in\B(\hil)\p$, and every $n\in\bN$,
\begin{align*}
\tilde\psi\nw(\rho^{\otimes n} \|\sigma^{\otimes n} |u)
&=
n\tilde\psi\nw(\rho\|\sigma |u),\ds\ds\ds u\in(0,1),\\
\tilde\psi\nw(\rho^{\otimes n} \|\sigma^{\otimes n} |u)\ofa
&=
n\tilde\psi\nw(\rho\|\sigma |u)\ofa,\ds\ds u\in\bR.
\end{align*}
\end{lemma}
\begin{proof}
The first equality is immediate from the multiplicativity of $Q_{\alpha}\nw$, given in Lemma \ref{lemma:tensor product}, and the second 
one follows from \eqref{ofa as limit}, as 
\begin{align*}
\tilde\psi\nw(\rho^{\otimes n} \|\sigma^{\otimes n} |u)\ofa
&=
\lim_{k\to+\infty}\frac{1}{k}\tilde\psi\nw(\rho^{\otimes kn} \|\sigma^{\otimes kn} |u)\fa
=
n\underbrace{\lim_{k\to+\infty}\frac{1}{nk}\tilde\psi\nw(\rho^{\otimes kn} \|\sigma^{\otimes kn} |u)\fa}_{=\tilde\psi\nw(\rho \|\sigma |u)\ofa}.
\end{align*}
\end{proof}

\begin{lemma}\label{lemma:Hanti subadd}
For any $\rho,\sigma\in\B(\hil)\p$ and $r\in\bR$,
\begin{align}\label{Hanti subadd}
n\mapsto \bH_{nr}\nw(\rho^{\otimes n}\|\sigma^{\otimes n})\ds\ds\ds
\text{is subadditive},
\end{align}
where $\bH_r\nw$ stands for any of the
Hoeffding anti-divergences in Definition \ref{def:Hanti def}.
Moreover,
\begin{align}\label{Hanti add}
H_{nr}\nw(\rho^{\otimes n}\|\sigma^{\otimes n})=nH_r\nw(\rho\|\sigma),\ds\ds\ds\ds
\hat H_{nr}\nw(\rho^{\otimes n}\|\sigma^{\otimes n})\ofa=n\hat H_r\nw(\rho\|\sigma)\ofa,\ds\ds
n\in\bN.
\end{align}
\end{lemma}
\begin{proof}
\eqref{Hanti add} is immediate from 
Lemma \ref{lemma:psi ofa additive}, and it trivially implies subadditivity
for the quantities defined in 
\eqref{Hanti def01}, 
\eqref{Hanti def03}, 
\eqref{Hanti def2}, 
\eqref{Hanti def4}. 
By Lemma \ref{lemma:superadd},
$n\mapsto\tilde\psi\nw(\rho^{\otimes n}\|\sigma^{\otimes n}|u)\fa$
is superadditive for every $u\in\bR$, which implies the subadditivity of
$n\mapsto H_{nr}\nw(\rho^{\otimes n}\|\sigma^{\otimes n})\fa$ and 
$n\mapsto \hat H_{nr}\nw(\rho^{\otimes n}\|\sigma^{\otimes n})\fa$.
\end{proof}

\begin{prop}\label{lemma:Hanto ofa lim}
For any $\rho,\sigma\in\B(\hil)\p$ and $r\in\bR$,
\begin{align}
\hat H_r\nw(\rho\|\sigma)\ofa
&=
\inf_{n\in\bN}\frac{1}{n}\hat H_{nr}\nw(\rho^{\otimes n}\|\sigma^{\otimes n})\fa
=
\lim_{n\to+\infty}\frac{1}{n}\hat H_{nr}\nw(\rho^{\otimes n}\|\sigma^{\otimes n})\fa
\label{Hanti ofa as inf}\\
&=
\inf_{n\in\bN}\frac{1}{2^n}\hat H_{2^nr}\nw(\rho^{\otimes 2^n}\|\sigma^{\otimes 2^n})\fa.
\label{Hanti ofa as inf2}
\end{align}
\end{prop}
\begin{proof}
The second equality in \eqref{Hanti ofa as inf} follows from the subadditivity of 
$n\mapsto \hat H_{nr}\nw(\rho^{\otimes n}\|\sigma^{\otimes n})\fa$ given in Lemma \ref{lemma:Hanti subadd}. 
To see the other two equalities in \eqref{Hanti ofa as inf}--\eqref{Hanti ofa as inf2}, note that 
by definition,
\begin{align*}
\hat H_r\nw(\rho\|\sigma)\ofa
&=
\max_{u\in[0,1]}\inf_{n\in\bN}
\frac{1}{n}\left\{unr-\tilde\psi\nw(\rho^{\otimes n}\|\sigma^{\otimes n}|u)\fa\right\}\\
&\le
\inf_{n\in\bN}\max_{u\in[0,1]}
\frac{1}{n}\left\{unr-\tilde\psi\nw(\rho^{\otimes n}\|\sigma^{\otimes n}|u)\fa\right\}
=
\inf_{n\in\bN}\frac{1}{n}\hat H_{nr}\nw(\rho^{\otimes n}\|\sigma^{\otimes n}|u)\fa\\
&\le
\inf_{n\in\bN}\max_{u\in[0,1]}\frac{1}{2^n}\left\{ u2^n r- \tilde\psi\nw(\rho^{\otimes 2^n}\|\sigma^{\otimes 2^n}|u)\fa\right\}
=
\inf_{n\in\bN}\frac{1}{2^n}\hat H_{2^nr}\nw(\rho^{\otimes 2^n}\|\sigma^{\otimes 2^n}|u)\fa.
\end{align*}
On the other hand, by 
Corollary \ref{cor:psi fa} and Lemma \ref{lemma:superadd}, 
$\frac{1}{2^n}\left\{ u2^n r- \tilde\psi\nw(\rho^{\otimes 2^n}\|\sigma^{\otimes 2^n}|u)\fa\right\}$
is upper semi-continuous in $u$ on the compact set $[0,1]$, and monotone decreasing
in $n$, whence 
\begin{align*}
\inf_{n\in\bN}\frac{1}{2^n}\hat H_{2^nr}\nw(\rho^{\otimes 2^n}\|\sigma^{\otimes 2^n}|u)\fa
&=
\inf_{n\in\bN}\max_{u\in[0,1]}\frac{1}{2^n}\left\{ u2^n r- \tilde\psi\nw(\rho^{\otimes 2^n}\|\sigma^{\otimes 2^n}|u)\fa\right\}\\
&=
\max_{u\in[0,1]}\inf_{n\in\bN}\frac{1}{2^n}\left\{ u2^n r- \tilde\psi\nw(\rho^{\otimes 2^n}\|\sigma^{\otimes 2^n}|u)\fa\right\}\\
&=
\max_{u\in[0,1]}\{ur-\tilde\psi\nw(\rho\|\sigma|u)\ofa\}
=
\hat H_r\nw(\rho\|\sigma)\ofa,
\end{align*}
where the second equality is due to Lemma \ref{lemma:minimax2}, and the third equality 
follows from Lemma \ref{lemma:superadd}.
\end{proof}

\bibliography{bibliography_sc_vegtelendim}
\end{document}